%% file: ms.tex
\documentclass[sigconf, nonacm=true]{acmart}

\setcopyright{none}
\renewcommand\footnotetextcopyrightpermission[1]{} % removes footnote with conference information in first column
\usepackage{multirow}

\usepackage{booktabs}
\usepackage{hyperref}
\usepackage{enumitem}
\usepackage{amsthm,bm,bbm}
\usepackage{graphicx}
\usepackage{dsfont}
\usepackage{algorithm}
\usepackage[noend]{algpseudocode}
\usepackage{color}
\usepackage[skip=0pt]{caption}
\usepackage[short, nocomma]{optidef}

\newcommand{\EPA}{\sc CRA}
\newcommand{\VPA}{\sc GPA}
\newcommand{\DRA}{\sc DRA}

\DeclarePairedDelimiter\ceil{\lceil}{\rceil}

\DeclareMathOperator*{\argmin}{arg\,min}
\DeclareMathOperator*{\argmax}{arg\,max}

\newcommand{\be}{\begin{eqnarray}}
	
	\newcommand{\ee}{\end{eqnarray}}
\newcommand{\ben}{\begin{eqnarray*}}
	\newcommand{\een}{\end{eqnarray*}}
\newcommand{\bfl}{\begin{flalign*}}
	\newcommand{\efl}{\end{flalign*}}
\newcommand{\dref}[1]{(\ref{#1})}

\newcommand{\calJ}{{\mathcal J}}

\newcommand{\calK}{{\mathcal K}}

\newcommand{\globl}[1] { {#1} }
\newcommand{\globlt}[1] { {#1} }   % previously { \bar{#1} }

\newcommand{\schedk} { \mathbf{\hat k} }
\newcommand{\schedki} { \hat{k} }
\newcommand{\assgnk} { \mathbf{k} }
\newcommand{\greedk} { \mathbf{\bar k} }
\newcommand{\greedki} { {\bar{k}} }
\newcommand{\assgnki} { {k} }

\newcommand{\setConv} {\Gamma}  % \mathcal{S}
\newcommand{\setOpt} {\Gamma^\star}  % \mathcal{S}^\star

\newcommand{\Kmax}{ K }  % K_{max}

\newcommand{\mumax}{ \mu_{max} }
\newcommand{\mumin}{ \mu_{min} }
\newcommand{\emptysrv}{\ell_e}
\newcommand{\bk}{ {\mathbf{k}} }
\newcommand{\statez}{ \mathbf{z} }
\newcommand{\dsone}{ {\mathds{1}} }
\newcommand{\exactJ}[1]{\mathcal{J}^{(#1)}}

\newcommand{\cfcnt}{ {C} }  % count of calK
\newcommand{\gcnt}{ {{C}^{(g)}} }  % count of calK^g
\newcommand{\gtcnt}{ {{C}^{(g)}_{\bm{\rho}}} }  % count of calK^g up to k^{(\ggi)} (depends on rho)
\newcommand{\gtcntp}[1]{ {C_{#1}} }  % index up to which local greedy is close to global (depends on t)
\newcommand{\qcnt}{ {\gcnt} }

\newcommand{\subL}{ {f(L)} }  % length order of subinterval
\newcommand{\LC}{ {Z} }  % Lyapunov character
\newcommand{\perm}[1]{ {\sigma_{#1}} }  % previously i_
\newcommand{\permp}[2]{ {\sigma^{#2}_{#1}} }  % previously i_
\newcommand{\ggi}{ {I_{\bm{\rho}}} }  % previously I
\newcommand{\maptg}[1]{ {g_{#1}} }  % mapping from global greedy to greedy
\newcommand{\ji}[1]{ {j_{#1}} }  % job type that is exact according to property 1

\newcommand{\jip}[1]{ {j^\prime_{#1}} }
\newcommand{\jiB}[2]{ {j_{#1}[#2]} }  % job type that is exact according to property 1
\newcommand{\jix}[1]{j(#1)}
\newcommand{\subind}{ \mathcal{\bar J} }  % index subset notation needed for stopping time definition
\newcommand{\maxm}[1]{ {M_{#1}} }
\newcommand{\emptyx}{x_{\varnothing}}

\newenvironment{subproof}[1][\proofname]{%
	\begin{proof}[\textit{#1}]%
	}{%
	\end{proof}%
}

\theoremstyle{remark}
\newtheorem{remark}{Remark}
\newtheorem{claim}{Claim}

\definecolor{cadmiumgreen}{rgb}{0.0, 0.42, 0.24}

\acmDOI{}

\begin{document}
\title[A Theory of Auto-Scaling for Resource Reservation in Cloud Services]{
A Theory of Auto-Scaling for Resource Reservation in\\ Cloud Services 
}
\author{Konstantinos Psychas, Javad Ghaderi}
\affiliation{%
    \institution{Department of Electrical Engineering \\
    	Columbia University}
}

\begin{abstract}
\input{abstract-v3}
\end{abstract}
\maketitle

\input{introduction-v3}

\input{model4}

\input{algorithm16}

\input{analysis13}

\input{main_results2}
\input{simulations-v3}

\input{conclusions}
\bibliographystyle{ACM-Reference-Format}
\bibliography{Bibl,references}

\section*{Appendix}
\appendix

\input{permutation_proof}
\input{greedy_convergence}
\input{optimality_proof3}
\input{fluid_proof4}

\input{convergence_proof3}

\input{bound_proof}
\input{convergence_proof_of_drift5}

\input{lyapunov4}
\input{average_subsubintervals2}

\input{maintheorem_proof}

\end{document}

%% file: abstract-v3.tex
\label{abstract} 
We consider a distributed server system consisting of a large number of servers, each with limited capacity on multiple resources (CPU, memory, disk, etc.). 
Jobs with different rewards arrive over time and require certain amounts of resources for the duration of their service. 
When a job arrives, the system must decide whether to admit it or reject it, and if admitted, in which server to schedule the job. 
The objective is to maximize the expected total reward received by the system. 
	This problem is motivated by control of cloud computing clusters, in which, jobs are requests for Virtual Machines or Containers that reserve resources for various services, and rewards represent service priority of requests or price paid per time unit of service by clients.
We study this problem in an asymptotic regime where the number of servers and jobs' arrival rates scale by a factor $L$, as $L$ becomes large. 
We propose a resource reservation policy that asymptotically achieves \textit{at least} $1/2$, and under certain monotone property on jobs' rewards and resources, at least $1-1/e$ of the  optimal expected reward. 
The policy automatically scales the number of VM slots for each job type as the demand changes, and decides in which servers the slots should be created in advance, without the knowledge of traffic rates. 
It effectively tracks a low-complexity greedy packing of existing jobs in the system while maintaining only a small number, $g(L)=\omega(\log L)$, of reserved VM slots for high priority jobs that pack well.
  

%% file: introduction-v3.tex
\section{Introduction}
\label{introduction}
There has been a rapid migration of computing, storage, applications, and other services to cloud. 
By using cloud (e.g., Amazon AWS~\cite{AWS}, Microsoft Azure~\cite{Azure}, Google Cloud~\cite{google}), clients are no longer required to install and maintain their own infrastructure. 
Instead, clients use the cloud resources on demand, by procuring Virtual Machines (VMs) or Containers~\cite{Googlecontainer, AWScontainer} with specific configurations of CPU, memory, disk, and networking in the cloud data center, depending on their needs. 

A key challenge for the cloud service providers is to efficiently support a wide range of services on their physical platform. They usually offer QoS guarantees (in SLAs)~\cite{AmazonSLA} for clients' applications and services, and allow the number of VM instances to scale up or down with demand to ensure QoS guarantees are met. 
For example, in Amazon EC2 auto-scaling~\cite{auto-scaler}, clients can define simple rules to launch or terminate VM instances as their application demand increases or decreases. Various predictive and reactive schemes have been proposed for dynamically allocating VMs to different services, e.g.,~\cite{han2012lightweight,jiang2013optimal,mao2010cloud,roy2011efficient,qu2018auto,ghobaei2018autonomic}, however, they mostly assume a dedicated hosting model where VMs of each application run on a dedicated set of servers. Such models do not consider potential consolidation of VMs in servers which is known to significantly improve efficiency and scalability~\cite{corradi2014vm,song2013adaptive}.
For instance, suppose a CPU-intensive VM, a disk-intensive VM, and a memory-intensive VM are located on three individual servers (for our purpose, we use the terms VM and Container interchangeably). 
The cloud operator can pack these VMs in a single server to fully utilize its resources along CPU, disk, and memory, then the two unused servers can be used to pack additional VMs and serve more requests. 
However, in the absence of an accurate estimate of the workload, or when the workload varies over time and space, it is not clear how many VM instances an application launches and which VMs must be packed in which servers to ensure efficiency.

In this paper, we consider a cloud data center consisting of a large number of servers. 
As an abstraction in our model, a VM is simply a multi-dimensional object (vector of resource requirements) that should be served by one server and \textit{cannot be fragmented}. 
Each server has a limited fixed capacity on its available resources (CPU, memory, disk, networking). 
VM requests belong to a collection of VM types, each with a specific resource requirement vector, and a specific reward that represents its service priority or the price that will be paid per time  unit of service by the client. 
When a VM request arrives, we must decide in an online manner whether to accept it, and, if so, in which server to schedule it. 
The objective is to \textit{maximize the expected total reward} received by the system. 
Note that finding the right packing for a given workload is a hard combinatorial problem (related to multi-dimensional Knapsack~\cite{kellerer2004multidimensional}). 
The absence of accurate estimate of workload (VM traffic rates and service durations) makes the problem even more challenging. 
For instance, consider a simple scenario with three types of VMs with the following (CPU, memory) requirement and rewards: (0.6, 0.6) with reward 4, (0.7, 0.1) with reward 3, and (0.1, 0.7) with reward 3. Server's capacity is normalized to (1,1). 
Hence, a server can accommodate a single (0.6, 0.6) VM, or pack one (0.7, 0.1) VM and one (0.1, 0.7) VM together. 
Suppose there is one empty server, and a (0.6, 0.6) VM request arrives. 
Should we admit this request and receive a reward of 4, or reserve the server to pack one (0.7, 0.1) VM and one (0.1, 0.7) VM in future, which can potentially yield a maximum reward of 6?

This problem is related to the \textit{Online Multiple Knapsack} problem, in which there is a set of bins of finite capacity, items with various sizes and profits arrive one by one, and the goal is to pack them in an \textit{online} manner into the bins so as to maximize their total profit. 
In general, this problem does not have any competitive (constant approximation) algorithm~\cite{marchetti1995stochastic}, even when items are allowed to be removed from any bin at any time. 
Hence, proposed competitive algorithms focus on more restricted cases of the problem~\cite{iwama2002removable,cygan2016online}.

In this paper, we study a stochastic version of the problem in an asymptotic regime, where the number of servers $L$ is large and requests for VMs of type $j$ arrive at rate $\lambda_j L$, $j=1, \cdots, J$, and each requires service with mean  duration $1/\mu_j$. The (normalized) load of the system is defined as $\bm{\rho}:=(\lambda_j/\mu_j, j=1, \cdots, J)$. 
This is the \textit{heavy-traffic} regime, e.g.~\cite{kelly1991loss, whitt1985blocking,hunt1994large,hunt1997optimization,xie2015power, karthik2017choosing,mukhopadhyay2015mean}, 
and it has been shown that algorithms with good performance in such a regime also show good performance in other regimes. 
The interesting scenario occurs when not all VM requests can be scheduled (e.g., $\bm{\rho}>\bm{\rho}_c$ for a \textit{critical load} $\bm{\rho_c}$ on the boundary of system capacity), in which case a fraction of the traffic has to be rejected even by the optimal policy. 
We propose an adaptive reservation policy that makes admission and packing decisions without the knowledge of $\bm{\rho}$. 
Packing decisions include placement of admitted VM in one of the feasible servers, and migration of at most one VM across servers when a VM finishes its service.

\subsection{Related Work}
There is classical work on large loss networks, e.g.~\cite{hunt1994large,bean_gibbens_zachary_1995,kelly1991loss,hunt1997optimization}, where calls with different bandwidth requirements and priorities arrive to a telecommunication network. 
Trunk reservation has been shown to be a robust and effective call admission policy in this setting, in which each call type is accepted if the residual link bandwidth is above a certain threshold for that type. 
The performance of trunk reservation policies has been analyzed in the asymptotic regime where the call arrival rates and link's capacity scale up by a factor $N$, as $N \to \infty$. 
This is different from our large-scale server model, where the server's capacity is ``fixed'' and only the number of servers scales (a.k.a. system scale-out as opposed to scale-up).
	This makes the problem significantly more difficult, because, due to resource fragmentation when packing VMs in servers, the resources of servers \textit{cannot} be viewed as one giant pool; hence our policy not only needs to make admission decisions, but also decide in which server to place the admitted VM.
    Moreover, VMs have multi-dimensional resource as opposed to  
    one-dimensional calls (bandwidth). If we restrict that every server 
    can fit exactly one VM, our policy reduces to classical trunk reservation.% 

There has been past work on VM allocation~\cite{maguluri2012stochastic, stillwell2012virtual,zhao2015joint,psychas2017non,maguluri2014heavy,psychas2018randomized} and stochastic bin packing~\cite{gupta2012online,stolyar2015asymptotic,stolyar2013infinite,stolyar2013large, ghaderi2014asymptotic}, however their models or objectives are different from ours. 
The works~\cite{maguluri2012stochastic,psychas2017non,maguluri2014heavy,psychas2018randomized} consider a queueing model where VM requests are placed in a queue and then served by the system. 
In this paper, we are considering a loss model without delay, i.e., each VM request upon arrival has to be served immediately, otherwise it is lost. 
The recent works~\cite{stolyar2015asymptotic,stolyar2013infinite,stolyar2013large, ghaderi2014asymptotic} study a system with an infinite number of servers and their objective is to minimize the number of occupied servers. 
The auto-scaling algorithm proposed in~\cite{guo2018online} also assumes such an infinite server model. 
These are different from our setting where we consider a finite number of servers and study the total reward of served VMs by the system, in the limit as the number of servers becomes large. 
In this regime, we have to address complex fluid limit behaviors, especially when the load is above the system capacity and VMs have different priorities.

The works~\cite{xie2015power, karthik2017choosing,mukhopadhyay2015mean,stolyar2017large} study the blocking probability in a large-scale server system where \textit{all VMs have the same reward}.
The work~\cite{stolyar2017large} assumes a subcritical system load and only shows local stability of fluid limits. 
The works~\cite{xie2015power, karthik2017choosing,mukhopadhyay2015mean} show that, under a power-of-d choices routing, the blocking probability drops much faster compared to the case of uniform random routing. 
However, there is no analysis of optimality, especially in a supercritical regime where even the optimal policy has a non-zero blocking probability. 
Moreover, such algorithms treat all VMs with the same priority (reward) when making decisions, thus a low priority VM can potentially block multiple high priority ones.

We remark that in real clouds, servers are monitored periodically~\cite{HadoopYarn,shao2010runtime,aceto2013cloud}, for resource management, security, recovery, billing, etc., 
hence scheduling decisions can be made based on available information about the global system state.
     
\subsection{Contributions}

We propose a dynamic resource reservation policy that makes admission and packing decisions based on the current system state, and prove that it asymptotically achieves \textit{at least} $1/2$, and under certain monotone property on VMs' rewards and resources, at least $1-\frac{1}{e}$ of the optimal expected reward, as the number of servers $L \to \infty$.
Further, simulations suggest that for real cloud VM instances, the achieved ratio is in fact very close to one.

The main features of our policy and analysis technique can be summarized as follows:

\textbf{Dynamic Reservation.} 
The policy \textit{reserves slots} for VMs in advance. 
A slot for a VM type will reserve the VM's required resources on a specific server. 
An incoming VM request then will be admitted if there is enough reservation in the system, in which case it will fill an empty slot of that type.
The policy effectively tracks a low-complexity greedy packing of existing VM requests in the system while maintaining only a small number $g(L)=\omega(\log L)$ of empty slots (e.g., $(\log L)^{1+\varepsilon}$),
for VM types that have high priority at the current time. 
The reservation policy is \textit{robust} and can automatically adapt to changes in the workload based on requests in the system and new arriving requests, without the knowledge of $\bm{\rho}$.

\textbf{Analysis Technique.} 
Our proofs rely on analysis of fluid limits under the proposed policy, however, a major difficulty happens when the workload is above the critical load. 
In this regime the slot reservation process evolves at a much faster time-scale compared to the fluid-limit processes of the number of VMs and number of servers in different packing configurations in the system. 
To describe the behavior of fluid limits, we devise a careful analysis based on averaging the behavior of fluid-scale process over \textit{small intervals of length $\omega(\log L/L)$}.
We then introduce a \textit{Lyapunov function} based on a Linear Program. 
It is designed to have a unique maximizer at a global greedy solution and determines the convergence properties of our policy in steady state.

\subsection{Notations}\label{sec:notation}
For two positive-valued functions $x(n)$ and $y(n)$, with $n \in \mathds{N}$,
we write $x(n) = o(y(n))$ if $\lim_{n \to \infty} x(n)/y(n) = 0$, and 
$x(n) = \omega(y(n))$ if $y(n)=o(x(n))$.
$\dsone(\cdot)$ is the indicator function. 
$\mathbf{e}_j$ denotes the $j$-th basis vector. $t^-$ and $t^+$ denote 
the times right before and after $t$.
$\mathds{R}_+$ is the set of nonnegative real numbers. $(\cdot)^+=\max\{\cdot,0\}$.

%% file: model4.tex
\section{Model and Definitions}\label{sec:model}
\textbf{Cloud Model.} 
We consider a collection of $L$ servers denoted by the set $\mathcal{L}$. Each server $\ell \in \mathcal{L}$ has a limited capacity on different resource types (CPU, memory, disk, networking, etc.). 
We assume there are $n \geq 1$ types of resource. 

\textbf{VM Model.} There is a collection of VM types denoted by the set $\mathcal{J}$. The VM types are indexed in arbitrary order from $1$ to $J$. 
Each VM type $j$ requires a vector of resources $\mathbf{R}_j=(R_j^1,\cdots,R_j^n)$, where $R_j^d$ is its requirement for the $d$-th resource, $d=1,\cdots,n$.

VMs are placed in servers and reserve the required resources. The sum of reserved resources by the VMs placed in a server should not exceed the server's capacity.
A vector $\assgnk = (\assgnki_1, \cdots, \assgnki_J) \in \mathds{Z}_+^J$ is said to be a feasible configuration if the server can simultaneously accommodate
$\assgnki_1$ VMs of type $1$, $\assgnki_2$ VMs of type $2$, $\cdots$, $\assgnki_J$ VMs of type $J$. We use $\mathcal{K}$ to denote the set of all feasible configurations (including the empty configuration $\mathbf{0}_{J}$).
The number of feasible configurations will be denoted by $\cfcnt := |\mathcal{K}|$. 

We define $\mathcal{K}^{\mathcal{J}^\prime}$ to be the set of feasible configurations that include only VMs from a subset of types $\mathcal{J}^\prime \subseteq \mathcal{J}$, i.e., 
\be
&\mathcal{K}^{\mathcal{J}^\prime}=\{\assgnk \in \mathcal{K}: \assgnki_j =0,  \forall   j \notin \mathcal{J}^\prime \}.
\ee

We use $\Kmax < \infty$ to denote the maximum number of VMs that can fit in a server. We use $\assgnk^\ell(t) = \assgnk$ to denote that at time $t$, server $\ell \in \mathcal{L}$ has configuration $\assgnk$. 

We do not necessarily need the resource requirements to be additive, only the monotonicity of the feasible configurations is sufficient, namely, if $\assgnk \in \mathcal{K}$, and $ {\mathbf{k}^\prime} \leq \assgnk$ (component-wise), then ${\mathbf{k}^\prime} \in \mathcal{K}$. 
This will allow sub-additive resources as well, when the cumulative resource used by the VMs in a configuration could be less than the sum of the resources used individually~\cite{rampersaud2014sharing}.

\textbf{Job and Reward Model.} Jobs for various VM types arrive to the system over time. We can consider two models for jobs: 

(i) \textit{Revenue interpretation}: a job of type $j$ is a request to create a new VM of type $j$.

(ii) \textit{Service interpretation}: a job of type $j$ is a request that must be served by an existing VM of type $j$ in the system.  

To simplify the formulations and use one model to capture both interpretations, we assume that each VM can serve at most one job at any time. As we will see, our algorithm works based on creating ``\textit{reserved VM  slots}'' in advance. Hence, serving a newly arrived type-$j$ job can be interpreted as deploying a VM of type $j$ in its reserved slot (revenue interpretation), or assigning it to an already deployed VM of type $j$ in the slot (service interpretation).

Each job type $j$ is associated with a reward $u_j$ which represents its priority (service interpretation) or price paid per time unit of service (revenue interpretation).   

We define the \textit{feasible job placement} $\schedk= (\schedki_1, \cdots, \schedki_J)$ to be the set of jobs that are simultaneously being served in a single server, where $\schedki_j$ corresponds to the number of type-$j$ jobs. Note that by the definition of server configuration, it holds that
$\schedk \leq \assgnk$, for some $\assgnk \in \mathcal{K}$. Hence, $\assgnk - \schedk$ can be viewed as the reserved VM slots, where $\assgnki_j - \schedki_j$ is the number of reserved type-$j$ VM slots.
We use $\schedk^\ell(t) = \schedk$, when at time $t$, the job placement in  server $\ell \in \mathcal{L}$ is $\schedk$.

\textbf{Traffic Model.} 
Jobs of type $j$ arrive according to a Poisson process of rate $\lambda_j L$, for a constant $\lambda_j >0$. Once scheduled in a server (more accurately, in a reserved slot of type $j$), a job of type $j$ requires an exponentially distributed service time with mean $1/\mu_j$, and generates \textit{reward} at rate $u_j$ during its service. 
We define the normalized workload of type-$j$ jobs as $\rho_j := {\lambda_j}/{\mu_j}$ and the workload vector $\bm{\rho}=(\rho_j, j \in \mathcal{J})$.

\begin{definition}[Configuration Reward]\label{def:rew}
	The reward $U(\assgnk)$ of a configuration $\assgnk \in \mathcal{K}$ is defined as its total reward per unit time when its slots are full, i.e.,
	$
	U(\assgnk) := \sum_{j=1}^J u_j \assgnk_j.
	$
\end{definition}
\begin{definition}[Configuration Ordering]\label{def:ordering}	
	For two vectors $\assgnk, \assgnk^\prime \in \mathcal{K}$, we say $\assgnk \succ \assgnk^\prime$, 
    if either $U(\assgnk) > U(\assgnk^\prime)$, 
    or $U(\assgnk) = U(\assgnk^\prime)$ and considering the smallest $j$ for which $\assgnki_{j} \neq \assgnki^\prime_{j}$, $\assgnki_{j} > \assgnki^\prime_{j}$.
\end{definition}

\begin{definition}[MaxReward]\label{def:maxreward}
Given a subset $\mathcal{K}_s \subseteq \mathcal{K}$, the maximum reward configuration of $\mathcal{K}_s$ is defined as 
$$
\textsc{MaxReward}({\mathcal{K}_s}) := \arg\max_{\mathbf{k} \in \mathcal{K}_s} U(\assgnk),
$$
where ties are broken based on the ordering in Definition \ref{def:ordering}.
\end{definition}

\begin{definition}[State Variables]\label{def:systemj} 
    Consider the system with $L$ servers. We use $X^L_{\assgnk}(t)$ to denote the number of servers assigned to configuration $\assgnk \in \mathcal{K}$ at time $t$. 
    To distinguish between servers assigned to the same configuration $\assgnk$, 
	we \textit{index} them from $1$ to $X^L_{\assgnk}(t)$, starting from the most recent server assigned to $\assgnk$ (without loss of generality).

	The system state at time $t$ can then be described as 
	\begin{equation}\label{eq:systemstate}
		\mathbf{S}^L(t) := ((\assgnk^\ell(t), \schedk^\ell(t), c^\ell(t)), \ell \in \mathcal{L}),
	\end{equation}
	where for each server $\ell \in \mathcal{L}$, $\assgnk^\ell(t) \in \mathcal{K}$
	is its configuration, $\schedk^\ell(t)$, with $\schedk^\ell(t) \le \assgnk^\ell(t)$, is its job placement,
	and $c^\ell(t)$ is its index among the servers with configuration $\assgnk^\ell(t)$.

    The number of jobs of type $j$ in the system at time $t$ is given by
	\begin{equation}\label{eq:jobtypenumber}
	Y^{L}_j(t) = \sum_{ \ell \in \mathcal{L}}   \schedki_j^\ell(t). 
	\end{equation}
    We also define the vectors  $\mathbf{Y}^{L}(t)=(Y^{L}_j(t), j \in \mathcal{J})$, %$\mathbf{Z}^L(t)=\left(Z^L_{\assgnk, \schedk}(t),\schedk\leq \assgnk,\assgnk \in \mathcal{K}\right)$, 
    and $\mathbf{X}^L(t)=(X^L_{\assgnk}(t),\assgnk \in \mathcal{K})$. 
    Clearly $\sum_{\assgnk \in \mathcal{K}}X^L_{\assgnk}(t)= L$ since there are $L$ servers.
\end{definition}

\textbf{Optimization Objective.}
	Given a Markov policy $\pi$, we define the expected reward of the policy per unit time as 
	\begin{equation}\label{eq:costdef}
	F^{\pi} (L)=\lim_{t \to \infty} \mathds{E}
	\Big[\sum_{j \in \mathcal{J}} Y^{L}_j(t) u_j\Big].
	\end{equation}
	Our goal is to maximize the expected reward, i.e.,
	\begin{equation}\label{eq:obj}
	\text{maximize}_{\pi} F^{\pi} (L),
	\end{equation}
	where the maximization is over all Markov scheduling policies $\pi$. 
    Hence, when jobs are requests for VMs, this optimization is a revenue maximization, whereas when jobs are requests to be served by existing VMs, it is a weighted QoS maximization where each service is weighted by its priority.

    Note that under any Markov policy, the system state $\mathbf{S}^L(t)$ is a continuous-time irreducible Markov chain over a finite state space, hence it is positive recurrent and \dref{eq:costdef} is well defined. 
	Let $\mathbf{X}^L(\infty)$ and $\mathbf{Y}^{L}(\infty)$ be random vectors with the stationary distributions of $\mathbf{X}^L(t)$ and $\mathbf{Y}^{L}(t)$, respectively, as $t \to \infty$.
	Note that if $\mathbf{Y}^{\star}(t)$ is the number of jobs in an $M/M/\infty$ system in which every job is admitted, then $\mathbf{Y}^{L}(\infty)$ is stochastically dominated by ${\mathbf Y}^{\star}(\infty)$ whose stationary distribution is Poisson with mean $L \bm{\rho}$~\cite{Bolch2006}.
	
    We study the problem (\ref{eq:obj}) in the asymptotic regime where the 
	number of servers $L \to \infty$, while the job arrival rates are $\lambda_jL$, $j \in \mathcal{J}$. 
	Note that we do not make any assumption on the values of $\rho_j$. 
	
	Notice that as $t \to \infty$, the scaled stationary random variables satisfy $\frac{1}{L}\mathbf{X}^L(\infty) \leq \mathbf{1}$ 
    and $\frac{1}{L}{\mathbf{Y}^{L}}(\infty) \leq \frac{1}{L}\mathbf{Y}^{\star}(\infty)$. 
    This implies that the sequence of scaled random variables  is tight~\cite{billingsley2008probability}, therefore the (random) limits $\mathbf{x}(\infty):=\lim_{L \to \infty}\frac{1}{L}\mathbf{X}^L(\infty)$, and $\mathbf{y}(\infty):=\lim_{L \to \infty} \frac{1}{L}\mathbf{Y}^{L}(\infty)$ exist along a subsequence of $L$. 
    The limits satisfy ${x}_{\assgnk}(\infty) \geq  0$, $\sum_{\assgnk \in \mathcal{K}}{x_{\assgnk}}(\infty)=1$, and $\mathbf{y}(\infty) \leq  \bm{\rho}$, $ \mathbf{y}(\infty)\leq \sum_{\assgnk \in \mathcal{K}}  x_\assgnk(\infty) \assgnk$.    

    To unify the algorithm descriptions for revenue maximization and QoS maximization, in the rest of the paper,
    we use the term ``\textit{slot}'' of type $j$ to refer to the resource (equal to a VM of type $j$) reserved for one job of type $j$ in a server. Filled slots have jobs already in them, while empty slots could accept jobs. Therefore, the term configuration applies to all the slots in a server, while placement applies to the filled slots in the server.

%% file: algorithm16.tex
\section{A Static Optimization and its Greedy Solution}\label{sec:static_optim}
Given a workload reference vector $\hat{\mathbf Y}^L=(\hat{Y}^L_j, j \in \mathcal{J})$, let $F^\star(L,\hat{\mathbf Y}^L)$ be the optimal value of the following linear program:
\begin{maxi!}{\mathbf X, \mathbf Y}{\sum_j u_j Y_j} 
	{\label{eq:opt1c}}{}
	\addConstraint{Y_j\le \hat{Y}^L_j}, {\ \forall j \in \mathcal{J}} \label{eq:opt1c_Y}
	\addConstraint{\sum_{\assgnk \in \mathcal{K}}  X_\assgnk \assgnki_{j} \ge Y_j}, {\ \forall j \in \mathcal{J}} \label{eq:opt1c_X}
	\addConstraint{\sum_{\assgnk \in \mathcal{K}} X_\assgnk = L},{\quad X_\assgnk\geq 0,\ \forall \assgnk \in \mathcal{K}} \label{eq:opt1c_s}
\end{maxi!}
\noindent where $\mathbf Y$ is the vector of jobs in the system, and  $\mathbf X$ is the vector of the number of servers assigned to each configuration.
If we choose $\hat{\mathbf Y}^L=\bm{\rho}L$, this optimization will provide an upper bound on optimization (\ref{eq:obj}), i.e., $F^{\pi}(L) \leq F^\star(L,\bm{\rho}L) $, for any Markov policy $\pi$.
The interpretation of the result is as follows. The average number of type-$j$ jobs in the system cannot be more than its workload (Constraint \dref{eq:opt1c_Y}), and further, it cannot be more than the average number of slots of type $j$ in the servers (Constraint \dref{eq:opt1c_X}). 
The sum of number of servers in different configurations is $L$, so their average should also satisfy \dref{eq:opt1c_s}.  

As $L \to \infty$, the normalized objective value $\frac{1}{L}F^\star(L,\bm{\rho}L)\to {U}^{\star}[\bm{\rho}] $, 
which is the optimal value of the linear program below
\begin{maxi!}{\mathbf{x}, \mathbf{y}}{\sum_j u_j y_j} 
	{\label{eq:opt1b}}{}
	\addConstraint{y_j\le \rho_j}, {\forall j \in \mathcal{J}}
	\addConstraint{\sum_{\assgnk \in \mathcal{K}} \assgnki_{j} x_\assgnk
		\ge y_j}, {\forall j \in \mathcal{J}}\label{eq:opt1b_wl_a}
		\addConstraint{\sum_{\assgnk \in \mathcal{K}} x_\assgnk = 1}, {\quad x_\assgnk\geq 0,\ \forall \assgnk \in \mathcal{K}}\label{eq:opt1b_s_a}
\end{maxi!}
where $x_\assgnk$ can be interpreted as the ideal fraction of servers which should be in configuration $\assgnk$ when $L$ is large. 
Hence, one can consider a static reservation policy where the cloud cluster is partitioned and $\lfloor x_\assgnk L \rfloor $ servers are assigned to each non-zero configuration $\assgnk \in \mathcal{K}$ (and the rest of servers can be empty to save resource or used to serve more jobs). 
Then once a type-$j$ job arrives, it will be routed to an empty slot of type $j$ in one of the servers, if any, otherwise it is rejected.
This will provide an asymptotic optimal policy since it achieves the normalized reward ${U}^{\star}[\bm{\rho}]$, as $L \to \infty$.

However, there are several issues with this approach: (i) solving optimization~(\ref{eq:opt1c}) or its relaxation~(\ref{eq:opt1b}) has a very high complexity, as the number of configurations is exponential in the number of job types $J$, and (ii) it requires knowing an accurate estimate of the workload $\bm{\rho}$ which might not be available. 
Inaccurate estimates of workload can lead to poor performance for such policies, e.g., see~\cite{Key1990} which illustrates that static reservation policies in classical loss networks can give very poor performance.
Even if we have an estimate of the workload and approximate the solution to (\ref{eq:opt1b}), to handle time-varying workloads, the new solution may require rearranging a large number of VMs and jobs to make their placements match the new solution. 
This is costly and also causes interruption of many jobs in service.

We first address the complexity issue, by presenting a greedy solution for the optimization, and analyze its asymptotic performance below.
\subsection{Greedy Solution}
    We describe a greedy algorithm, called \textit{Greedy Placement Algorithm} ({\VPA}), for solving   optimization~(\ref{eq:opt1c}). 

    {\VPA} takes as input the workload reference vector
    $\hat{\mathbf Y}^L$,  and returns an assignment vector 
    $\hat{\bm{X}}^L$ which indicates which configurations
    should be used and in how many servers. 
    The assignment consists of at most $J$ configurations, 
    which are found in $J$ iterations.  
    In each iteration $i$, {\VPA} maintains a set of candidate 
    job types $\mathcal{J}[i]$, and finds a configuration $\mathbf{k}[i]$. 
    Initially $\mathcal{J}[1]=\mathcal{J}$. 
    In iteration $i$:
   	\begin{enumerate}[leftmargin=*] 
	\item It finds $\mathbf{k}[i]=\textsc{MaxReward}(\mathcal{K}^{\mathcal{J}[i]})$, which is the configuration of highest reward among the configurations that have jobs from the set $\mathcal{J}[i]$, according to Definition~\ref{def:maxreward}. 
	\item It computes the number of servers $\hat{X}^L_{\assgnk[i]}$ that should be assigned to $\mathbf{k}[i]$, 
	until at least one of the job types $j$, for which $k_{j}[i] > 0$, has no more jobs left, or there are 
	no more unused servers left. We refer to this job type as $j^\star$.
	\item It then creates $\mathcal{J}[i+1]$ by removing job type $j^\star$ from $\mathcal{J}[i]$.
   	\end{enumerate}

	\begin{algorithm}[t]
    \caption{Greedy Placement Algorithm ({\VPA})}\label{alg:vpa}
    \begin{algorithmic}[1]
        \Function{{\VPA}}{$\mathbf{\hat{Y}}$}
        \State $\mathbf{r} \gets \mathbf{\hat{Y}}$ \Comment tracks the vector of number of jobs left
        \State $N \gets L$  
        \Comment tracks the number of servers left
        \State $i \gets 1$,\ $\mathcal{J}[1] = \mathcal{J}$
        \While {$\mathcal{J}[i] \neq \varnothing$}
        \State \label{ln:kupdate_vpa} $\assgnk[i] \gets 
        \textsc{MaxReward}(\mathcal{K}^{\mathcal{J}[i]})$
        \State $j^\star \gets \argmin_{j:\assgnki_{j}[i]>0} 
            \ceil{\frac{r_j}{\assgnki_{j}[i]}}$ \Comment break ties arbitrarily 
        \State ${\hat X}_{\assgnk[i]} \gets \min \left(\ceil{\frac{r_{j^\star}}{\assgnki_{j^\star}[i]}}, N\right)$
        \State $\mathbf{r} \gets \mathbf{r} - {\hat X}_{\assgnk[i]} \assgnk[i]$
        \State $N \gets N - {\hat X}_{\assgnk[i]}$
        \State \label{ln:Jgupdate_vpa}
        $\mathcal{J}[i+1] \gets \mathcal{J}[i] - \{
        j^\star \}$
        \State $i \gets i+1$
        \EndWhile
        \State \Return  ${\hat X}_{\assgnk[j]}$, $j=1,\cdots,J$
        \EndFunction
    \end{algorithmic}
    \end{algorithm}
A pseudocode for {\VPA} is given by Algorithm~\ref{alg:vpa}.
We use the vector $\mathbf{\hat X}^{L}=({\hat X}^{L}_\assgnk, \assgnk \in \mathcal{K})$ to denote the output of {\VPA}, which has at most $J$ non-zero elements corresponding to $\assgnk[i]$, $i=1, \ldots, J$.

\begin{remark}
	$\textsc{MaxReward}$ finds the maximum reward configuration of a subset of job types, which is equivalent with \textit{unbounded} Knapsack problem (unbounded number of items for each type). 
    This problem is tractable with Pseudopolynomial algorithms to solve it exactly ~\cite{andonov2000unbounded,martello1990exact} or fully polynomial approximation algorithms~\cite{ibarra1975fast}. 
    {\VPA} needs to solve at most $J$ instances of this problem. 
    Note that the number of different instances of the problem is bounded and we can compute \textsc{MaxReward} for all of them offline as \textit{they are not workload dependent}. 
    This is in contrast to optimization \dref{eq:opt1c}, 
    which is equivalent to multi Knapsack problem which is strongly NP-hard~\cite{kellerer2004multidimensional}, and requires resolving when workload reference $\mathbf{\hat Y}$ changes.  
\end{remark}

We next define the limit of $\hat{\bm{X}}^L/L$ for input $\hat{\bm{Y}}^L = L \bm{\rho}$, as  $L \to \infty$, which we refer to as \textit{Global Greedy Assignment}. 
To describe this assignment, we first define a unique ordering of the job types through the following proposition.

    \begin{proposition}\label{prop:index}
        For any permutation $\sigma = (\perm{1}, \perm{2}, \ldots, \perm{J})$ of 
        job types in $\calJ$, let 
        $\mathcal{J}^\sigma_{j} := \{\perm{j}, \ldots, \perm{J}\}$,
        and $\mathbf{k}^{(j)} := \textsc{MaxReward}(\mathcal{K}^{\mathcal{J}^\sigma_j})$.
        Given a workload $\bm{\rho}$, there is a ``unique'' permutation $\sigma = (\perm{1}, \perm{2}, \ldots, \perm{J})$ of 
        job types, such that the following holds: 
        \begin{itemize}[leftmargin=0.15 in]
        \item[1)] $\forall j\in \calJ$, ${k}^{(j)}_\perm{j} > 0$, and
        there are constants $\globl{z}^{(j)}[\bm{\rho}] \ge 0$, such that
        \begin{equation}\label{eq:rho_xg}
        \rho_\perm{j} = \sum_{\ell=1}^j {k}^{(\ell)}_\perm{j}
        \globl{z}^{(\ell)} [\bm{\rho}],
        \end{equation}
        \item[2)]for any two indexes $j, j^\prime \in \calJ$, 
        with $j < j^\prime$,  if
        \begin{equation}\label{eq:rho_xg2}
        \rho_\perm{j^\prime} = \sum_{\ell=1}^{j} {k}^{(\ell)}_\perm{j^\prime} 
        \globl{z}^{(\ell)} [\bm{\rho}],
        \end{equation}
          then we should have $\perm{j} < \perm{j^\prime}$.
        \end{itemize}
    \end{proposition}
    \begin{proof}
        See Appendix~\ref{indexproof}.
    \end{proof}

    The Global Greedy Assignment is defined as follows
    \begin{definition}[Global Greedy Assignment]\label{def:gga}
        Define the index $\ggi \leq J$ for which
        \ben
        & \sum_{i=1}^{\ggi-1} z^{(i)}[\bm{\rho}] < 1, \quad
        \sum_{i=1}^{\ggi} z^{(i)}[\bm{\rho}] \ge 1,
        \een
         with the convention that $\ggi=J+1$ if $\sum_{i=1}^{J} z^{(i)}[\bm{\rho}] < 1$.
        The \textit{global greedy assignment} $\globlt{\mathbf{x}}^{(g)}[\bm{\rho}]$ is defined as
        \begin{equation}
        \globlt{x}^{(g)}_{\bk^{(i)}}[\bm{\rho}] =\begin{cases}
        z^{(i)}[\bm{\rho}],&  \text{for }  i < \ggi \\
        0, & \text{for }   i > \ggi \\
        1 - \sum_{j=1}^{i-1} \globlt{x}^{(g)}_{\bk^{(j)}}[\bm{\rho}], &
        \text{for }  i = \ggi,
        \end{cases}
        \end{equation} 
        where $\bk^{(i)}$ and $z^{(i)}[\bm{\rho}]$, $i=1, \ldots, J$, 
        were defined in Proposition~\ref{prop:index}, and ${\bk}^{(J+1)}:=\mathbf{0}$ (empty configuration). 
        We call the ordered configurations $\bk^{(i)}$, $i=1, \ldots, J+1$, the ``global greedy configurations'' of workload $\bm{\rho}$. For any configuration $\bk \in \calK$ not in 
        global greedy configurations, $\globlt{x}^{(g)}_{\bk}[\bm{\rho}]=0$. 
        When it is clear from the context, 
        the dependency $[\bm{\rho}]$ will be omitted.
    \end{definition}

    %%%%%%%%%%%%%%%%%%%%%%%%%%%%%%%%%%%%%%%%%%
    Since global greedy configurations ${\assgnk}^{(\ell)}$, $\ell=1, \ldots, J+1$, depend on $\bm{\rho}$, 
    the following configurations will come in handy when
    the analysis needs to be agnostic to $\bm{\rho}$.
    \begin{definition}[Greedy Configurations]\label{def:greedy_conf}
        The greedy configuration set $\mathcal{K}^{(g)}$ includes all configurations that are output of 
        \textsc{MaxReward}($\mathcal{K}^{\mathcal{J}^\prime}$)
        for any $\mathcal{J}^\prime \subseteq \mathcal{J}$.
        That is the set of all possible	configurations which 
        may be assigned by {\VPA}, and the empty configuration.
        We define $\gcnt := |\mathcal{K}^{(g)}|$. 	
        We enumerate configurations of $\mathcal{K}^{(g)}$ 
        as $\greedk^{(i)}$, for $i=1, \ldots, \gcnt$,
        such that $\greedk^{(i_1)} \succ \greedk^{(i_2)}$ if $i_1 < i_2$ (according to Definition~\ref{def:ordering}), and $\greedk^{(\gcnt)} = \mathbf{0}_J$.
    \end{definition}
    Notice that 
    $
    \{\bk^{(j)}, j=1, \ldots, J+1\} \subseteq \{\greedk^{(i)}, i=1, \ldots, \gcnt\},
    $
    and their order is consistent with Definition~\ref{def:ordering}, as defined below.
    \begin{definition}[Mapping global greedy to greedy]
        \label{def:mapgreedy}
        For any $j, j^\prime \in \{1, \ldots, J+1\}$, with
        $j < j^\prime$, there are indexes $\maptg{j}, \maptg{j^\prime}
        \in \{1, \ldots, \gcnt\}$, such that
        $\bk^{(j)} \equiv \greedk^{(\maptg{j})}$, 
        $\bk^{(j^\prime)} \equiv \greedk^{(\maptg{j^\prime})}$, and 
        $\maptg{j} < \maptg{j^\prime}$. 
        We also define $\gtcnt:=\maptg{\ggi}$ to be the index for which $\bk^{(\ggi)} \equiv \greedk^{(\gtcnt)}$.
    \end{definition}

    The following proposition states the connection between {\VPA} and Global Greedy Assignment $x^{(g)}_{\bk}[\bm{\rho}]$.
    \begin{proposition}\label{prop:greedy_conv}
        Let $\hat{\bm{X}}^L = $ {\VPA}$(L \bm{\rho})$.
        Then
        \begin{equation}
        \lim_{L \to \infty} \frac{\hat{X}_{\bk}^L}{L} =
        x^{(g)}_{\bk}[\bm{\rho}],\ \forall \bk \in \calK,
        \end{equation}
        where $x^{(g)}_{\bk}[\bm{\rho}]$ is the Global Greedy Assignment
        of Definition~\ref{def:gga}.
    \end{proposition}
    \begin{proof}
        See Appendix~\ref{sec:greedy_conv}.
    \end{proof}
    Note that clearly $\mathbf{\globlt{x}}^{(g)}[\bm{\rho}]$ is a feasible solution for optimization (\ref{eq:opt1b}) and it is easy to see that its corresponding objective value is
    \begin{equation}\label{eq:greedyobjective}
    {U}^{(g)}[\bm{\rho}] :=
    \sum_{j=1}^J u_j 
    \sum_{\ell=1}^J k^{(\ell)}_j \globlt{x}^{(g)}_{\mathbf{k^{( \ell)}}}[\bm{\rho}].
    \end{equation}
    It is also easy to see that in optimization~(\ref{eq:opt1b}) we can replace the inequality in (\ref{eq:opt1b_wl_a}) with equality and the optimal value will not change.
    Let $\mathbf{x}^\star[\bm{\rho}]$ be one such optimal solution to optimization (\ref{eq:opt1b}) for workload $\bm{\rho}$. 
    Then the optimal objective value is
    \begin{equation}
    {U}^{\star}[\bm{\rho}] :=
    \sum_{j\in \calJ} u_j 
    \sum_{\bk \in \mathcal{K}} k_j x^{\star}_{\bk}[\bm{\rho}].
    \end{equation}
    The following corollary is immediate from Proposition~\ref{prop:greedy_conv}.
    \begin{corollary}\label{cor:relation}
    	Let $F^{\VPA}(L,\bm{\rho} L)$ be the total reward of {\VPA} in the system with $L$ servers given reference workload $\hat{\bm{Y}}^L = \bm{\rho} L$. 
        Then
    	\ben
        \lim_{L \to \infty}	\frac{F^{\VPA}(L,\bm{\rho} L) }{F^{\star}(L,\bm{\rho} L)}=\frac{{U}^{(g)}[\bm{\rho}]}{{U}^{\star}[\bm{\rho}]}.
    	\een  
    \end{corollary}	
    The theorem below bounds the above ratio.
    \begin{theorem}\label{thm:optimality}
        The global greedy assignment 
        ${\globlt{\mathbf{x}}}^{(g)}[\bm{\rho}]$ provides at least $\frac{1}{2}$ 
        of the optimal normalized reward, i.e., 
        $\frac{{U}^{(g)}[\bm{\rho}]}{{U}^{\star}[\bm{\rho}]}\geq \frac{1}{2}$, $\forall \bm{\rho}\geq0$.
    \end{theorem}
    \begin{proof}

    	Consider the permutation of job types according to Proposition~\ref{prop:index}. 
        By the global greedy definition and the feasibility of $\mathbf{x}^\star[\bm{\rho}]$, for any job type $\perm{j}$, $j=1, \ldots, \ggi-1$, we have
        \begin{equation}
        \sum_{\mathbf{k} \in \mathcal{K}} x^\star_{\mathbf{k}} 
        k_\perm{j} \le \sum_{\ell=1}^j \globlt{x}^{(g)}_{\mathbf{k}^{( \ell)}} 
        k^{( \ell)}_\perm{j} = \rho_\perm{j},
        \end{equation}
        from which it follows that
        \begin{equation}\label{eq:sumIM1}
        \sum_{j=1}^{\ggi-1} \sum_{\mathbf{k} \in \mathcal{K}} x^\star_{\bk} 
        k_\perm{j} u_\perm{j} \le \sum_{j=1}^{\ggi-1} \sum_{\ell=1}^j 
        \globlt{x}^{(g)}_{\mathbf{k}^{( \ell)}} k^{(\ell)}_\perm{j} u_\perm{j} = 
        \sum_{j=1}^{\ggi-1} \rho_\perm{j} u_\perm{j} \le {U}^{(g)}[\bm{\rho}].
        \end{equation}
        % The next summation may be empty
        Also for the job types $\perm{j}$, for $j = \ggi, \ldots, J$, we have
        \begin{equation}\label{eq:sumI}
        \begin{aligned}
        &\sum_{j=\ggi}^J \sum_{\mathbf{k} \in \mathcal{K}} x^\star_{\bk} 
        k_\perm{j} u_\perm{j} 
        = \sum_{\mathbf{k} \in \mathcal{K}} x^\star_{\bk} 
        \sum_{j=\ggi}^J k_\perm{j} u_\perm{j} \stackrel{(a)}{\le} \\
        &\argmax_{\mathbf{k} \in \mathcal{K}} 
        \sum_{j=\ggi}^J k_\perm{j} u_\perm{j}
        \stackrel{(b)}{=} \sum_{j=\ggi}^J k^{( \ggi)}_\perm{j} u_\perm{j} \stackrel{(c)}{\le} {U}^{(g)}[\bm{\rho}].
        \end{aligned}
        \end{equation}
        where (a) is due to the fact that $\sum_{\bk \in \mathcal{K}} x^\star_{\bk}= 1$, 
        (b) is by the definition of $\assgnk^{(\ggi)}$, 
        and (c) is because ${U}^{(g)}[\bm{\rho}]$ is a convex combination of rewards of $\assgnk^{(1)}, \ldots, \assgnk^{(\ggi)}$, 
        which all have a reward no less than that of $\assgnk^{(\ggi)}$.
        Then adding (\ref{eq:sumIM1}) and (\ref{eq:sumI}), we get
        \[
        {U}^{\star}[\bm{\rho}] = 
        \sum_{j=1}^J \sum_{\bk \in \mathcal{K}} x^\star_{\bk} 
        k_\perm{j} u_\perm{j}
        \le 2 {U}^{(g)}[\bm{\rho}]. \qedhere
        \]
    \end{proof}
    
    Theorem~\ref{thm:optimality} can be improved when job types and rewards satisfy a monotone greedy property described next.
    \begin{definition}\label{def:monotone greedy}
        We say the job types and the rewards have \textit{monotone greedy} 
        property if for any two instances of the optimization (\ref{eq:opt1b}) with
        $\bm{\rho}_1 \ge \bm{\rho}_2$, ${U}^{(g)}[\bm{\rho}_1] \ge {U}^{(g)}[\bm{\rho}_2]$.
    \end{definition}
    
    It is easy to verify that any system with two job types always has the property in Definition~\ref{def:monotone greedy}. 
    However, in general the property depends on the profile of jobs types and their rewards, and might not hold for adversarial profiles.
    The next theorem describes the improved bound when the monotone greedy property holds.
    \begin{theorem}\label{thm:optimality2}
        If job types and rewards satisfy the \textit{monotone greedy} property,
        then, for any $\bm{\rho}$, 
        $\frac{{U}^{(g)}[\bm{\rho}]}{{U}^{\star}[\bm{\rho}]} \ge 1 - 1/e$.
    \end{theorem}
    \begin{proof}
        Define a workload $\label{eq:rhost}
        \bm{\rho^\star} := \sum_{\bk \in \mathcal{K}}
        \bk x^\star_{\bk}[\bm{\rho}].
        $
        We notice that ${U}^{\star}[\bm{\rho}] = {U}^{\star}[\bm{\rho^\star}]$ in LP~(\ref{eq:opt1b}). 
        Also by the monotone greedy property, ${U}^{(g)}[\bm{\rho}] \ge {U}^{(g)}[\bm{\rho^\star}]$, since $\bm{\rho} \ge \bm{\rho^\star}$. 
        Hence, it suffices to prove the theorem for instances where $\bm{\rho} = \bm{\rho}^\star$ or in other words, instances for which, in the optimal solution, workload fits exactly in servers.
        
        Consider now the projection of the workload $\bm{\rho^\star} = \bm{\rho}$ onto the global greedy
        configuration space $\{\mathbf{k}^{(i)}[\bm{\rho}], 
        \ i=1, \ldots, J \}$. 
        Since these configurations are independent, we can write
        \be\label{eq:rhost2}
        & \bm{\rho^\star} = \bm{\rho}=\sum_{i=1}^J z^{(i)}[\bm{\rho}] 
        {\bk^{(i)}},
        \ee
        for $z^{(i)}[\bm{\rho}]$ introduced in Proposition~\ref{prop:index}.
        For notational compactness, define $q_i = z^{(i)}[\bm{\rho}]$,
        $i=1, \ldots, J$, and
        $p_i = \globlt{x}^{(g)}_{\bk^{( i)}}[\bm{\rho}]$, $i=1, \ldots, \ggi$, and let $W^{(i)} := U(\bk^{(i)}) = \sum_{j=i}^J u_\perm{j} k^{(i)}_\perm{j}$. 

        Then,
        \begin{align}
        \sum_{j=i}^J q_j W^{(j)} &= 
        \sum_{j=i}^J u_\perm{j} \sum_{\ell=i}^j q_{\ell} k^{( \ell)}_\perm{j}
        \stackrel{(a)}{\le} \sum_{j=i}^J \rho_\perm{j} u_\perm{j} \nonumber \\
        &\stackrel{(b)}{\le} \sum_{j=i}^J k^{(i)}_\perm{j} u_\perm{j} = W^{(i)}.\label{eq:rewardboundq}
        \end{align}
        Inequality (a) is because $\sum_{\ell=i}^j q_{\ell} k^{(\ell)}_\perm{j} \le \rho_\perm{j}$, 
        and Inequality (b) is because we assumed there is 
        an assignment that can completely accommodate workload $\bm{\rho}$, and hence $\rho_\perm{j}$ for $j=i,\ldots,J$.
        If we remove all jobs with types $1, \ldots, i-1$ from assignment $\mathbf{x}^\star$, 
        the configurations used in the resulting assignment belong to the subset $\mathcal{K}^{\{\perm{i},\ldots,\perm{J}\}}$ 
        and $\bk^{(i)}$ is the configuration with the highest reward from this set.
        
        An equivalent representation of (\ref{eq:rewardboundq}) is that, for some constants $b_i$, $0 \le b_i \le 1$, $i=1, \ldots J$,
        \begin{equation}
        \begin{aligned}
        & b_i W^{(i)} = \sum_{j=i}^J q_j W^{(j)},\ \text{and } (b_i - q_i) W^{(i)} = b_{i+1} W^{(i+1)}.
        \end{aligned}
        \end{equation}
        For completeness, we also define $b_{J+1}=1$.
        Based on this representation, and using (\ref{eq:rhost2}) and $\bm{\rho} = \bm{\rho^\star}$ by assumption, 
        we get
        \begin{equation}
        \begin{aligned}
        & \frac{{U}^{(g)}[\bm{\rho}]}{{U}^{\star}[\bm{\rho}]} = 
        \frac{\sum_{i=1}^{\ggi} p_i W^{(i)}}{\sum_{i=1}^{J} q_i W^{(i)}} = 
        \frac{\sum_{i=1}^{\ggi} p_i \prod_{j=1}^{i-1}\frac{b_j - q_j}{b_{j+1}}W^{(1)}}
        {b_1 W^{(1)}} = \\
        & \sum_{i=1}^{\ggi} \frac{p_i}{b_i} 
        \prod_{j=1}^{i-1}\frac{b_j - q_j}{b_{j}} =
        1 - \prod_{i=1}^\ggi \left(1 - \frac{p_i}{b_i} \right).
        \end{aligned}
        \end{equation}
        The right-hand side is minimized if $b_i=1$, $i=1, \ldots, \ggi$,
        since $p_i \ge 0$. Then given $\sum_{i=1}^\ggi p_i = 1$, the expression 
        is minimized for $p_i = 1/\ggi$, $i=1, \ldots, \ggi$, and its minimum value is
        $
        1 - \left(1 - \frac{1}{\ggi}\right)^{\ggi} > 1 - e^{-1}.
        $
    \end{proof}
    
    \begin{proposition}\label{prop:bound_tight}
        The worst-case ratio of ${U}^{(g)}[\bm{\rho}]/ {U}^\star[\bm{\rho}]$ is not greater than $1 - 1/e$.
    \end{proposition}
    \begin{proof}
        We construct an adversarial example that achieves this bound. See Appendix~\ref{prf:optimality}.
    \end{proof}
    Hence, the global greedy assignment achieves a factor within $1/2$ to $1-1/e$ of the optimal normalized reward in ``all'' the cases. 
    Further, the bound $1-1/e$ is tight when monotone greedy property holds. 
    The assignment might actually achieve $1-1/e$ in all the cases but requires a more careful analysis. 
    In view of Corollary~\ref{cor:relation}, $\text{\VPA}(\bm{\rho} L)$ asymptotically achieves the same factor of the optimal reward. 
    In simulations in Section~\ref{sec:sim}, based on cloud VM instances, we were not able to find any scenario where the ratio is below $1-1/e$, and in fact the ratio is much better ($\approx 0.97$ on average).     
    
    However, $\text{\VPA}(\bm{\rho} L)$ requires the knowledge of $\bm{\rho}$. 
    In the next section, we propose a dynamic reservation algorithm that is appropriate for use in online settings without the knowledge of $\bm{\rho}$. 
    Its achievable normalized reward still converges to that of the global greedy assignment 
    and it can also adapt to changes in the workload.

%%%%%%%%%%%%%%%%%%%%%%%%%%%%%
\section{Dynamic Reservation Algorithm}\label{sec:main}
We present a Dynamic Reservation Algorithm, called {\DRA}, which makes admission decisions and configuration assignments, without the knowledge $\bm{\rho}$. We first introduce the following notations:
\begin{itemize}[leftmargin=*]
	\item Recall the indexing of servers in the same configuration as in Definition~\ref{def:systemj}. 
	We use $\ell_{\mathbf{k}, i}$ to refer to the server with configuration $\assgnk$ and index $i$.
	\item A key parameter of {\DRA} is the \textit{reservation factor} $g(L)$. 
    It is the number of empty slots (safety margin) that the algorithm ideally wants to reserve for each job type if possible. 
	For later analysis, we assume that $g(L)=\omega(\log(L))$, and is $o(L)$.
\end{itemize}

The configuration assignment occurs at \textit{update times}. 
To simplify the analysis, we consider update times to be times when a job is admitted to or departs from the system. 
To avoid preemptions, only servers that are empty (have no jobs running) can be assigned to a new configuration.

At update time $t$, {\DRA} updates the workload reference vector $\mathbf{\hat{Y}}^L(t)$ as
\be \label{eq:hatYdefinition}
\mathbf{\hat{Y}}^L(t) = \mathbf{{Y}}^L(t) + g(L)\mathbf{1},
\ee
where $\mathbf{{Y}}^L(t)$ in the vector of jobs in the system, \textit{after} any job admission or job departure at time $t$.
$g(L)$ is the reservation factor as defined earlier.

Then {\DRA} classifies the servers into two groups: \textit{Accept Group} (\texttt{AG}) and \textit{Reject Group} (\texttt{RG}).
Servers in {Accept Group} keep their current configurations and {\DRA} attempts to have all their slots filled by scheduling new jobs in them, while servers in {Reject Group} do not have desirable configurations and {\DRA} attempts to make them empty, by not scheduling new jobs in them and possibly migrating their jobs to servers in Accept Group, so they can be reassigned to other configurations.

    A pseudocode for {\DRA} is given in Algorithm~\ref{alg:main3}. It has three main components which we describe in detail below:
    
    \textbf{Classification and Reassignment Algorithm ({\EPA}).}
        This is the subroutine used by {\DRA} to classify servers and possibly reassign some of them.
    It attempts to greedily reduce the disparity between the configuration assignment in the system $\mathbf{X}^L(t)$ and the output of {\VPA} $\mathbf{\hat X}^L(t)=\textrm{\VPA}(\mathbf{\hat{Y}}^L(t))$.
    To do so, it assigns \textit{ranks} to servers in different configurations, 
    which range from $1$ to $J+1$.

    Initially, all servers are assigned rank $J+1$.
    Any empty server of rank $J+1$ can be reassigned to reduce the disparity
    between  $\mathbf{X}^L(t)$ and $\mathbf{\hat X}^L(t)$. 
    We use $\emptysrv$ to denote one of empty rank $J+1$ servers, and if no such server exists $\emptysrv=\varnothing$.

    Iterating over configurations $\assgnk[i]$ found by {\VPA}, for $i=1, \ldots, J$:
    \begin{itemize}[leftmargin=*]
        \item If $X^L_{\assgnk[i]} < \hat{X}^L_{\assgnk[i]}$, it increases $X^L_{\assgnk[i]}$ by reassigning any $\emptysrv$ to $\assgnk[i]$, until either (i) it matches $\hat{X}^L_{\assgnk[i]}$, or (ii) $\emptysrv=\varnothing$.
        In either case, all servers of configuration $\assgnk[i]$ get rank $i$.  
        \item If $X^L_{\assgnk[i]}(t) \ge \hat{X}^L_{\assgnk[i]}$, it assigns rank $i$ to all servers of configuration ${\assgnk}[i]$ 
        with indexes greater than $X^L_{\assgnk[i]}(t)-\hat{X}^L_{\assgnk[i]}(t)$.
    \end{itemize}

    We use $I^\star(t)$ to denote the first $i$ for which $X^L_{\assgnk[i]}$ cannot be matched to $\hat{X}^L_{\assgnk[i]}$, i.e. the first $i$ at which $\emptysrv=\varnothing$. If all configurations are matched,
    then $I^\star(t)=J$.
    At the end of {\EPA}, servers with rank greater than $I^\star(t)$ 
    and index $1$ in any configuration are classified 
    as Reject Group, while the rest of the servers are 
    classified as Accept Group. 

    See Figure~\ref{fig:epa_example} for an illustrative example for the state of {\EPA}.

    \begin{figure}[t]
        \centering
        \includegraphics[width=\columnwidth]{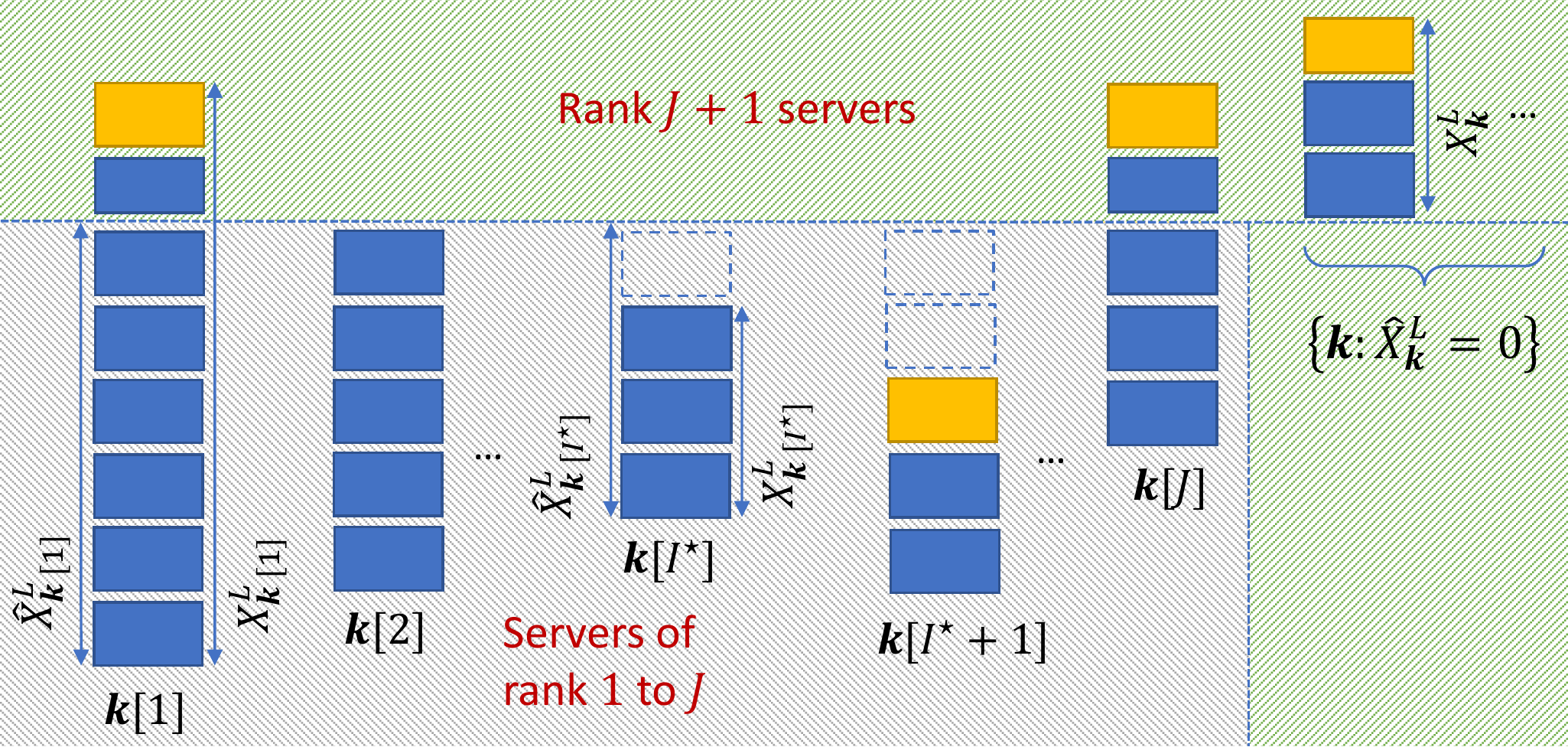}
        \caption{An example illustrating the state at the end of {\EPA}. 
            Servers in each configuration are stacked from largest
            to smallest index.
            ${\assgnk}[1], \ldots, {\assgnk}[J]$ are the configurations 
            returned by {\VPA}. The dashed boxes indicate how many more 
            servers need to be reassigned to a respective configuration 
            to match the solution of {\VPA} (horizontal line).  
            $I^\star$ is the first $i$ for which $X^L_{{\assgnk}[i]} < \hat{X}^L_{{\assgnk}[i]}$ at the end of the procedure. 
            Orange servers are the servers of Reject Group. 
        }
        \Description{An example illustrating the state at the end of {\EPA}.}
        \label{fig:epa_example}
    \end{figure}

	%%%%%%%%%%%%%%%%%%%%
	\textbf{Scheduling Arriving Job.} 
        When {\DRA} needs to schedule an arriving job of type $j$, 
        it places the job in one of the servers of Accept Group with empty type-$j$ slot.
        If no such server exists, the job is rejected. 
	   
        We use ${\texttt{AG}}_j$ to denote one of the servers of Accept Group with empty type-$j$ slot. If no such server exists ${\texttt{AG}}_j=\varnothing$.

	%%%%%%%%%%%%%%%%%%%%
	\textbf{Migrating Job after Departure.}
	Let ${\texttt{RG}}_j$ denote the highest rank server
	among the Reject Group servers with type-$j$ jobs. If no such server exists, ${\texttt{RG}}_j=\varnothing$.
        
    If a type-$j$ job departs from a server in Accept Group, {\DRA} migrates one of the type-$j$ jobs from ${\texttt{RG}}_j$ to the slot that emptied because of the departure, if ${\texttt{RG}}_j \neq \varnothing$.

	%%%%%%%%%%%%%%%%%%%%
	\textbf{Initialization.}\label{sec:init}
	Initially servers have no indexes or classification (and might not even have configurations), so we need to specify how the system state is initialized (say at time $0$) under {\DRA}.
	If servers do not have configurations, but have jobs in them, we initialize $\assgnk^{\ell}(0)=\schedk^\ell(0)$, i.e., the configuration of each server $\ell$ is set to its job placement. 
	If servers have configurations, we keep their existing configuration.
	Indexing among the servers of a configuration can be arbitrary. 
	We then run {\EPA} that performs classification and reassigns any 
    possibly empty servers.

	\begin{algorithm}[t]
		\caption{{\DRA:} Dynamic Reservation Algorithm}\label{alg:main3}
		\begin{algorithmic}[1]
			\Function{\textsc{CRA}}{$\mathbf{\hat{Y}}^L, \mathbf{X}^L$}
			\State $\mathbf{\hat{X}}^L \gets$ {\VPA} 
			$(\mathbf{\hat Y}^L)$. 
			\State Set rank of all servers to $J+1$.
            \State $I^\star \gets J$
			\For {$i = 1$ to {$J$}} 
			\Comment $J$ configurations found in {\VPA}
			\State $Z \gets 0$, \ $c \gets {X}^{L}_{\mathbf{k}[i]}$ 
            \Comment $c$ is the index of server
			\While {$Z < \hat{X}^{L}_{\mathbf{k}[i]}$}
			\State $Z \gets Z + 1$,\ $c \gets c - 1$ 
			\If {$c \le 0$}
            \If {$\emptysrv \neq \varnothing$} 
			\State Set rank of $\emptysrv$ to $i$.
			\State Reassign configuration of $\emptysrv$ to $\mathbf{k}[i]$.
            \Else
                \State $I^\star \gets \min(I^\star, i)$
			\EndIf
			\Else
			\State Set rank of $\ell_{\mathbf{k}[i], c}$ to $i$. 
			\EndIf
			\EndWhile
			\EndFor
			\EndFunction
		\end{algorithmic}
		\begin{algorithmic}[1]
			\Procedure{Arrival}{$j, t$} \Comment Type-$j$ arrival at time $t$
			\If {$\texttt{AG}_j \neq \varnothing$}
			\State Schedule job in $\texttt{AG}_j$.
			\State \textsc{CRA} 
			$(\mathbf{Y}^L(t) + g(L)\mathbf{1}, \mathbf{X}^L(t))$
			\Else
			\State Reject job.
			\EndIf
			\EndProcedure
		\end{algorithmic}
		\begin{algorithmic}[1]
			\Procedure{Departure}{$j, t$} \Comment Type-$j$ departure at time $t$
            \If {${\texttt{RG}}_j \neq \varnothing$ and the slot emptied is in Accept Group} 
			\State Migrate the job in ${\texttt{RG}}_j$ to the slot that emptied.
			\EndIf
            \State \textsc{CRA}
            $(\mathbf{Y}^L(t) + g(L)\mathbf{1}, \mathbf{X}^L(t))$
			\EndProcedure
		\end{algorithmic}
	\end{algorithm}

    \begin{remark}
    	Notice the duality of actions performed on arrivals and departures for any job type: 
        jobs are admitted/migrated to empty slots in servers of Accept Group, and depart/migrate from filled slots in servers of Reject Group. 
        The number of servers in Reject Group under our algorithm is at most one per configuration, i.e., at most $\gcnt$ servers (constant independent of $L$) which is negligible compared to the number of servers $L$, as $L \to \infty$. 
    	Further, job admissions and migrations are performed to slots which are already deployed in advance.
        The reservation factor $g(L)$ is critical for maintaining enough
        deployed slots in the maximum reward configurations for
        future demand.
    	
    	In contrast, a naive static reservation algorithm, that solves \dref {eq:opt1c} by replacing $\mathbf{\hat Y}$ with an estimate of workload, might require changing the configuration of a constant fraction of servers (the equivalent of Reject Group), as workload estimate changes. 
        This would result in preemptions (or migrations) in $O(L)$ interrupted servers.     
    	
    	Lastly, more accurate estimates of workload, if available, can be simply used in the input $\mathbf{\hat Y}$ to {\EPA}, and {\EPA} itself can be executed less regularly, depending on the complexity and convergence time tradeoff. 
    \end{remark}

The following theorem states the main result regarding {\DRA}.
\begin{theorem}\label{thm:main}
	Let $F^{\DRA}(L)$ be the expected reward under {\DRA} and $F^{\star}(L)$ be the  optimal expected reward in optimization \dref{eq:obj}. Then
	\ben
	\lim_{L\to \infty} \frac{F^{\DRA}(L)}{F^{\star}(L)} \geq \frac{1}{2}.
	\een
	Further, under the monotone greedy property (Definition~\ref{def:monotone greedy}),
	$$\lim_{L\to \infty}\frac{F^{DRA}(L)}{F^{\star}(L)} \geq 1-\frac{1}{e}.$$
\end{theorem}

\begin{remark}
	Note that we did not make any assumption on the value of $\bm{\rho}$, and Theorem~\ref{thm:main} holds for any $\bm{\rho}$.  
    Define
    \be \label{eq:loadregion}
    \Lambda= \Big\{\mathbf{y}: \mathbf{y}\leq \sum_{\mathbf{k}\in \mathcal{K}}x_{\mathbf{k}}\mathbf{k},\text{ for }  x_{\mathbf{k}}\geq 0,  \assgnk \in \mathcal{K}, \sum_{\mathbf{k}\in \mathcal{K}}x_{\mathbf{k}} = 1\Big\}.
    \ee
	Theorem~\ref{thm:main} holds even if $\bm{\rho}$ is outside $\Lambda$. 
    In this scenario, a nonzero fraction of traffic has to be rejected even by the optimal policy.
\end{remark}

The proof of Theorem~\ref{thm:main} is based on analysis of fluid limits and a suitable Lyapunov function to show convergence, as we do next in Sections~\ref{sec:analysis} and \ref{sec:converge}.

%% file: analysis13.tex
\section{Fluid Limits under {\DRA}}\label{sec:analysis}

We first define two useful variables, which are functions of the system state, and will be used in our convergence analysis.
\begin{definition}[Effective Number of Assigned Servers]\label{def:eff}
	The effective number of servers in configuration 
	$\assgnk$ is defined as
	\begin{equation}\label{eq:effectivedef}
	X^{L(e)}_{\assgnk}(t) := 
	\min(X^L_{\assgnk}(t), {\hat X}^{L}_{\assgnk}(t)).
	\end{equation}
	Note that $X^{L(e)}_{\assgnk}(t)={\hat{X}_{\assgnk}}^{L}(t)=0$ if $ \assgnk \notin \{\greedk^{(i)}, i=1, \ldots, \gcnt\}$.
	With a minor abuse of terminology, we say the
	servers in configuration $\assgnk$ with 
    indexes from $X^{L}_{\assgnk}(t) - X^{L(e)}_{\assgnk}(t) + 1$ to $X^{L}_{\assgnk}(t)$, 
    have effective configuration $\assgnk$.
\end{definition}

\begin{remark}
    Note that $X^{L(e)}_{\assgnk}(t)$ is \textit{independent} of the indexing of servers in configuration $\assgnk$.
   Also note if $\assgnk = \assgnk[j]$, where $\assgnk[j]$, $j \leq J$, is the $j$-th configuration returned by {\VPA} at time $t$, then in {\DRA}, servers with effective configuration $\assgnk[j]$ get rank $j$, 
   and servers without effective configuration have rank $J+1$. 
\end{remark}

\begin{definition}\label{def:rgi}
    Given an $i\leq \gcnt$, Reject Group servers can be divided as $\texttt{RG}=\overline{\texttt{RG}}(i) \cup \texttt{RG}(i)$.  
    The servers with index $1$ without effective configuration
    in $\greedk^{(\ell)}$, for $\ell=1,\ldots, i$, belong to $\overline{\texttt{RG}}(i)$, while the rest of servers of Reject Group belong to $\texttt{RG}(i)$.
\end{definition}

\subsection{Effective Slot Deficit: \texorpdfstring{$q$}{q} Process }\label{sec:qproc}
The job admission and configuration assignment under {\DRA} crucially depends on the $\mathbf{q}$ process defined below.
\begin{definition}\label{def:qdef}
	For  $i \in \{1, \cdots, \qcnt\}$, and $j \in \mathcal{J}$, we define 
	\begin{equation}\label{eq:q_def}
	q^L_{\greedk^{(i)}, j}(t) := \sum_{\ell=1}^i
	X_{\greedk^{(\ell)}}^{L(e)}(t) \greedki^{(\ell)}_j - Y^L_{j}(t) - g(L).
	\end{equation}
	Note that,  $\forall j \in \mathcal{J}$, $q^L_{\greedk^{(i_2)}, j}(t)\geq q^L_{\greedk^{(i_1)}, j}(t) $ if $i_2 \geq i_1$.  	
\end{definition}
In words, $q^L_{\greedk^{(i)}, j}(t)$ measures the difference between the total number of type-$j$ slots (filled or empty) in servers that have effective configurations in the set $\{\greedk^{(\ell)}: \ell \leq i\}$ (see Definition~\ref{def:eff}),  
and the number of type-$j$ jobs in the system $Y^L_{j}(t)$ and $g(L)$ type-$j$ reservation slots.

Note that {\DRA} (specifically {\VPA}) will stop assigning configurations that have type-$j$ slots, once $Y^L_{j}(t)+g(L)$ slots can be accommodated in 
servers with effective configuration in $\{\greedk^{(\ell)}, \ell \leq i\}$. 
Since slots are created per server basis, by assigning configurations which each has at most $\Kmax$ slots,
we have     
${q}^L_{\greedk^{(i)}, j}(t) < \Kmax$.

To gain more insight, note that when $q^L_{\greedk^{(i)}, j}(t) \ge 0$ for an $i \in \{1, \cdots, \qcnt\}$, it means type-$j$ jobs have enough reservation. 
When it is negative, it indicates the deficit of slots in servers 
with effective configuration $\{\greedk^{(\ell)}, \ell \leq i\}$.  
When $q^L_{\greedk^{(i)}, j}(t) > -g(L) + J \Kmax$, for an $i \in \{1, \cdots, \gcnt\}$, a type-$j$ arrival at time $t$ will certainly find a valid empty slot ($\texttt{AG}_j \neq \varnothing$) and will be admitted. 
This is because the number of empty slots of type $j$ in Reject Group servers with any effective configuration is less than $J\Kmax$.

The $\mathbf{q}$ process also determines the configuration  
assigned by {\EPA} to an empty server $\emptysrv$ chosen for reassignment. 
The configuration would be $\greedk^{(i)}$, $i < \gcnt$, if:
\begin{subequations}
\begin{align}
\max_{j: \greedki^{(\ell)}_j > 0}
q^L_{\greedk^{(\ell)},j}(t) \ge 0,& \ \forall \ell \leq i-1, \label{eq:qfillall}
\\
\max_{j: \greedki^{(i)}_j > 0}
q^L_{\greedk^{(i)},j}(t) < 0 & \label{eq:qfill}.
\end{align}
\end{subequations}
This also implies that if only \dref{eq:qfill} holds,
the server would be assigned to one of the configurations
$\greedk^{(\ell)}$, $\ell=1, \ldots, i$.

%%%%%%%%%%%%%%%%%%%%%%%%%%%%%%%%%%%%%%
\subsection{Existence of Fluid Limits}

We define the scaled (normalized with $L$) processes 
$\mathbf{x}^{L(e)}(t)$, $\mathbf{y}^{L}(t)$, %and $\mathbf{{\tilde y}}^{L(\greedk^{(i)})}(t)$, $i \in \{1, \ldots, \gcnt\}$, 
as follows. 
For $i \in \{1, \ldots, \gcnt\}$, and $j \in \mathcal{J}$,
\begin{flalign*}
x^{L(e)}_{\greedk^{(i)}}(t) = \frac{1}{L}{X_{\greedk^{(i)}}^{L(e)}(t)},\quad %\\ 
y^{L}_j=\frac{1}{L} Y^{L}_j(t),
\end{flalign*}
and define $z^{L}(t):=(\mathbf{x}^{L(e)}(t),\mathbf{y}^{L}(t))$. 
We also define the space
\ben \label{eq:Zregion}
\mathcal{Z} = \Big \{(\mathbf{x}^{(e)},\mathbf{y}): \mathbf{y} \in \Lambda,\  x^{(e)}_{\greedk^{(i)}} \geq 0,
\sum_{i=1}^\gcnt x^{(e)}_{\greedk^{(i)}}\leq 1,
\ \sum_{i=1}^\gcnt x^{(e)}_{\greedk^{(i)}} {\greedk^{(i)}} \leq \mathbf{y}
\Big \}.
\een
where $\Lambda$ was defined in \dref{eq:loadregion}.
%%%%%%%%%%%%%%%%%%%%%

\begin{proposition}\label{prop:fl_limits}
	Consider a sequence of systems with increasing $L$, and initializations
    $\mathbf{z}^{L}(0)=(\mathbf{x}^{L(e)}(0),\mathbf{y}^{L}(0)) \in \mathcal{Z}$, as $L \to \infty$. 
    Then there is a subsequence of $L$ such that $\mathbf{x}^{L(e)}(t)\to \mathbf{x}^{(e)}(t)$, $\mathbf{y}^{L}(t) \to \mathbf{y}(t)$,  
    along the subsequence.
	Any limit  $\mathbf{z}(t):=(\mathbf{x}^{(e)}(t), \mathbf{y}(t))$, $t \geq 0$, is called a fluid limit sample path.
    The convergence is almost surely u.o.c. (uniformly over compact time intervals) and the fluid limit sample paths are Lipschitz continuous.
\end{proposition}
\begin{proof} %[Proof of Proposition~\ref{prop:fl_limits}]
	Proof is standard, and can be found in Appendix~\ref{sec:fl_limits}.
\end{proof}

%%%%%%%%%%%%%%%%%%%%%%%%%%%%%%%%%%%%%%%%%%%%%%%%%%%%%%%%%%%%
\subsection{Description of Fluid Limits}\label{sec:informal}
We provide an informal description of fluid limit equations here. 
The formal definitions and proofs can be found in Appendix~\ref{proof:prop:properties}.

The properties of the fluid limit processes crucially depend on the $\mathbf{q}$ process (Definition~\ref{def:qdef}). 
First note that, from \dref{eq:q_def} and since ${q}^L_{\greedk^{(i)}, j}(t) < \Kmax$, it follows that 
\be\label{eq:yjA1}
\sum_{\ell=1}^{{\gcnt}} \greedki^{(\ell)}_j x^{(e)}_{\greedk^{(\ell)}}(t) \leq y_j(t),\
\forall j \in \calJ.
\ee

Let $\emptyx(t)$ be the fraction of servers which are empty and of rank $J+1$ 
at the fluid limit. 
When $\emptyx(t)>0$, then {\EPA} always finds empty rank $J+1$ servers available for reassignment. 
In this case, every job type will have enough empty slots, and all the arrivals will be admitted, i.e., we can find an $\epsilon$ sufficiently small such that for every job type $j$ and every time $\tau \in [t, t + \epsilon)$, $q^L_{\greedk^{(\qcnt)}, j}(\tau) \ge 0$. 
Hence, noting that at the fluid limit type-$j$ jobs arrive at rate $\lambda_j$ and existing type-$j$ jobs depart at rate $y_j(t)\mu_j$, 
\begin{subequations}
\begin{align}
&dy_j(t)/{dt} = \lambda_j - y_j(t)\mu_j,\ \forall j \in \mathcal{J},\\
&\sum_{\ell=1}^{\qcnt} \greedki^{(\ell)}_j x^{(e)}_{\greedk^{(\ell)}}(t) = y_j(t) \label{eq:noncriticalequality}, 
\end{align}
\end{subequations}
where Equality~\dref{eq:noncriticalequality} is based on (\ref{eq:q_def}) and due to the fact that $\lim_{L\to \infty} \frac{1}{L} q^L_{\greedk^{(\qcnt)}, j}(t)=0$ in this case.

A major difficulty in describing fluid limits happens on the \textit{boundary} $\emptyx(t) = 0$, i.e., when there are not always empty rank $J+1$ servers
available for reassignment when {\EPA} runs.
In this case, let $i^\star(t)$ be the \textit{largest} index in $\{1,\cdots, \gcnt-1\}$ such that for every $i \leq i^\star(t)$,
\begin{equation}\label{eq:yjA}
\sum_{\ell=1}^{i} \greedki^{(\ell)}_{j_i} x^{(e)}_{\greedk^{(\ell)}}(t) = y_{j_i}(t), \mbox{ for some } j_i \in \calJ,
\end{equation}
with the convention that $i^\star(t)=0$ if \dref{eq:yjA} does not hold for $i=1$.
%any $i \in \{1,\cdots, \gcnt\}$.
If $i^\star(t) < \gcnt-1$, then for $L$ sufficiently large, and 
every time $\tau \in [t, t + \epsilon)$ for $\epsilon$ sufficiently small,
\begin{equation}\label{eq:ql0_}
\max_{j : \greedki^{(i^\star(t)+1)}_j > 0} 
q^L_{\greedk^{(i^\star(t)+1)}, j}(\tau) < 0.
\end{equation}

Based on Definition~\ref{def:rgi}, servers in $\overline{\texttt{RG}}(i^\star(t)+1)$ have higher ranks compared to those in $\texttt{RG}(i^\star(t)+1)$, so any migrations by {\DRA} will take place from $\overline{\texttt{RG}}(i^\star(t)+1)$ first.
We can then show that servers of $\overline{\texttt{RG}}(i^\star(t)+1)$
empty at the fluid scale, at a rate of at least 
\be\label{eq:rank0rate}
& \frac{\mumin}{J \Kmax \cfcnt^2} \Big(
1 - \sum_{\ell=1}^{i^\star(t)+1} x^{(e)}_{\greedk^{(\ell)}}(t) \Big),
\ee
where $\mumin := \min_{j \in \calJ} \mu_j$ (see Lemma~\ref{lem:emptyrate} in Appendix~\ref{proof:bound}).

The algorithm will reassign any such server that empties to one of configurations $\greedk^{(\ell)}$ for $\ell=1, \ldots, i^\star(t)+1$.
If instead $i^\star(t)=\gcnt-1$, then it is uncertain whether servers that
empty need to be reassigned to a new configuration or not, 
depending on whether $\max_{j \in \calJ: \greedki^{(i)}_j > 0} q^L_{\greedk^{(i)}, j}(\tau) < 0$, 
for some $i < \gcnt$ at time $\tau \in [t, t + \epsilon)$.

Hence, what we see is that, if $\emptyx(t) = 0$, when a server gets empty, it can be assigned to one of the configurations $\greedk^{(i)}$, $i=1, \ldots, i^\star(t)+1$. 
Exact characterization of these assignment rates, however, 
is not easy as they depend on values of processes ${q}^L_{\greedk^{(i)}, j}(\tau)$, $i\in \{1, \ldots, i^\star(t)\}$, $j \in \mathcal{J}$, which evolve at a much faster time scale than the scaled 
processes $\mathbf{x}^{L(e)}$ and $\mathbf{y}^L$. 
By the continuity of the fluid limit sample paths, at any regular time $t$, we can choose $\epsilon$ small enough such that for all $\tau \in [t, t+\epsilon)$, $\mathbf{y}(\tau)$, and $\mathbf{x}^{(e)}(\tau)$ are approximately constant and equal to $\mathbf{y}(t)$ and $\mathbf{x}^{(e)}(t)$, respectively (their actual change being of order $\epsilon$). 
However, over the same interval, the $\mathbf{q}^L$ process makes $O(L)$ transitions and its elements can change in the range $[-L\Kmax, \Kmax]$. 
This phenomenon is known as \textit{separation of time scales} and has been also observed in other systems, e.g.~\cite{hunt1997optimization,hunt1994large}.

%%% #taudef
To further analyze fluid limits in our setting,
we divide the interval $[t, t+\epsilon)$ into smaller intervals of length 
$\omega(\log L/L)$, 
and infer properties for the fluid limits over $[t, t+\epsilon)$ 
based on averaging the behavior of scaled processes over 
these smaller intervals, as $L \to \infty$, and $\epsilon \to 0$. 
To this end, we first make a few definitions.

Since the rate of change of any of the processes $x^{L(e)}_{\greedk^{(i)}}(\tau)$ and $y^{L}_{j}(\tau)$ over a subinterval is of interest, we give it a special name below.
\begin{definition}[Local Derivatives] Given an interval $[\tau_a, \tau_b)$, we define the ``local derivatives'' of the scaled processes as  
	\begin{eqnarray}\label{eq:local derivative}
&	{\nabla} x^{L(e)}_{\greedk^{(i)}}[\tau_a, \tau_b) := 
	\frac{x^{L(e)}_{\greedk^{(i)}}(\tau_b)
		-x^{L(e)}_{\greedk^{(i)}}(\tau_a)}{\tau_b-\tau_a} 
	\quad i=1, \ldots, \gcnt \\
&	\nabla {y^L_j}[\tau_a,\tau_b) := 
	\frac{y^{L}_{j}(\tau_b)-y^{L}_{j}(\tau_a)}{\tau_b-\tau_a}
	\quad j \in \mathcal{J}.
	\end{eqnarray}
\end{definition}
\begin{definition}\label{def:calJ}
	For any $i \leq \gtcnt-1$, we define a set
	\be\label{eq:defcalJ}
&	\exactJ{i}  := \{j \in \calJ: \greedki^{(i)}_j > 0, \
	\sum_{\ell=1}^{i} \greedki^{(\ell)}_j \globlt{x}^{(g)}_{\greedk^{(\ell)}} = 
	\rho_j\}.
	\ee
\end{definition}

\begin{definition}\label{def:gprime}
	For given positive constants $\alpha_i$, $i=1, \ldots, \gtcnt-1$,
	we define $\gtcntp{\alpha}(t)$ to be the largest index 
	at time $t$ such that $\gtcntp{\alpha}(t) \le \min(i^\star(t), \gtcnt-1)$ and
	\begin{equation}\label{eq:alpha_prop}
	\forall i \in [1, \ldots, \gtcntp{\alpha}(t)]: \
	x^{(g)}_{\greedk^{(i)}} - x^{(e)}_{\greedk^{(i)}}(t) < \alpha_{i}.
	\end{equation}
\end{definition}

\subsubsection{Subinterval construction} \label{sec:subintervals}
%The subintervals will be constructed in two steps and
%	in the first one we divide them into equal sizes. 
We first define a function $f(L)$ below, which will control the length of subintervals.	
\begin{definition}\label{def:sublen}
    The function $\subL$ is defined as
	\be \label{eq:funf}
	&	\subL := \frac{\sqrt{g(L)\log(L)}}{L}
	\ee
    where $g(L)$ is the reservation factor as defined in {\DRA}.
\end{definition}

	We divide $[t, t+\epsilon)$ into smaller intervals $[\tau_{n}, \tau_{n+1})$, such that
	\begin{equation}\label{eq:taudef}
	\begin{aligned}
	\tau_0 = t, \ \tau_n = \tau_{n-1} + D_{L, \epsilon},\ n=1, \ldots, N_{L}, 
	\end{aligned}
	\end{equation}
	where $N_{L} = \ceil{1/\subL}$ is the number of such smaller intervals, and
	$D_{L, \epsilon} = \frac{\epsilon}{N_{L}}$ is the length of each one. 
    We then further divide each $[\tau_n, \tau_{n+1})$ into a \textit{constant}  number $M_n$ of subintervals $[\tau^{(m-1)}_n, \tau^{(m)})$, $m=1, \ldots, \maxm{n}$, $\tau^{(0)}_n = \tau_n$, $\tau^{(\maxm{n})}_n = \tau_{n+1}$. 
    For every $n$, the sequence of stopping times $\tau^{(m)}_n$ is recursively generated as follows:
	   
Each time $\tau^{(m)}_n$ is associated with a \textit{driving} set of job indexes $\subind[m]$, with the initialization $\subind[0] =\varnothing$ and $\tau^{(0)}_n = \tau_n$.
    Suppose $\subind[m-1] := \{ \ji{i}: i=1, \ldots, G_{m-1} \}$ at time $\tau^{(m-1)}_n$, where $\ji{i}\in \exactJ{i}$ (Definition~\ref{def:calJ}). 
    Define $h^{\subind[m-1], (\ell)}(t)$, $\ell=1, \ldots, G_{m-1}$, to be the (unique) solution to the following system of equations
	\begin{equation}\label{eq:h_system}
	\sum_{\ell=1}^i \greedki^{(\ell)}_{\ji{i}} h^{\subind[m-1], (\ell)}(t) = 
	\lambda_{\ji{i}} - \mu_{\ji{i}} y_{\ji{i}}(t),
	\quad i=1, \ldots, G_{m-1}.
	\end{equation}
	The next $\tau^{(m)}_n$ is the \textit{earliest} time 
	$\tau \in [\tau^{(m-1)}_n, \tau_{n+1})$ such that
	$q^L_{\greedk^{(G_m)}, j}(\tau) \ge 0$
	for some $G_m \le \min(G_{m-1}+1, \gtcntp{\alpha}(t))$
	and some $j \in \exactJ{G_m}$.
	Further, if $G_m \le G_{m-1}$, we additionally require that 
	\be\label{eq:h_cond}
    \sum_{\ell=1}^{G_m}
	\greedki^{(\ell)}_{j} h^{\subind[m-1], (\ell)}(t) > \lambda_{j} - \mu_{j} y_{j}(t). 
	\ee
    At such a time $\tau$, we set $\tau_n^{(m)}=\tau$, and the driving index set is set to 
    \begin{equation}
    \subind[m] := \{ \jip{i}: i=1, \ldots, G_m \},
    \end{equation}
    where $\jip{i} = \ji{i}$ for $i=1, \ldots, G_m-1$,
    and $j^\prime_{G_m} = j$. 
    Also, $h^{\subind[m], (\ell)}(t)$, $\ell=1, \ldots, G_m$, is set to the solution of the system of equations \dref{eq:h_system} for the set $\subind[m]$. 
    If no time $\tau \in [\tau^{(m-1)}_n, \tau_{n+1})$
    satisfies the given conditions, then $m = \maxm{n}$
    and $\tau^{(\maxm{n})}_n = \tau_{n+1}$. 

	The importance of quantities $h^{\subind[m], (i)}(t)$, $i=1,\ldots, G_m$, will become evident later where we will show (see Lemma~\ref{lem:hrec} in Appendix) that 
	\begin{equation}
	{\nabla} x^{L(e)}_{\greedk^{(i)}}[\tau^{(m)}_n, \tau^{(m+1)}_n)
	= h^{\subind, (i)}(t) + \frac{o(\subL)}{\tau^{(m+1)}_n - \tau^{(m)}_n}.
	\end{equation}
    Hence, roughly, (\ref{eq:h_system}) gives the values of local derivatives, 
    while when (\ref{eq:h_cond}) occurs, the values of local derivatives change.

Note that the number of stopping times $\maxm{n}$ 
in any interval $[\tau_n, \tau_{n+1})$ is bounded. 
This is because the number of different driving sets $\subind[m]$ 
is finite and no set may appear twice in that sequence,
since the comparison (\ref{eq:h_cond})
induces a total ordering between the sets. 
Considering all possible driving set of indexes
that may appear in the sequence, we have $\maxm{n} \leq 1 + \sum_{i=1}^{\gtcntp{\alpha}(t)} \prod_{\ell=1}^i |\exactJ{i}| < \infty$.

\subsubsection{Properties of fluid limits over subintervals}	
    Given an $\epsilon_\rho > 0$, we first define the set of fluid limit states  
	\begin{equation}\label{eq:epsrho}
	\setConv[\epsilon_\rho] := \{(\mathbf{x}^{(e)},\mathbf{y}): 
    \mathbf{y} \le \bm{\rho} + \epsilon_\rho \} \cap \mathcal{Z}.
	\end{equation}
    The following lemma states the invariant property of  $\setConv[\epsilon_\rho]$.
	\begin{lemma}\label{lem:S}
	If $(\mathbf{x}^{(e)}(0),\mathbf{y}(0)) \in \mathcal{Z}$, 
    then for any $\epsilon_\rho > 0$, there is a time $T_{\epsilon_\rho} > 0$ such that for all $t \geq T_{\epsilon_\rho}$, 
	$(\mathbf{x}^{(e)}(t),\mathbf{y}(t)) \in \setConv[\epsilon_\rho]$. 
    Further, convergence is uniform over all initial states in $\mathcal{Z}$.
\end{lemma}
\begin{proof} % [\textbf{Proof of Lemma~\ref{lem:S}}]
	See Appendix~\ref{proof:lem:S}.
\end{proof}
The following proposition states the behavior of scaled processes over the subintervals. 
\begin{proposition}\label{prop:properties}
	For every $m \in \{0, \ldots, \maxm{n}-1\}$, let
	$\subind[m] = \{\ji{i}: i=1, \ldots, G_m\}$ be the index set
	corresponding to time $\tau^{(m)}_n$, and $\ell_m := G_m + 1$.
	Then we can choose $\alpha_i$s in Definition~\ref{def:gprime}, and $\epsilon_\rho$ 
	in \dref{eq:epsrho} sufficiently small, such that, for any regular time $t \geq T_{\epsilon_\rho}$, with probability at least $1 - o(L^{-2})$, 
	all the following properties hold: 
	\begin{enumerate}[leftmargin=*,label=P.\arabic*.]
		\item \label{pr:exact}
		For every $i \in \{1, \ldots, \ell_m-1\}$,
		\be\label{eq:exact}
		\sum_{\ell=1}^{i} \greedki^{(\ell)}_\ji{i}
		\nabla {x}^{L(e)}_{\greedk^{(\ell)}}[\tau^{(m)}_n, \tau^{(m+1)}_n] = 
		\lambda_\ji{i} - \mu_\ji{i} y_\ji{i}(t) 
		+ \frac{o(\subL)}{\tau^{(m+1)}_n - \tau^{(m)}_n}
		\ee
		\item\label{pr:bound1}
		If $\ell_m < \gtcnt$, 
		\be\label{eq:bound1}
		\nabla {x}^{L(e)}_{\greedk^{(\ell_m)}}[\tau^{(m)}_{n}, \tau^{(m+1)}_n] > 
		\frac{\mumin}{2} \alpha_{\ell_m}
		+ \frac{o(\subL)}{\tau^{(m+1)}_n - \tau^{(m)}_n}
		\ee
		\item\label{pr:bound2}
		If $\ell_m = \gtcnt$, 
		\begin{flalign}
		& \nabla {x}^{L(e)}_{\greedk^{(\gtcnt)}}
		[\tau^{(m)}_{n}, \tau^{(m+1)}_n] > 
		 \frac{o(\subL)}{\tau^{(m+1)}_n - \tau^{(m)}_n}+ \nonumber \\
		& \min 
		\left\{\frac{\mumin}{J \Kmax \cfcnt^2}
		\Big(1 - \sum_{i=1}^{\gtcnt}
			{x}^{(e)}_{\greedk^{(i)}}(t)\Big)
		- \sum_{i=1}^{\gtcnt-1} \left(
		\nabla {x}^{L(e)}_{\greedk^{(i)}}
		[\tau^{(m)}_{n}, \tau^{(m+1)}_n]\right)^+, \right. \nonumber \\
		& \left. \min_{j: \greedki^{(\gtcnt)}_j > 0}
		%\left(
		\frac{\lambda_j - \mu_j y_j(t) 
			- \sum_{i=1}^{\gtcnt-1} \greedki^{(i)}_j 
			\nabla {x}^{L(e)}_{\greedk^{(i)}}
			[\tau^{(m)}_{n}, \tau^{(m+1)}_n]}{\greedki^{(\gtcnt)}_j} %\right)
		\right\} \label{eq:bound2}
		\end{flalign}
	\end{enumerate}
\end{proposition}
In words, (\ref{pr:exact}) states that, roughly, for any $i < \ell_m$, 
there is a job type $\ji{i}$ such that
each of the effective number of servers with configurations 
$\{\greedk^{(\ell)}$ for $\ell=1, \ldots, i\}$ changes at a rate 
that can accommodate exactly additional type-$\ji{i}$ arrivals.

(\ref{pr:bound1}) states that effective number of servers with
configuration $\greedk^{(\ell_m)}$ increases by an amount proportional to $\alpha_{\ell_m}$.
This implies that the rate at which ${x}^{(e)}_{\greedk^{(\ell_m)}}(t)$ 
converges to the global greedy solution is lower bounded by a constant 
independent of the system state.

(\ref{pr:bound2}) describes the change in the effective number of servers in $\greedk^{(\gtcnt)}$, the last configuration of the global greedy solution.
The change either satisfies	the same condition as 
(\ref{pr:exact}) or it is bounded by the difference of
how fast Reject Group servers empty (based on \dref{eq:rank0rate} 
for $i^\star(t)=\gtcnt-1$) and at what rate they are assigned to configurations $\greedk^{(i)}$ for $i < \gtcnt$.

\begin{proof}[Proof of Proposition~\ref{prop:properties}]
The proof, including all supporting Lemmas,
is provided in Appendix~\ref{proof:prop:properties}.
\end{proof}

%% file: main_results2.tex
%%%%%%%%%%%%%%%%%%%%%%
\section{Convergence Analysis}\label{sec:converge}
We show that the fluid limit of the effective configuration process $\mathbf{x}^{(e)}(t)$ (which is a lower bound on the number of servers in each configuration) converges to the global greedy solution $\mathbf{x}^{(g)}$.
\begin{theorem}\label{thm:convergence}
	Consider the fluid limits of the system under {\DRA},
	under any workload $\bm{\rho}$, and any initial state
	$\statez(0) \in \mathcal{Z}$.
	Then
	\begin{equation}
	\lim_{t \to \infty} x^{(e)}_{\bk}(t) =
	\globlt{x}^{(g)}_{\bk}, \quad \bk \in \mathcal{K}^{(g)}.
	\end{equation}
\end{theorem}
\begin{proof}
	Recall that $\mathbf{z}(t) =(\mathbf{x}^{(e)}(t),\mathbf{y}(t)).$
	We want to show that $\mathbf{z}(t)$ converges
	to a point in the set $\setOpt$ defined as
	\begin{equation}\label{eq:substar}
	\setOpt := \{ \mathbf{z}:=(\mathbf{x}^{(e)},\mathbf{y}) \in \setConv[\epsilon_\rho]:
	x^{(e)}_{\assgnk} = \globlt{x}^{(g)}_{\assgnk},\ \assgnk \in \mathcal{K}^{(g)} \}.
	\end{equation}
	where $\setConv[\epsilon_\rho]$ was defined in \dref{eq:epsrho}.

	To show convergence, we use a Lyapunov function of the form
	\begin{equation}\label{eq:Vdef}
	V(\mathbf{z}(t)) := \sum_{i=1}^{\gtcnt} \LC_{i}
	\left(\globlt{x}^{(g)}_{\greedk^{(i)}} -
	{x}^{(e)}_{\greedk^{(i)}}(t) \right)
	+ \LC \sum_{j=1}^J (y_j(t) - \rho_j)^+,
	\end{equation}
where $\LC$ and $\LC_i$, $i \in \{1, \ldots, \gtcnt\}$, are positive constants satisfying
\begin{equation}\label{eq:Vcond}
	\LC > 4\LC_1, \ \ \LC_i > \xi \LC_{i+1},\ i=1, \ldots, \gtcnt-1,
\end{equation}
for a $\LC_{\gtcnt} > 0$, and a sufficiently large constant $\xi > 2\Kmax + 1$.

The constants $\epsilon_\rho$ and $\xi$ will be chosen carefully to ensure the conditions of LaSalle's invariance principle~\cite{LaSalle60, Cohen17} hold for any $\mathbf{z} \in \setConv[\epsilon_\rho]$, i.e.,
\begin{itemize}
	\item [(i)] For any $\mathbf{z} \in \setConv[\epsilon_\rho]$, we have
	$V(\mathbf{z}) \ge 0$ and
	$V(\mathbf{z}) = 0$ if and only if $\mathbf{z} \in \setOpt$,
\item [(ii)] For any $\mathbf{z}(t) \in \setConv[\epsilon_\rho] \setminus \setOpt$,
	${{\rm d}V(\mathbf{z}(t))}/{{\rm d}t} < 0$, almost surely.
\end{itemize}
	These conditions together with Lemma~\ref{lem:S} will then imply that the
	limit points of trajectory $\mathbf{z}(t)$ are in $\setOpt$.

	We state each condition as a Proposition followed by its proof.
	\begin{proposition}\label{prop:VS_}
		Consider $V(\statez)$ in \dref{eq:Vdef}, with coefficients in \dref{eq:Vcond}, for any $\xi > (2\Kmax+1)$, and
		$\epsilon_\rho > 0$.
		Then we have $V(\statez) \ge 0$ for any
		$\statez \in \setConv[\epsilon_\rho]$, and
		$V(\statez) = 0$ if and only if $\statez \in \setOpt$.
	\end{proposition}
	\begin{subproof}[Proof of Proposition~\ref{prop:VS_}]
		Consider the following maximization problem over $\bm{\eta} \in \mathds{R}^{\gtcnt}, \bm{\theta} \in \mathds{R}^{J}$, where $\eta_i$ corresponds to  ${x}^{(e)}_{\greedk^{(i)}}(t)$ and
		$\theta_j$ corresponds to $(y_j(t) - \rho_j)^+$ in \dref{eq:Vdef},
		\begin{maxi!}{\bm{\eta}, \bm{\theta}}
			{\sum_{i=1}^{\gtcnt} \LC_i \eta_i - \sum_{j=1}^J \LC \theta_j} {\label{eq:optZ_}}{}
			\label{eq:optZ_o_}
			\addConstraint{}{\textstyle \sum_{i=1}^{\gtcnt} \eta_i \le 1,}\label{eq:optZ_s_}
			\addConstraint{}{\textstyle \sum_{i=1}^{\gtcnt}
				\greedki^{(i)}_{j} \eta_i - \theta_j \le \rho_j,}
			{\quad j=1, \ldots, J}\label{eq:optZ_wl_}
			\addConstraint{}{\theta_j\leq \epsilon_\rho,} {\quad j=1, \ldots, J}
			\label{eq:optZ_eps_}
			\addConstraint{}{\eta_i\geq 0,} {\quad i=1, \ldots, \gtcnt} \label{eq:optZ_ge0_}
			\addConstraint{}{\theta_j\geq 0,} {\quad j=1, \ldots, J}
			\label{eq:optZ_ge0b_}.
		\end{maxi!}

		To prove the proposition, it is enough to show that the assignment $(\bm{\eta}^{(g)}, \bm{\theta}^{(g)})$ that corresponds to the global greedy solution $\mathbf{\globlt{x}}^{(g)}$ is the unique maximizer of the above LP.
		This assignment is
		\begin{equation}\label{eq:eta_sol}
		\begin{aligned}
		&\eta_{i}^{(g)} = \globlt{x}^{(g)}_{\bk^{(i)}},
		\ i=1, \ldots, \gtcnt, \\
		&\theta_j^{(g)} = 0, \quad j = 1, \ldots, J.
		\end{aligned}
		\end{equation}
		First note that \dref{eq:eta_sol} is a basic feasible solution for LP (\ref{eq:optZ_}), i.e., it is a corner point of the LP's Polytope, since it is on the boundary of $\gtcnt +J$ independent inequalities (equal to the number of variables).

		To show that (\ref{eq:eta_sol}) is the ``unique maximizer'', we need to verify
		that every neighboring corner point has lower objective value, and to do this, it suffices to verify that by moving along any valid direction within the Polytope, starting from assignment (\ref{eq:eta_sol}), the objective value is reduced.
		This proves that point (\ref{eq:eta_sol}) is locally optimal, which implies it is also global optimal, since the optimization
		is LP (and convex) \cite{Boyd2004}. 	In the rest of the proof, we use $\maptg{j}$ to be the mapping in
		Definition~\ref{def:mapgreedy} for $j=1, \ldots, \ggi$, and
		$\perm{j}$ to be the permutation of indexes $\{1, \ldots, J\}$ as defined in Proposition~\ref{prop:index}.

		We define $\Delta \eta_i := \eta_i^\prime - \eta_i^{(g)}$
		for $i \in \{1, \ldots, \gtcnt\}$, and
		$\Delta \theta_j := \theta_j^\prime$
		for  $j \in \{1, \ldots, J\}$, where
		$\eta_i^\prime$ and $\theta_j^\prime$ are the values of
		a feasible point.
		We prove that the change in objective is negative considering only one positive
		$\Delta \eta_i$ for some $i \in \{1, \ldots, \gtcnt\}
		\setminus \{\maptg{j}: j=1, \ldots, \ggi\}$, while the other $\Delta \eta_i$s in this set
		are $0$, and constraints (\ref{eq:optZ_s_})--(\ref{eq:optZ_ge0b_})
		are not violated.
		This suffices because any feasible point can be constructed
		as a convex summation of the changes $\Delta \eta_i$ and
		if individual changes reduce objective, their convex sum
		will reduce the objective too.

		Suppose $i^\star \in \{1, \ldots, \gtcnt\} \setminus
		\{\maptg{j}: j=1, \ldots, \ggi\}$ is the index for which
		$\Delta \eta_{i^\star} > 0$. A feasible point will necessarily
		satisfy the following set of equations, which correspond to
		$\gtcnt + J$ constraints
		(specifically, \dref{eq:optZ_s_}, \dref{eq:optZ_wl_} for
		$j \in \{\perm{j^\prime}: j^\prime = 1, \ldots, \ggi\}$,
		and \dref{eq:optZ_ge0b_} for
		$j=1, \ldots, J$) which held as
		equalities at point \dref{eq:eta_sol},
		\begin{equation}\label{eq:sys}
		\begin{aligned}
		&-\Delta\theta_j \le 0, \quad j=1, \ldots, J, \\
		&-\Delta \eta_{i^\star} < 0;\ \Delta \eta_{i} = 0, i \neq i^\star,
		i \in \{1, \ldots, \gtcnt\} \setminus \{\maptg{j}: j=1, \ldots, \ggi\}, \\
		&\greedki^{(i^\star)}_\perm{j} \Delta \eta_{i^\star} +
		\textstyle \sum_{\ell=1}^j \greedki^{(\maptg{\ell})}_\perm{j}
		\Delta \eta_\maptg{\ell} - \Delta\theta_\perm{j} \le 0,
		\quad j=1, \ldots, \ggi-1, \\
		&\textstyle  \Delta \eta_{i^\star} + \sum_{j=1}^{\ggi} \Delta \eta_\maptg{j} \le 0.
		\end{aligned}
		\end{equation}
		Notice that the conditions \dref{eq:sys} are not necessarily sufficient so
		even if all of them are satisfied the resulting point
		may be infeasible. Nevertheless, we prove that in any case the objective function will be reduced.
		The change in value of objective function is given by
		\be\label{eq:dobj}
		& \Delta F:= \sum_{\ell=1}^\ggi \LC_{\maptg{\ell}} \Delta \eta_{\maptg{\ell}} +
			\LC_{i^\star} \Delta \eta_{i^\star} - \LC \sum_{j=1}^J \Delta\theta_j.
		\ee

		Given the conditions (\ref{eq:sys}), we show (\ref{eq:dobj}) will be
		negative by finding constants $\beta > 0$,
		$\beta_j > 0$, $j=1, \ldots, \ggi$, and
		$\gamma_j > 0$, $j=1, \ldots, J$, such that
		\be \label{eq:objcase2}
		\Delta F &= & \textstyle  \beta(- \Delta \eta_{i^\star}) + \sum_{j=1}^{\ggi-1} \beta_j \left(
		\greedki^{(i^\star)}_\perm{j} \Delta \eta_{i^\star} +
		\sum_{\ell=1}^j \greedki^{(\maptg{\ell})}_{\perm{j}}
		\Delta \eta_\maptg{\ell} - \Delta \theta_\perm{j}
		\right) \nonumber
		\\
		& & \textstyle  + \beta_\ggi
		\left(\Delta \eta_{i^\star} +
		\sum_{\ell=1}^\ggi \Delta \eta_\maptg{\ell}
		\right)
		+ \sum_{j=1}^J \gamma_j \left(- \Delta \theta_j\right).
		\ee
		It is not difficult to show by matching the coefficients of
		\dref{eq:dobj} and \dref{eq:objcase2} that
		the values of $\beta$, $\beta_j$,
		for $j=1, \ldots, \ggi$ and $\gamma_j$ for $j=1, \ldots, J$,
		are strictly positive for the choice of $\LC$ and $\LC_i$'s in the
		proposition's statement. The details can be found in Appendix~\ref{prf:convergence}.
	\end{subproof}
    \begin{proposition}\label{prop:dV}
        For function $V(\statez)$, as defined in~(\ref{eq:Vdef}) and \dref{eq:Vcond}, there is a constant $\xi > 2\Kmax + 1$, such that
        if $\statez(t) \in \setConv[\epsilon_\rho] \setminus \setOpt$, then $\frac{\rm d}{{\rm d}t} V(\statez(t))<0$.
    \end{proposition}
    To prove Proposition~\ref{prop:dV}, we first prove the following 
    lemma for the local derivatives over subintervals $[\tau_n, \tau_{n+1})$ defined in Section~\ref{sec:subintervals}.

	\begin{lemma}\label{lem:Dx}
    Consider the Lyapunov function $V(\statez)$ defined in~(\ref{eq:Vdef}). 
    We can choose the constant $\xi > 2\Kmax + 1$ sufficiently large such that the following holds. 
    If at a regular time $t$, $V(\statez(t)) > \epsilon_V$, for some $\epsilon_V > 0$,  
    then there is a $\delta(\epsilon_V) > 0$ such that for any $n \in \{0, \ldots, N_{L}-1\}$,
		\begin{equation}\label{eq:rX}
		\sum_{i=1}^{\gtcnt} \LC_i
		\nabla x^{L(e)}_{\greedk^{(i)}}[\tau_{n}, \tau_{n+1}]
		> \delta(\epsilon_V) + \sum_{j=1}^J \LC \frac{\rm d}{{\rm d}t} (y_j(t) - \rho_j)^+
		+ o(1),
		\end{equation}
		with probability greater than $1 - o(L^{-2})$
	\end{lemma}
	\begin{subproof}[{Proof of Lemma~\ref{lem:Dx}}]
		The proof of Lemma~\ref{lem:Dx} is based on using (i) properties of fluid limits in Proposition~\ref{prop:properties}, and (ii) the boundedness of local derivatives (Lemma~\ref{lem:rbound} in Appendix~\ref{proof:bound}), 
        and (iii) the fact that $ \frac{\rm d}{{\rm d}t} (y_j(t) - \rho_j)^+ \le
		-\mu_j(y_j(t) - \rho_j)^+.$

		The detailed proof can be found in Appendix~\ref{proof:lem:Dx}.
	\end{subproof}
	Finally, by using Lemma~\ref{lem:Dx}, we can show that change
	of $V(\statez(t))$ is negative, almost surely, by averaging the
	change of $V(\statez(t))$ over all the subintervals
	$[\tau_n, \tau_{n+1})$ of $[t, t+\epsilon)$, as we do below.
	
	\begin{subproof}[Proof of Proposition~\ref{prop:dV}]

        Note that at any regular time $t$,
        \begin{equation}\label{eq:dV}
            \frac{\rm d}{{\rm d}t} V(\statez(t)) = -\sum_{i=1}^{\gtcnt} \LC_i \frac{\rm d}{{\rm d}t} x^{(e)}_{\greedk^{(i)}}(t)
            + \sum_{j=1}^J \LC
            \frac{\rm d}{{\rm d}t} (y_j(t) - \rho_j)^+,
        \end{equation}
        and
        $
        \frac{\rm d}{{\rm d}t} x^{(e)}_{\greedk^{(i)}}(t)=\lim_{\epsilon \to 0} \lim_{L \to \infty}  \frac{x^{L(e)}_{\greedk^{(i)}}(t+\epsilon)-x^{L(e)}_{\greedk^{(i)}}(t)}{\epsilon}.
        $
        Hence, using the division of $[t, t+\epsilon)$ into $N_{L}$ subintervals $[\tau_n, \tau_{n+1})$ of equal size, as defined in Section~\ref{sec:subintervals}, we can write
		\begin{equation*}\label{eq:loopclose}
		\begin{aligned}
        \frac{\rm d}{{\rm d}t} V(\statez(t))= & -\lim_{\epsilon \to 0} \lim_{L \to \infty}
		\frac{1}{N_{L}} \sum_{n=1}^{N_{L}}
		\sum_{i=1}^{\gtcnt} \LC_i
		\nabla x^{L(e)}_{\greedk^{(i)}} [\tau_{n}, \tau_{n+1}]	\\
		& + \sum_{j=1}^J \LC \frac{\rm d}{{\rm d}t} (y_j(t) - \rho_j)^+  \\
		\stackrel{(a)}{<} & -\delta(\epsilon_V)
		-\lim_{\epsilon \to 0} \lim_{L \to \infty}
		\frac{1}{N_{L}} \sum_{n=1}^{N_{L}}
		o(1) \stackrel{(b)}{=}  -\delta(\epsilon_V)
		< 0,
		\end{aligned}
		\end{equation*}
		where in (a) we used \dref{eq:rX} of Lemma~\ref{lem:Dx} in every subinterval $[\tau_{n}, \tau_{n+1}]$ and in (b) we used the property that
		$ \sum_{n=1}^{N_L} o(1) /N_L = o(1)$.

		Let $E_L$ be the event that
		$$
		-\frac{1}{N_{L}} \sum_{n=1}^{N_{L}}
		\sum_{i=1}^{\gtcnt} \LC_i
		\nabla x^{L(e)}_{\greedk^{(i)}} [\tau_{n}, \tau_{n+1}]
		+ \sum_{j=1}^J \LC \frac{\rm d}{{\rm d}t} (y_j(t) - \rho_j)^+ > 0.
		$$
		The probability that \dref{eq:rX} holds for all $N_L$ subintervals, is
		at least $1 - N_L o(L^{-2}) = 1 - o(L^{-1})$, which follows from $N_L = \Theta(1/\subL)$ based on
		Definition~\ref{def:sublen}. Hence, $\mathds{P}(E_L) < o(L^{-1})$, and $\frac{\rm d}{{\rm d}t} V(\statez(t)) < 0$ holds
		in probability.
		We can further show that convergence is almost sure. 
        This is because
		$\sum_{L=1}^{\infty} \mathds{P}(E_L) <
		\sum_{L=1}^{\infty} o(L^{-1}) < \infty,$
		and by the Borel-Cantelli Lemma~\cite{billingsley2008probability},
		$\frac{\rm d}{{\rm d}t} V(\statez(t))<0$, almost surely.
	\end{subproof}
	Propositions~\ref{prop:VS_} and \ref{prop:dV} complete the proof of Theorem~\ref{thm:convergence}.
\end{proof}

\begin{proof}[\textbf{Proof of Theorem~\ref{thm:main}}]
	The proof follows from Theorem~\ref{thm:convergence} and Theorems~\ref{thm:optimality} and \ref{thm:optimality2}. The details are standard and can be found in Appendix~\ref{sec:maintheorem_proof}.
\end{proof}

%% file: simulations-v3.tex
\section{Simulation Results}\label{sec:sim}

\subsection{Evaluation using synthetic traffic}
In this section, we evaluate the approximation ratio and convergence properties of {\DRA}.
We start by choosing the VM types considering
the VM instances offered by major cloud providers like Google Cloud,
are mainly optimized for either memory, CPU, or regular usage.
Further, instances are priced proportional to the resources they request, with each resource having a base pricing rate.
To simplify simulations, we considered instances that only have memory and CPU requirements. 
\begin{table}[ht]
\centering
	\begin{tabular}{ | c | c |c|c|c|}
		\hline	
		% \begin{tabular}{c|c}
			\multicolumn{2}{|c|}{vCPU} & \multicolumn{3}{|c|}{Memory: GB per vCPU}\\ \hline
			Small&Large& High&Low&Regular \\ \hline
				2,4, or 8& 32 or 64& 8 or 16 &1 or 2 &4 \\
		\hline
	\end{tabular}
    \vspace{0.05 in}
	\caption{The representative VM instances from Google Cloud based on combination of vCPU and Memory.}\label{table1}
	\vspace{-0.1 in}
\end{table} 
In particular, we used representative VM instances, based on combination of vCPU and memory in Table~\ref{table1}. 
Lastly, each vCPU usage generates $8$ reward per unit time, while each GB of memory generates $1$. 
This choice was made based on the relative pricing of CPU and memory of VMs offered by Google Cloud, 
according to which 8 GB memory is approximately priced as much as 1 vCPU \cite{GoogleVMpricing}. 
We generated random collections of VM types, each with three small and three large VMs, with vCPU and memory chosen randomly from Table~\ref{table1}. 
Servers always have capacity of 80 vCPUs and 640 GB of memory. 
The normalized workload $\rho_j$ for each VM type $j$ is selected uniformly at random between $0.2$ to $2$.
The statistics we obtained based on $50$ randomly generated VM collections and workloads was that, 
in $23$ of them reward of global greedy was identical to the optimal, 
on average its ratio compared to optimal was $0.972$ and in the worst case it was no less than $0.86$. 
Recall that optimal can be found by solving optimization (\ref{eq:opt1b}).
For the rest of simulations, we considered a subset of the worst-case VM collection and its corresponding workload, namely, VM types are:  (1, 1), (4, 16), (2, 32), (32, 256), and $\bm{\rho}$ rounded to $(2, 1/2, 4/3, 1)$.

To better understand how workload may affect the approximation ratio,  
we study this worst-case example and scale its workload $\bm{\rho}$ by a factor $\alpha$ that ranges from $0$ to $10$. 
Figure~\ref{fig:reward_vs_wl} shows the reward for the global greedy 
${U}^{(g)}[\alpha\bm{\rho}]$ and the optimal reward ${U}^{\star}[\alpha\bm{\rho}]$.
We notice there are two critical $\alpha$ points. 
Before the first point,
the workload is low enough such that the global greedy assignment 
can fully accommodate it, hence its reward is the same as the optimal which should also be able to accommodate the full workload. 
The second point is a point above which the workload is high such that it is possible to assign the configuration of maximum reward to all servers without leaving any slots empty. 
In this case, both the rewards will coincide again, and take the maximum possible value.

In Figure~\ref{fig:reward_vs_wl}, the two critical points
are $\alpha = 6/7$ and $\alpha = 6$. 
The worst ratio between the reward of global greedy and the optimal occurs at $\alpha = 1$, which is~$\approx 0.862$.
%in which case optimal Assignment is
%\be
%(0, 0, 2, 1) \times  1/3 L \quad
%(2, 1, 1, 1) \times  1/2 L \quad
%(6, 0, 1, 1) \times  1/6 L
%\ee
%and the Global Greedy Assignment is
%\be
%(12, 3, 8, 0) \times 1/6 L \quad
%(0, 0, 0, 1)  \times 5/6 L.
%\ee
Note that in general ${U}^{(g)}[\alpha\bm{\rho}]$ and ${U}^{\star}[\alpha\bm{\rho}]$ might coincide even between the critical points although this is not the case for this example.

To study the impact of the number of servers $L$, we run {\DRA} in systems with various number of servers, and compare the obtained average normalized reward (normalized with $L$) with the global greedy reward ${U}^{(g)}[\bm{\rho}]$, and the optimal reward ${U}^{\star}[\bm{\rho}]$. 
The arrivals are generated at rate $\rho_j L$, and service times are exponentially distributed with mean 1. 
The result is depicted in Figure~\ref{fig:reward_vs_srv}, which clearly shows that as the number of servers $L$ becomes large, {\DRA} approaches the global greedy reward and 86\% of the optimal reward.  
Further, Figure~\ref{fig:reward_time} shows how the reward of {\DRA} evolves over time and converges to the global greedy reward when $L=180$.

\begin{figure}
	\centering
	\includegraphics[width=0.8\columnwidth%, height=3.2cm
	]{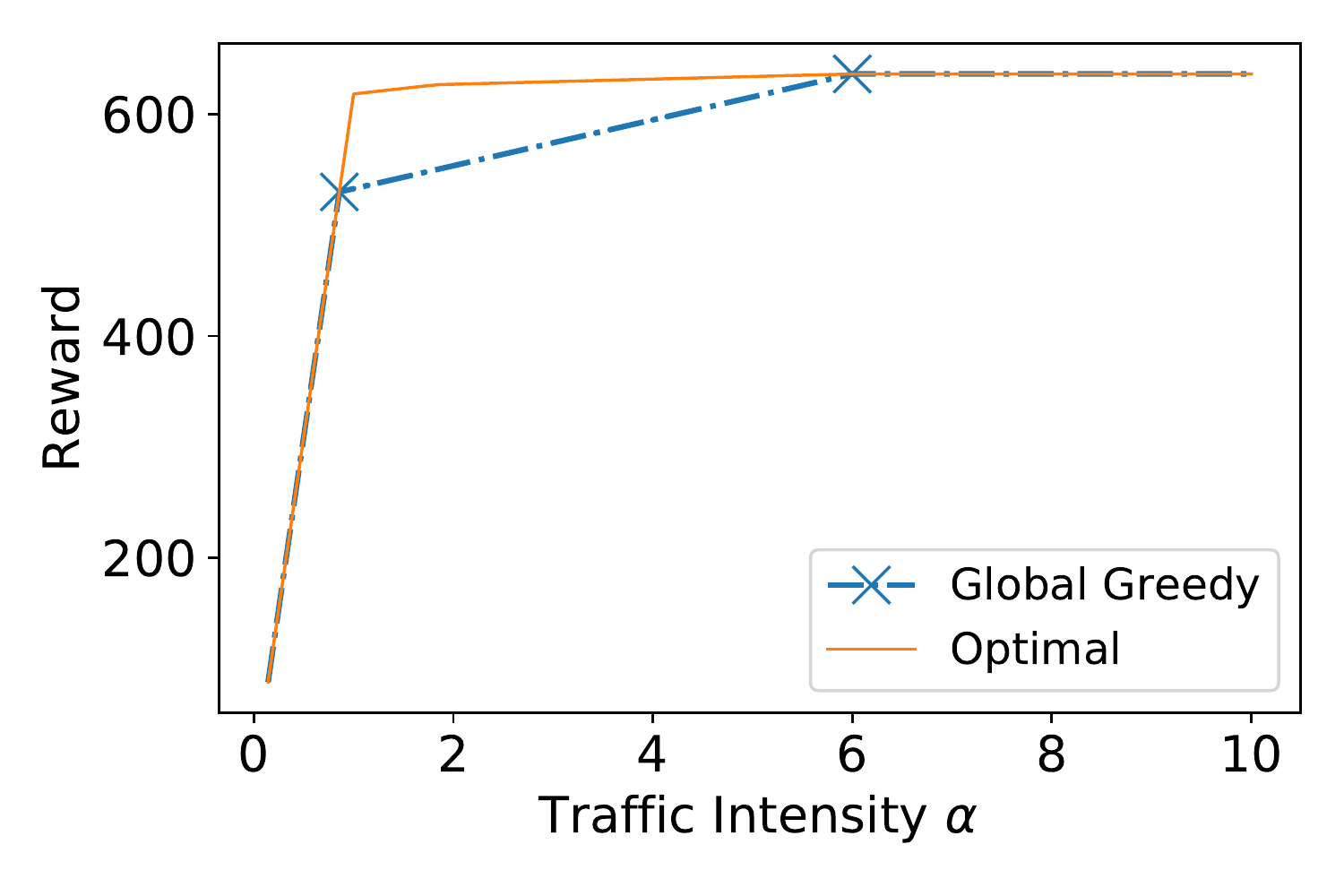}
	\caption{Global greedy vs. optimal,
		as workload $\alpha \bm{\rho}$ increases. The rewards coincide outside the marked points.}
	\Description{Global greedy reward coincides with optimal reward
		outside limit points, but not necessarily in between.}
	\label{fig:reward_vs_wl}
	\vspace{-0.1in}
\end{figure}
\begin{figure}
    \centering
    \includegraphics[width=0.9\columnwidth]{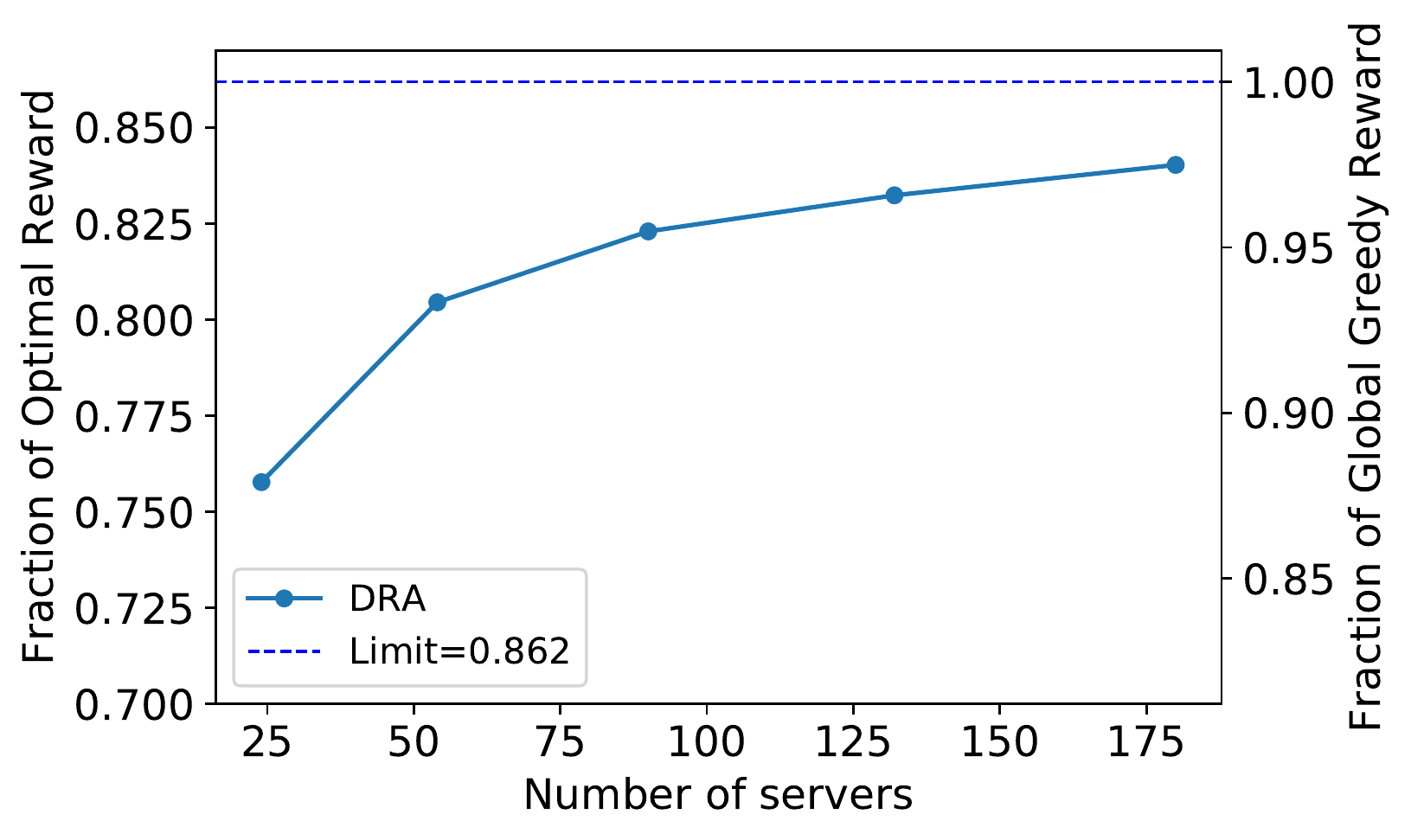}
    \caption{The reward of {\DRA} as a fraction of the optimal reward 
        (left y-axis), and that of the global greedy (right y-axis).}
    \Description{The reward of {\DRA} approaches the global greedy one
        as number of servers increase.}
    \label{fig:reward_vs_srv}
\end{figure}
\begin{figure}
    \centering
    \includegraphics[width=0.8\columnwidth]{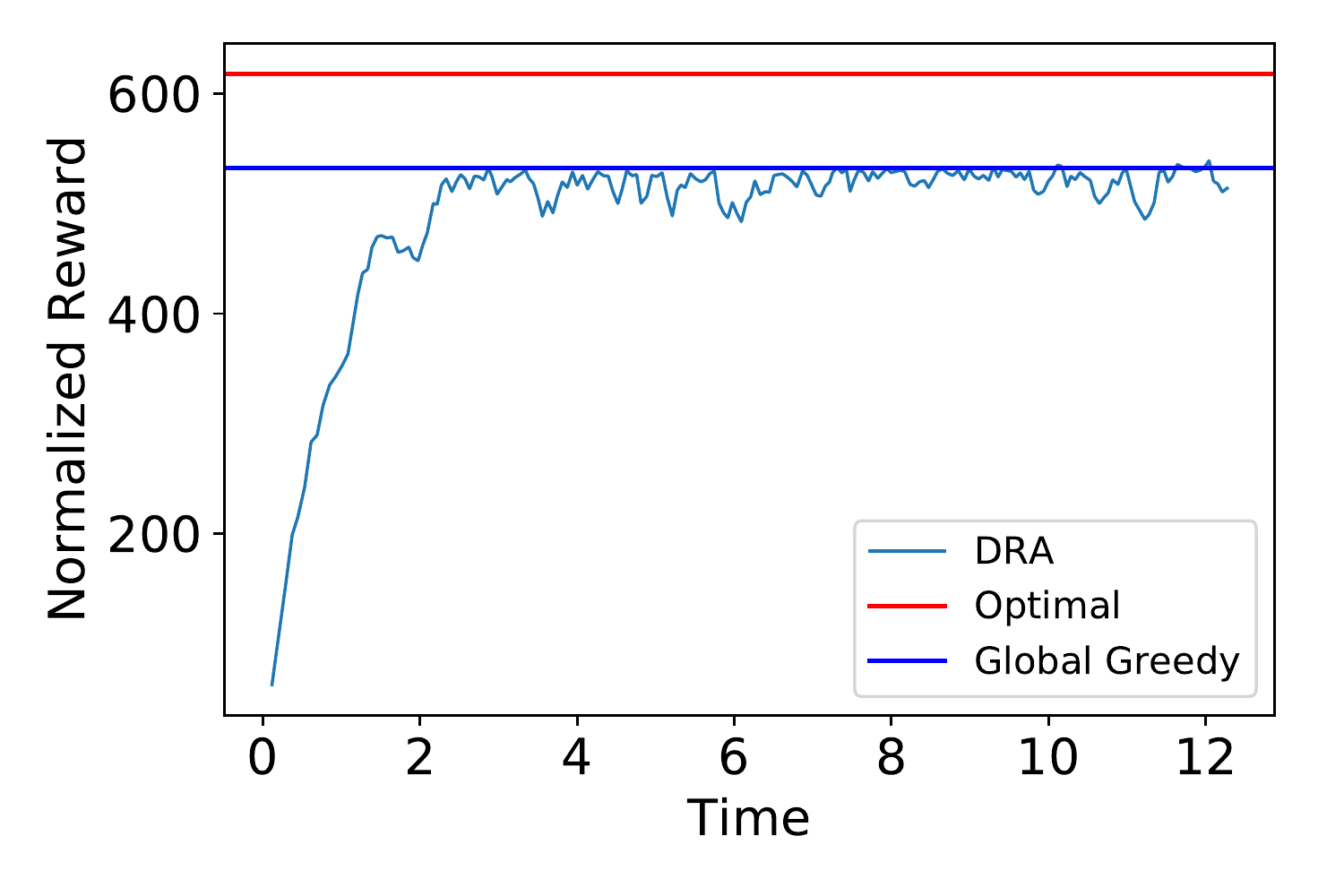}
    \caption{Convergence of the reward of {\DRA} to that of the global greedy assignment over time	when $L=180$ servers.}
    \Description{{\DRA} converges over time to global greedy in this examples of $L=180$ servers.}
    \label{fig:reward_time}
\end{figure}
\begin{figure}
	\centering
	\includegraphics[width=0.9\columnwidth%, height=3.6cm
	]{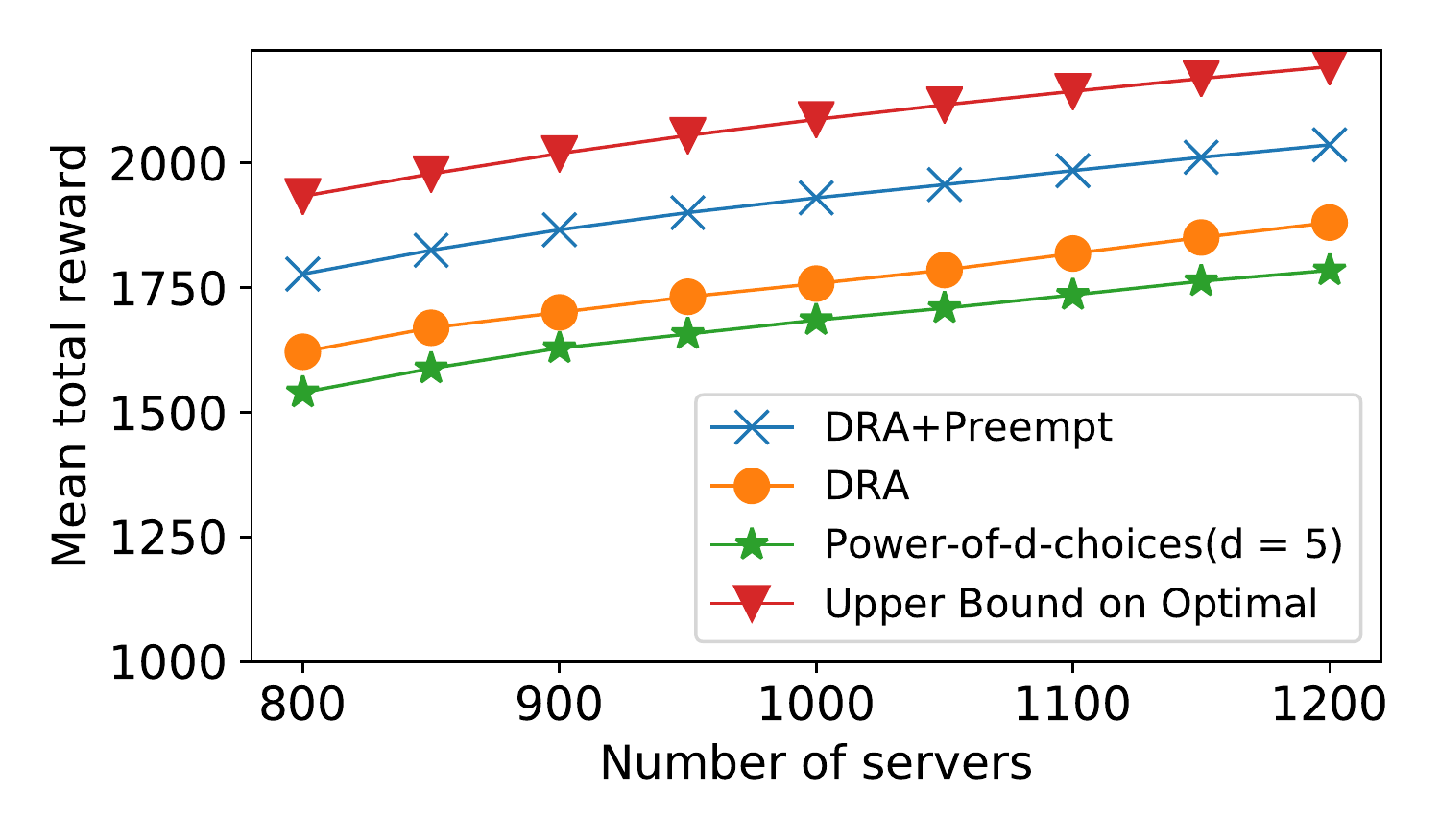}
	\caption{The comparison of rewards for different number of servers based on the Google trace.}
	\Description{The reward of {\DRA} compared to two alternatives and
		an upper bound for different number of servers when simulated
		on the Google trace. {\DRA} with preemptions is better
		than the default {\DRA} and better than the power-of-$d$-choices}
	\label{fig:gtraceSeries}
	\centering
	\includegraphics[width=0.9 \columnwidth]{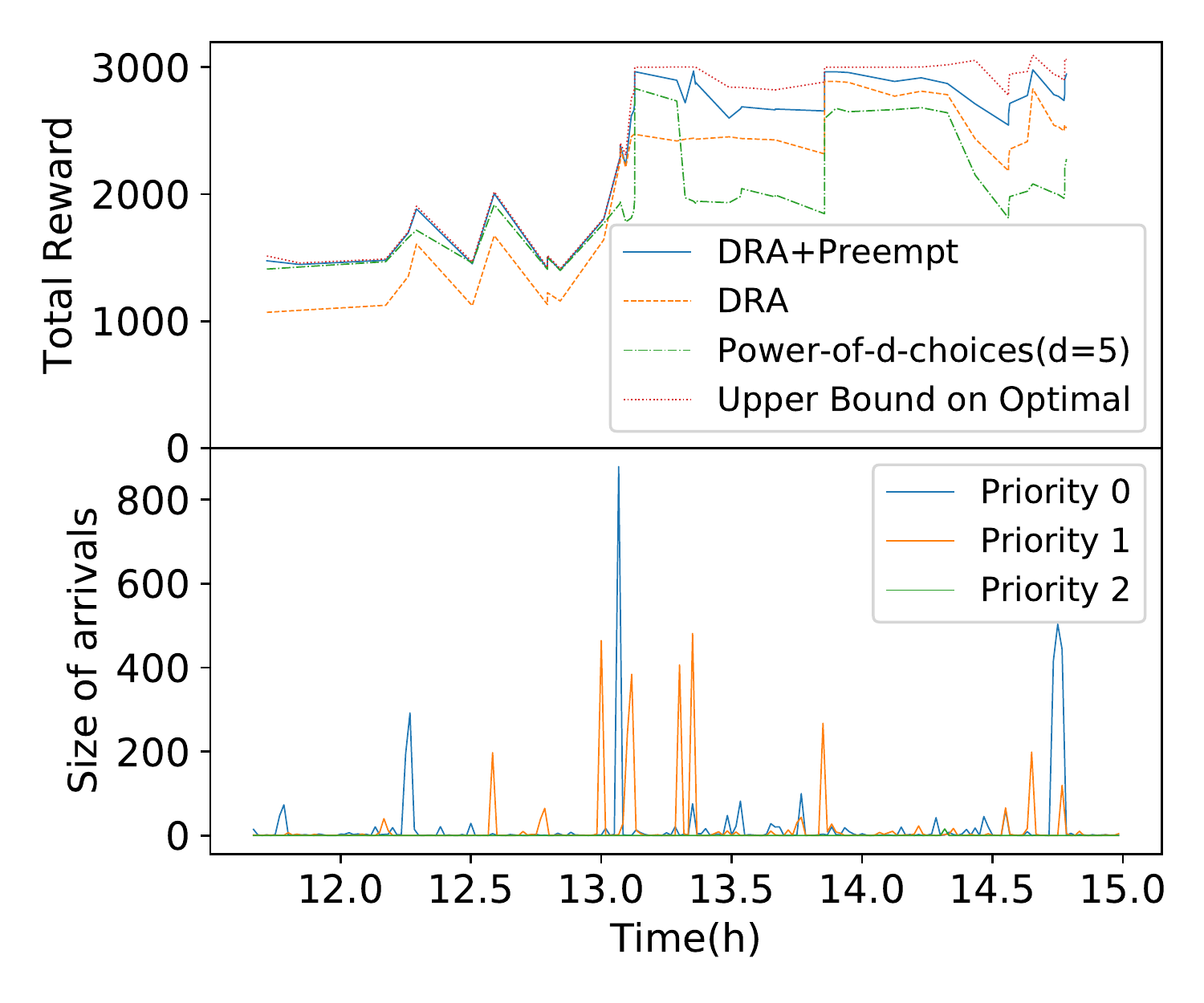}
	\caption{Comparison of the reward over time of different algorithms for a part of the Google trace.}
	\Description{Comparison of the reward of different algorithms over time for a part of the Google trace. {\DRA} with preemptions
		reacts to spikes in demand better than the other two.}
	\label{fig:gtraceDRA}
\end{figure}

\subsection{Evaluation using  real traffic trace}

We evaluate our algorithm using a more realistic setting with arrival and service times extracted from a Google cluster dataset~\cite{ClusterData}.
In particular, we extracted tasks which were completed within the time window of the trace and used the first $1$ million in all simulations.
Tasks were mapped to types by setting their resource requirements
to be the largest of the requested resources and rounding it up to the closest power of $1/2$. 
Their reward was set to be equal to their rounded size multiplied by a factor that depends on their priority. 
Factor is $1,3,9$ for priorities $0,1,2$ respectively.
Tasks have the same type if both their priority and normalized size are equal.
The size of servers is normalized to $1$.

We compare the performance of {\DRA} and three other algorithms:

    \textbf{Upper bound}: It solves optimization~(\ref{eq:opt1c}) with $\hat{\mathbf Y}(t)$ being the number of
	jobs in an infinite server system that rejects no jobs. 
    This gives an upper bound on the performance of any algorithm.
	
	\textbf{Power-of-$d$-choices}: Upon an arrival, it picks $d$ servers 
	and attempts to schedule the job arrived in the least loaded server if it fits~\cite{xie2015power}.
	We picked $d=5$, but behavior of the algorithm is not expected to 
	change significantly for larger $d$.

	\textbf{DRA+Preemption}: This is simply an extension to our algorithm that
	preempts some of the jobs of priority $0$, 
    when a job of type $j$ with priority $1$ or $2$ gets rejected. 
    Notice that preemptions of low priority jobs 
	is already considered in similar scenarios that in Google cluster
	setting~\cite{Verma2015}.
	Specifically, our algorithm attempts to preempt jobs of priority $0$ 
	starting from those of smallest size. Considering reservation factor is 
    $g(L)$ and size of type-$j$ job is $s_j$, preemptions will stop if the total size of preempted jobs is $g(L) s_j$ or no more priority $0$ jobs are available. 
	The algorithm finds which jobs to preempt, if any, the same way it finds jobs to migrate so this addition needs minimal changes in implementation.

Figure~\ref{fig:gtraceSeries} shows the performance results (the time-average of rewards) with
varying number of servers. Especially, considering preemptions in {\DRA} makes a great difference. Note that the upper bound may be impossible
to achieve by any algorithm.
 
To give more insight, in Figure~\ref{fig:gtraceDRA} we plot 
the total reward over time for all algorithms
for a part of the simulation of $1000$ servers, 
including the corresponding total size of arrivals of all job types 
of each priority. 
We notice that power-of-$d$-choices algorithm can be better than {\DRA} in parts of trace in which a spike in demand of priority 0 jobs is followed by 
a spike in demand of priority 1 jobs. 
This is because reservation of {\DRA} is not sufficient to account for spikes in demand, while power-of-$d$-choices does not efficiently use 
the resources of all servers and may have more free capacity when a spike occurs. 
{\DRA} with preemptions is particularly effective in such scenarios as it does not need to reserve resources in advance. 
In addition, it makes efficient use of the resources of all the servers the same way {\DRA} does and thus is strictly better than both of the other algorithms 
in almost all parts of the trace.

%% file: conclusions.tex
\section{Conclusions}
In this paper, we proposed a VM reservation and admission policy that operates in an online manner and can guarantee at least $1/2$ (and under certain monotone property, $1-1/e$) of the optimal expected reward. 
Assumptions such as Poisson arrivals and exponential service times are made to simplify the analysis, and the policy itself does not rely on this assumption. 
The policy strikes a balance between good VM packing and serving high priority VM requests, by maintaining only a small number $g(L)=\omega(\log L)$ of reserved VM slots at any time.
Our techniques for analysis of fluid-scale processes on the boundary in our problem, and the design of LP-based Lyapunov functions with a unique maximizer at the given desired equilibrium, can be of interest on their own.

    Although we considered that the policy classifies and reassigns
    servers at arrival and departure events, this was only to simplify 
    the analysis, and in practice {\EPA} can make such updates periodically, 
    by factoring all arrival or departures in the past period 
    in its input for the current period. 
    Further, if a more accurate estimate of the workload is available, we can incorporate that estimate in the vector $\hat{\mathbf Y}$ used by {\DRA}, to improve the convergence time. 

Moreover, the policy can be extended to a multi-pool server system, where constant fractions of servers belong to different server types. We postpone the details to a future work.         

%% file: permutation_proof.tex
\section{Proof of Proposition~\ref{prop:index}}\label{indexproof}
	We omit the notation $[\bm{\rho}]$ for compactness.
	Also we use the following notations for shorthand purposes
	\begin{equation}\label{eq:amin_assgn}
	\begin{aligned}
		&\textrm{amin}(j) :=
		\argmin_{i:k^{(j)}_i>0} \frac{\rho_i}{k^{(j)}_i} - 
		\sum_{\ell=1}^{j-1} \globl{z}^{(\ell)} 
		\frac{k^{(\ell)}_i}{k^{(j)}_i}, \\
		&\textrm{assgn}(j) := 
		\min_{i:k^{(j)}_i>0} \frac{\rho_i}{k^{(j)}_i} - 
		\sum_{\ell=1}^{j-1} \globl{z}^{(\ell)} \frac{k^{(\ell)}_i}{k^{(j)}_i}.
	\end{aligned}
	\end{equation}
	As a convention, if minimum is attained by more that one indexes, 
	the lowest one is chosen.
	We define
	\begin{equation}\label{eq:Pdef}
	\begin{aligned}
	& \sigma := (\textrm{amin}(1), 
		\textrm{amin}(2), \ldots, \textrm{amin}(J)), \\
	& \globl{z}^{(j)} := \textrm{assgn}(j), \quad 
		j=1, \ldots, J.
	\end{aligned}
	\end{equation}
	We can verify that $\sigma$ and its corresponding values
	$\globl{z}^{(j)}$ satisfy all the conditions of
	Proposition~\ref{prop:index}.
	It remains to prove that this permutation $\sigma$ is unique.

	Suppose there is another permutation
	$\sigma^\prime := \{\permp{1}{\prime},
	\permp{2}{\prime}, \ldots, \permp{J}{\prime}\}$
	that satisfies the properties of
	Proposition~\ref{prop:index} and $j$ is the lowest index for which
	$\perm{j} \neq \permp{j}{\prime}$.
%			The index $i^\prime_j$ should satisfy
%			\begin{equation}
%			\rho_{i_j^\prime} = \sum_{\ell=1}^j {k}^{(\ell)}_{i_j^\prime} 
%			\globl{z}^{(\ell)}
%			\end{equation}
%			and ${k}^{(j)}_{i^\prime_j} > 0$.
	We define $D^j_i := \frac{\rho_{i}}{k^{(j)}_{i}} - 
	\sum_{\ell=1}^{j-1} \globl{z}^{(\ell)}
	\frac{k^{(\ell)}_{i}}{k^{(j)}_{i}}$ and
	compare $D^j_\perm{j}$ to $D^j_{\permp{j}{\prime}}$.
	We will reach a contradiction in all possible cases,
	which proves that permutation $\sigma$ is unique. 
	\begin{enumerate}[leftmargin=*]
		\item If $D^j_\perm{j} > D^j_{\permp{j}{\prime}}$, 
			then $\perm{j} := \textrm{amin}(j)$ is not the minimizer of \dref{eq:amin_assgn} and this contradicts the definition of $\textrm{amin}(j)$.
			% in
			%\dref{eq:amin_assgn}.
%					we have a contradiction with definition of 
%					$\perm{j} := \textrm{amin}(j)$ since it is 
%					the one that minimizes expression $D^j_i$ among all 
%					$i \in \mathcal{J}$ for	which $k^{(j)}_i > 0$.
		\item If $D^j_\perm{j} = D^j_{\permp{j}{\prime}}$, 
			then we consider the index $j_a$ 
			for which $\perm{j} = \permp{j_a}{\prime}$
			and the index $j_b$ for which 
			$\permp{j}{\prime} = \perm{j_b}$. 
			This implies
			\begin{equation} 
			D^j_{\perm{j_b}} = D^j_{\permp{j}{\prime}} 
			= D^j_\perm{j} = D^j_{\permp{j_a}{\prime}} 
			= \globl{z}^{(j)},
			\end{equation}
			or equivalently
			\begin{equation}\label{eq:rho_ab}
			\begin{aligned}
			&\rho_{\permp{j_a}{\prime}} = \sum_{\ell=1}^{j} 
				{k}^{(\ell)}_{\permp{j_a}{\prime}} 
				\globl{z}^{(\ell)} = \rho_{\permp{j}{\prime}}, \\
			&\rho_{\perm{j_b}} = \sum_{\ell=1}^{j} 
				{k}^{(\ell)}_{\perm{j_b}}
				\globl{z}^{(\ell)} = \rho_{\perm{j}}.
			\end{aligned}
			\end{equation}
			% is j < j^\prime ? A: yes
			We also notice that $j < j_a$,
			since $\perm{j} \neq \permp{i}{\prime}$ for $i=1, \ldots, j$
			and similarly $j < j_b$.
			Then, considering (\ref{eq:rho_ab}), 
			assumption (\ref{eq:rho_xg2}),
			$j < j_a$ and $j < j_b$, we get 
			$\permp{j}{\prime} < \permp{j_a}{\prime}$ and
			$\perm{j} < \perm{j_b}$ which are contradictory as 
			they imply
			$$
			\perm{j} < \perm{j_b} = \permp{j}{\prime} < 
			\permp{j_a}{\prime} = \perm{j}.
			$$
		\item If $D^j_\perm{j} < D^j_{\permp{j}{\prime}}$, 
			then we consider the index $j_a$ 
			for which $\perm{j} = \permp{j_a}{\prime}$.
			Then for permutation $\sigma^\prime$ to be valid, 
			there should be constants $\globl{z}^{(\ell)} \ge 0$
			for $\ell = 1, \ldots, j_a$ such that,
			\begin{equation}
			\rho_{\permp{j_a}{\prime}} = \sum_{\ell=1}^{j_a} 
			{k}^{(\ell)}_{\permp{j_a}{\prime}}
			\globl{z}^{(\ell)}.
			\end{equation}
			On the other hand,
			\begin{equation}\label{eq:contrA}
			\sum_{\ell=1}^{j_a} 
			{k}^{(\ell)}_{\permp{j_a}{\prime}}
			\globl{z}^{(\ell)} =
			\rho_{\permp{j_a}{\prime}} \equiv \rho_\perm{j}
			\stackrel{(a)}{<} \sum_{\ell=1}^{j} 
			{k}^{(\ell)}_{\permp{j_a}{\prime}} \globl{z}^{(\ell)}.
			\end{equation}
			where (a) is a consequence of 
			$D^j_\perm{j} < D^j_{\permp{j}{\prime}}$
			when $\globl{z}^{(j)} = D^j_{\permp{j}{\prime}}$.
			From (\ref{eq:contrA}), we also get
			\begin{equation}
				\sum_{\ell=j+1}^{j_a} 
				{k}^{(\ell)}_{\permp{j_a}{\prime}} \globl{z}^{(\ell)} < 0,
			\end{equation}
			which contradicts the assumption $\globl{z}^{(\ell)} \ge 0$
			for $\ell = j+1, \ldots, j_a$, if we consider
			${k}^{(\ell)}_{\permp{j_a}{\prime}} \ge 0$
			for $\ell = j+1, \ldots, j_a$. 
	\end{enumerate}

%% file: greedy_convergence.tex
\section{Proof of Proposition~\ref{prop:greedy_conv}}\label{sec:greedy_conv}
    For $\bk \not \in \calK^{(g)}$, it is obvious that
    ${\hat{X}^L_{\bk}} = 0$,
    since $\bk$ is never assigned by {\VPA} for any input.
    Thus $\lim_{L\to \infty}\frac{\hat{X}^L_{\bk}}{L} = 0 = x^{(g)}_{\bk}.$ Hence, it remains to prove the proposition for
    $\bk \in \calK^{(g)}$, i.e., for $\greedk^{(i)}$, $i=1, \ldots, \gcnt$.
    For this we will use the following Lemma.
    \begin{lemma}\label{lem:hat_conv}
        For $i=1, \ldots, \gcnt$, 
        \begin{equation}\label{eq:hat_conv}
        \lvert {\hat X}^{L}_{\greedk^{(i)}}
        - Lx^{(g)}_{\greedk^{(i)}} \rvert \le (\Kmax+1)^{i-1},
        \end{equation} 
        where ${\hat X}^{L}_{\greedk^{(i)}} = {\VPA} (L\bm{\rho})$ when the number of servers is $L$
        and $\Kmax$ is an upper bound on the maximum number of jobs in any configuration.
    \end{lemma}
    \begin{proof}
        For shorthand purposes, define  ${\hat X}^{a}_{i} := {\hat X}^{L}_{\greedk^{(i)}}(t) $ and
        ${\hat X}^{b}_{i} := Lx^{(g)}_{\greedk^{(i)}}$.        
        By definition of global greedy assignment, for every $i \in \{1, \ldots, \gcnt\}$,
        one of the following holds:
        \begin{enumerate}
        \item There is some $j \in \mathcal{J}$ 
        such that $L \rho_{j}$ fits exactly in
        ${\hat X}^{b}_{\ell}$ servers assigned to $\greedk^{(\ell)}$ 
        for $\ell=1, \ldots, i$, i.e.,
        \begin{equation}\label{eq:Xb1}
            L \rho_{j} = \sum_{\ell=1}^i {\hat X}^{b}_{\ell} \greedki^{(\ell)}_{j}.
        \end{equation}
        \item All servers are assigned to one of the configurations
        $\greedk^{(\ell)}$ for $\ell=1, \ldots, i$, i.e.,
        \begin{equation}\label{eq:Xb2}
        {\hat X}^{b}_{i} = L - \sum_{\ell=1}^{i-1} {\hat X}^{b}_{i}.
        \end{equation}
        \end{enumerate}
    
        Similarly, for {\VPA}, there is an index $I_L$ such that one of the following holds:
        \begin{enumerate}
        \item For $i \in \{1,\ldots, I_L-1\}$, there is $j \in \calJ$
        such that $L \rho_{j}$ jobs fit in ${\hat X}^{a}_{\ell}$ servers 
        assigned to $\greedk^{(\ell)}$ for $\ell=1, \ldots, i$, 
        but not in ${\hat X}^{a}_{\ell}$ servers assigned to $\greedk^{(\ell)}$ for 
        $\ell=1, \ldots, i-1$ and ${\hat X}^{a}_{i} - 1$ servers
        assigned to $\greedk^{(i)}$.
        This implies that
        \begin{equation}\label{eq:Xa1}
        \begin{aligned}
        {\hat X}^{a}_{i} \greedki^{(i)}_j &\ge 
        L \rho_j
        - \sum_{\ell=1}^{i-1} {\hat X}^{a}_{\ell} \greedki^{(\ell)}_j, 
        \\
        ({\hat X}^{a}_{i} - 1) \greedki^{(i)}_j &< 
        L \rho_j
        - \sum_{\ell=1}^{i-1} {\hat X}^{a}_{\ell} \greedki^{(\ell)}_j.
        \end{aligned} 
        \end{equation}
        \item For $i \in \{I_L, \ldots, \gcnt\}$, all servers are assigned 
        to one of the configurations $\greedk^{(\ell)}$ for $\ell=1, \ldots, i$, i.e.,
        \begin{equation}\label{eq:Xa2}
        {\hat X}^{a}_{i} = L - 
        \sum_{\ell=1}^{i-1} {\hat X}^{a}_{\ell}.
        \end{equation}
        \end{enumerate}

        We can show inductively that for any $i \in \{1, \ldots, I_L-1\}$ 
        and for large enough $L$, 
        if (\ref{eq:Xa1}) holds then (\ref{eq:Xb1}) holds for the same job type $j$. 
        By assuming otherwise we can easily reach a contradiction 
        (details are omitted).
        This means we can replace $L \rho_j$ in (\ref{eq:Xa1}) with 
        the right hand side of (\ref{eq:Xb1}).
        Also with similar arguments we can prove that 
        if (\ref{eq:Xa2}) holds then (\ref{eq:Xb2}) holds as well. 
        Therefore, for $i=1$, we either get 
        ${\hat X}^{a}_{1} = {\hat X}^{b}_{1} = L$ or
        ${\hat X}^{a}_{1} \greedki^{(1)}_j \ge {\hat X}^{b}_{1} \greedki^{(1)}_j >
        ({\hat X}^{a}_{1}-1) \greedki^{(1)}_j$.
        Hence, in either case, we have $|{\hat X}^{a}_{1} - {\hat X}^{b}_{1}| < 1$,
        which proves (\ref{eq:hat_conv}) for $i=1$. Now suppose the statement is true for indexes $1, \ldots, i-1$. 
        We show that it is also true for $i$. 
        
        If (\ref{eq:Xa1}) holds, then by replacing
        $L \rho_j$ in (\ref{eq:Xa1}) with 
        the right-hand-side of (\ref{eq:Xb1}), we get
        \ben
        & |{\hat X}^{a}_{i} - {\hat X}^{b}_{i}| < 
        1 + \sum_{\ell=1}^{i-1} 
        \frac{\greedki^{(\ell)}_{j}}{\greedki^{(i)}_{j}} 
        |X^{a}_{\ell} - X^{b}_{\ell}|.
        \een
        Hence, noting that $\frac{\greedki^{(\ell)}_{j}}{\greedki^{(i)}_{j}} 
        \le \Kmax$, $\ell=1,\ldots, i-1$, we get
        \ben
        |{\hat X}^{a}_{i} - {\hat X}^{b}_{i}| \le
        1 + \sum_{\ell=1}^{i-1} 
        \Kmax (1+\Kmax)^{\ell-1} = (1+\Kmax)^{i-1}.
        \een
        %%%%%%%%
        
        If instead (\ref{eq:Xa2}) holds, then since (\ref{eq:Xb2}) also holds,
        and we get
        \ben
        |{\hat X}^{a}_{i} - {\hat X}^{b}_{i}| & \le &
        \sum_{\ell=1}^{i-1} |X^{a}_{\ell} - X^{b}_{\ell}|\\
        & \le &
        \sum_{\ell=1}^{i-1} 
        (1+\Kmax)^{\ell-1} < (1+\Kmax)^{i-1}.
        \een
        This completes the proof of (\ref{eq:hat_conv}) for arbitrary $i$.
    \end{proof}
    The proposition then follows since for any $i \in \{1, \ldots, \gcnt\}$,
    \begin{equation*}
    \lim_{L \to \infty}
    \left \lvert \frac{{\hat X}^{L}_{\greedk^{(i)}}}{L}
    - x^{(g)}_{\greedk^{(i)}} \right\rvert \le 
    \lim_{L \to \infty}
    \frac{(\Kmax+1)^{i-1}}{L} = 0.
    \end{equation*} 

%% file: optimality_proof3.tex
\section{Proof of Proposition~\ref{prop:bound_tight}}\label{prf:optimality}

	Consider a system with $J$ job types. 
    Suppose type-$i$ jobs, for each $i=1, \ldots, J-1$, 
	can fit $J$ times in an empty server, and type-$J$ jobs can fit $N+1$ times.
	Suppose the configuration that uses $1$ job of each type $i$ and 
	$N$ jobs of type $J$ is feasible as well.
	The aforementioned configurations will be maximal if we assume we have
	$J+1$ resources and 
	\begin{itemize}[leftmargin=*]
		\item each type-$i$ job, $i=1, \ldots, J-1$, occupies $1/J$ of 
		resource $i$ and $1/J$ of resource $J+1$.
		\item each type-$J$ job occupies $1/(N+1)$ of resource $J$ and 
		$1/(JN)$ of resource $J+1$.
	\end{itemize}
	Assume that each type-$i$ job, $i=1, \ldots, J-1$, gives reward
	$u_i = \frac{1}{J}\left(\frac{J-1}{J}\right)^{i-1} u$, and each type-$J$ job
	gives a reward $u_J = \frac{1}{N+1} \left(\frac{J-1}{J}\right)^{J-1}u$.
	Let the workload $\bm{\rho}$ be such that $\rho_i = 1$ for $i=1, \ldots, J-1$ and $\rho_J = N$.
	
	In this example, the global greedy assignment assigns only the $J$ 
	configurations that consist of a single job type and each one is assigned to
	$\frac{1}{J}$ fraction of servers.
	The normalized reward of $\mathbf{x}^{(g)}$ is
	\begin{equation}
	\begin{aligned}
	& U^{(g)}(J,N) = \frac{1}{J} \left(\sum_{i=1}^{J-1} J \frac{1}{J} 
	\left(\frac{J-1}{J}\right)^{i-1} + (N+1) \frac{1}{N+1}
	\left(\frac{J-1}{J}\right)^{J-1} \right)u \\
	& = \left(1 - \left(1 - 1/J\right)^{J} \right) u.
	\end{aligned}
	\end{equation}
	The optimal assignment assigns the configuration that uses $1$ job of each type $i$ and $N$ jobs of type $J$ to all servers.
    The normalized reward of $\mathbf{x}^{\star}$ is therefore
	\begin{equation}
	\begin{aligned}
	& U^{\star}(J, N) =  \left(\sum_{i=1}^{J-1} \frac{1}{J} 
	\left(\frac{J-1}{J}\right)^{i-1} + N \frac{1}{N+1}
	\left(\frac{J-1}{J}\right)^{J-1} \right)u = \\
	&  \left( 1 - \left(1 - 1/J\right)^{J-1} + \frac{N}{N+1} 
	\left(1 - 1/J\right)^{J-1}
	\right) u.
	\end{aligned}
	\end{equation}
		
	From these, the result is obvious, as
	\begin{equation}
	\begin{aligned}
	&\lim_{N \to \infty} \frac{ U^{(g)}(J,N) }{ U^{\star}(J, N) } = \\
	&\frac{\left(1 - \left(1 - 1/J\right)^{J} \right)}
	{\left( 1 - \left(1 - 1/J\right)^{J-1} + \left(1 - 1/J\right)^{J-1}
		\right)} = 1 - \left(1 - 1/J\right)^{J},
	\end{aligned}
	\end{equation}
	and
	\begin{equation}
	\lim_{J \to \infty} \lim_{N \to \infty} 
	\frac{ U^{(g)}(J,N) }{ U^{\star}(J, N) } =
	1 - \lim_{J \to \infty} \left(1 - 1/J\right)^{J}  =
	(1 - 1/e). %\qedhere
	\end{equation}

%% file: fluid_proof4.tex
\section{Proof of Proposition~\ref{prop:fl_limits}}\label{sec:fl_limits}

	For the proof of this proposition  % Proposition~\ref{prop:fl_limits}, 
	we will need the following Lemma
	\begin{lemma}\label{lem:X_change}
		For $i=1, \ldots, \gcnt$, 
        %and $\bk \in \mathcal{K}$ such that$\bk \le \greedk^{(i)}$, 
		the absolute jump in ${\hat X}^{L}_{\greedk^{(i)}}(t)$,
		${X}^{L}_{\greedk^{(i)}}(t)$, ${X}^{L(e)}_{\greedk^{(i)}}(t)$, 
		%and ${Z}^{L}_{\greedk^{(i)}, \bk}(t)$, for $\bk \le \greedk^{(i)}$, 
		after a job arrival or departure event, is at most $(\Kmax+1)^{i-1}$, where 
		$\Kmax$ is an upper bound on the maximum number of jobs in any configuration.
	\end{lemma}
	\begin{proof}
		We first prove the result for ${\hat X}^{L}_{\greedk^{(i)}}(t)$.
        We consider \\
        $\sum_{i=1}^{\gcnt-1} {\hat X}^{L}_{\greedk^{(i)}}(t) < L$ before and after an event, 
        as otherwise the range of change of any ${\hat X}^{L}_{\greedk^{(i)}}(t)$ will be even smaller, because of the extra constraint.
		Consider an arrival or departure event takes place. We denote the 
		values ${\hat X}^{L}_{\greedk^{(i)}}(t)$ for 
		$i \in \{1, \ldots, \gcnt\}$, as given by 
		Algorithm~\ref{alg:vpa}, before and after the event by 
        ${\hat X}^{a}_{i}$ and ${\hat X}^{b}_{i}$ respectively.

		We define $i^\star$  to be the first index in $\{1, \ldots, \gcnt\}$
		for which ${\hat X}^{a}_{i^\star} \neq {\hat X}^{b}_{i^\star}$ so 
		for $\ell \in \{1, \ldots, i^\star-1\}$ we have
		${\hat X}_{\ell} := {\hat X}^{a}_{\ell} = {\hat X}^{b}_{\ell}$.
		We also define for $j \in \{1, \ldots, J\}$,
		$Y_{j} := Y^L_{j}(t) + g(L)$, where $Y^L_{j}(t)$ 
		is the number of jobs in the system before the event.
		Finally we define $\zeta$ to be $1$ 
		if the event is arrival and $-1$ if the event is departure.
		
		We prove by induction that for $i \ge i^\star$, 
		$|X^{a}_{i} - X^{b}_{i}| \le (1 + \Kmax)^{i-i^\star}$.
		We start with the base case $i = i^\star$.
		Before any event, we know there is some $j \in \mathcal{J}$ 
		such that $Y_j$ jobs fit in 
		${\hat X}^{a}_{\ell}$ servers assigned to $\greedk^{(\ell)}$ 
		for $\ell=1, \ldots, i$, 
        but not in ${\hat X}^{a}_{\ell}$ servers assigned to $\greedk^{(\ell)}$ for $\ell=1, \ldots, i-1$ and ${\hat X}^{a}_{i} - 1$ servers assigned to $\greedk^{(i)}$. 
        This implies
		\begin{flalign}
		{\hat X}^{a}_{i} \greedki^{(i)}_j \ge Y_j 
		- \sum_{\ell=1}^{i-1} {\hat X}_{\ell} \greedki^{(\ell)}_j, \mbox{ and }  ({\hat X}^{a}_{i} - 1) \greedki^{(i)}_j < Y_j
		- \sum_{\ell=1}^{i-1} {\hat X}_{\ell} \greedki^{(\ell)}_j.\nonumber 
		\end{flalign}
		We can use similar argument after the event when $Y_j$ changes to $Y_j+\zeta$, i.e.,
		\begin{flalign}
        {\hat X}^{b}_{i} \greedki^{(i)}_j \ge Y_j + \zeta
		- \sum_{\ell=1}^{i-1} {\hat X}_{\ell} \greedki^{(\ell)}_j,\  ({\hat X}^{b}_{i} - 1) \greedki^{(i)}_j < Y_j + \zeta
		- \sum_{\ell=1}^{i-1} {\hat X}_{\ell} \greedki^{(\ell)}_j. \label{eq:dXbase} \nonumber
		\end{flalign}
		Algebraic manipulations based on this set of equations shows 
		\ben
		& |{\hat X}^{a}_{i} - {\hat X}^{b}_{i}| < \frac{|\zeta| + 1}{\greedki^{(i)}_j} \le 2,
		\een
		or equivalently $|{\hat X}^{a}_{i} - {\hat X}^{b}_{i}| \le 1$.
		
		Now consider $i > i^\star$ and suppose for $\ell = i^\star, \ldots, i-1$, 
		$|X^{a}_{\ell} - X^{b}_{\ell}| \le (1+\Kmax)^{\ell-i}$. 
		Similar to the arguments for the base case, the following equations have to hold for a job type $j^\prime \in \mathcal{J}$,
		\begin{flalign*}
		 &\textstyle \sum_{\ell=i^\star}^{i} {\hat X}^{a}_{\ell} 
		\greedki^{(\ell)}_{j^\prime} \ge Y_{j^\prime} 
		- \sum_{\ell=1}^{i^\star-1} {\hat X}_{\ell} \greedki^{(\ell)}_{j^\prime}, &\\
		 & \textstyle \sum_{\ell=i^\star}^{i-1} {\hat X}^{a}_{\ell} 
		\greedki^{(\ell)}_{j^\prime} 
		+ ({\hat X}^{a}_{i} - 1)
		\greedki^{(i)}_{j^\prime} < Y_{j^\prime}
		- \sum_{\ell=1}^{i^\star-1} {\hat X}_{\ell} \greedki^{(\ell)}_{j^\prime}, &\\
		 & \textstyle \sum_{\ell=i^\star}^{i} {\hat X}^{b}_{\ell} 
		\greedki^{(\ell)}_{j^\prime} \ge Y_{j^\prime} 
		- \sum_{\ell=1}^{i^\star-1} {\hat X}_{\ell} \greedki^{(\ell)}_{j^\prime}, &\\
		 & \textstyle \sum_{\ell=i^\star}^{i-1} {\hat X}^{b}_{\ell} 
		\greedki^{(\ell)}_{j^\prime} 
		+ ({\hat X}^{b}_{i} - 1)
		\greedki^{(i)}_{j^\prime} < Y_{j^\prime}
		- \sum_{\ell=1}^{i^\star-1} {\hat X}_{\ell} \greedki^{(\ell)}_{j^\prime}.&
		\end{flalign*}
		With algebraic manipulations, we get
		\ben
		& |{\hat X}^{a}_{i} - {\hat X}^{b}_{i}| <
		1 + \sum_{\ell=i^\star}^{i-1} 
		\frac{\greedki^{(\ell)}_{j^\prime}}{\greedki^{(i)}_{j^\prime}} 
		|X^{a}_{\ell} - X^{b}_{\ell}|.
		\een
        Hence, considering $\frac{\greedki^{(\ell)}_{j^\prime}}{\greedki^{(i)}_{j^\prime}} \le \Kmax$ for $\ell=i^\star,\ldots, i-1$, we get
		\ben
		|{\hat X}^{a}_{i} - {\hat X}^{b}_{i}| \le
		1 + \sum_{\ell=i^\star}^{i-1} 
		\Kmax (1+\Kmax)^{\ell-i^\star} = (1+\Kmax)^{i-i^\star}.
		\een
		
		The result for ${X}^{L}_{\greedk^{(i)}}(t)$ then follows by noticing:
		\begin{enumerate}[leftmargin=*]
			\item 
            If ${X}^{L}_{\greedk^{(i)}}(t) \ge {\hat X}^{L}_{\greedk^{(i)}}(t)$ then 
				after an event ${X}^{L}_{\greedk^{(i)}}(t)$ may 
				not increase more than what ${\hat X}^{L}_{\greedk^{(i)}}(t)$
				does, which is at most $(\Kmax+1)^{i-1}$.
				Similarly, it may not decrease more than the increase 
				of ${X}^{L}_{\greedk^{(\ell)}}(t)$ for $\ell=1, \ldots, i-1$
				which is at most
				\ben
				& \sum_{\ell=1}^{i-1} (\Kmax+1)^{\ell-1} < (\Kmax+1)^{i-1},
				\een 
				or more than the decrease of ${\hat X}^{L}_{\greedk^{(i)}}(t)$
				which is again $(\Kmax+1)^{i-1}$. 
				Notice that the last claim assumes ${X}^{L}_{\greedk^{(i)}}(t) = {\hat X}^{L}_{\greedk^{(i)}}(t)$, 
                because in case ${X}^{L}_{\greedk^{(i)}}(t) > 
				{\hat X}^{L}_{\greedk^{(i)}}(t)$ it means 
				server of Reject Group assigned to $\greedk^{(i)}$ 
                is not empty so maximum decrease of ${X}^{L}_{\greedk^{(i)}}(t)$ is $1$	when that server empties.
			\item 
            If ${X}^{L}_{\greedk^{(i)}}(t) < {\hat X}^{L}_{\greedk^{(i)}}(t)$ then no server not assigned to
			a configuration $\greedk^{(\ell)}$ for $\ell=1, \ldots, i-1$
			will be empty. Then	after an event ${X}^{L}_{\greedk^{(i)}}(t)$ 
			may not increase more than the decrease of 
			${X}^{L}_{\greedk^{(\ell)}}(t)$ for $\ell=1, \ldots, i-1$
			which is at most 
			\ben
			& \sum_{\ell=1}^{i-1} (\Kmax+1)^{\ell-1} \le (\Kmax+1)^{i-1}-1.
			\een
			and decrease of ${X}^{L}_{\mathbf{k}}(t)$ with 
			$\mathbf{k} \not \in \{\greedk^{(\ell)}: \ell=1, \ldots, i\}$,
			which is at most $1$ since none of them was empty and at most
			one may empty after each event.
			Thus, the total decrease of all servers that may be reassigned to
			$\greedk^{(i)}$ is no more than $(\Kmax+1)^{i-1}$.
			Also decrease is at most $(\Kmax+1)^{i-1}$ following the same argument as in the case ${X}^{L}_{\greedk^{(i)}}(t) \ge {\hat X}^{L}_{\greedk^{(i)}}(t)$.
			\item 
            For any $\mathbf{k} \not \in 
			\{\greedk^{(i)}:i=1,\ldots,\gcnt-1\}$, $X^L_\mathbf{k}$ may only
			decrease and the decrease will be at most
			\ben
			& \sum_{\ell=1}^{\gcnt-1} (\Kmax+1)^{\ell-1} < (\Kmax+1)^{\gcnt-1}.
			\een 
		\end{enumerate}
		Finally, it trivially follows that the maximum change of
		${X}^{L(e)}_{\greedk^{(i)}}(t)$ is $(\Kmax+1)^{i-1}$ as well, 
		for $i=1,\ldots, \gcnt-1$, since
		$$\mathbf{X}^{L(e)}_{\greedk^{(i)}}(t) = 
		\min(\mathbf{\hat X}^{L}_{\greedk^{(i)}}(t), 
		\mathbf{X}^{L}_{\greedk^{(i)}}(t)).$$
	\end{proof}

	We can now prove the existence of fluid limits of the process
	${X}^{L(e)}_{\assgnk}(t)$, for ${\assgnk}=\greedk^{(i)}$, $i =1, \ldots, \gcnt.$
	For each job type $j$, we define two independent unit-rate 
	Poisson processes $\Pi^a_i(\cdot)$ and $\Pi^d_i(\cdot)$. 
	By the Functional Strong Law of Large Numbers, almost surely,
	\begin{equation}\label{eq:uoc}
	\frac{\Pi^a_i(Lt)}{L} \to t, \quad u.o.c. \quad
	\frac{\Pi^d_i(Lt)}{L} \to t, \quad u.o.c.
	\end{equation}
	where u.o.c means uniformly over compact time intervals. 
	
	Define $h^a_{j, \assgnk}(\mathbf{S}^L(t))$ and  $h^d_{j, \assgnk}(\mathbf{S}^L(t))$ to be the amount of change in 
	${X}^{L(e)}_{\assgnk}(t)$ due to an arrival and departure of a type-$j$ job, respectively, at state $\mathbf{S}^L(t)$. 
	Then the process ${X}^{L(e)}_{\assgnk}(t)$ can be described as
	\begin{equation}\label{eq:nkL}
	{X}^{L(e)}_{\assgnk}(t) = {X}^{L(e)}_{\assgnk}(0) +	A^{L}_{\assgnk}(0,t) - D^{L}_{\assgnk}(0,t)
	\end{equation}
	where, for any $0 \le t_1 < t_2 $, without loss of generality, we construct the arrival and departure processes for the $L$-th system, and the corresponding jumps, as 
	\begin{align*}
	A^{L}_{\assgnk}(t_1,t_2)&=\sum_{j=1}^{J} \sum_{n=1}^{\Pi^a_j (
	\int_{t_1}^{t_2}  \lambda_j L ds)}h^a_{j, \assgnk}(\mathbf{S}^L(T_n)),\\
		D^{L}_{\assgnk}(t_1,t_2)& =\sum_{j=1}^{J} \sum_{n=1}^{\Pi^d_j (
	\int_{t_1}^{t_2} \mu_j Y^L_j(s) ds)} h^d_{j, \assgnk}(\mathbf{S}^L(T_n)),
    \end{align*}
	where $T_n$ is the time of the $n$-th jump in corresponding Poisson processes.
	By Lemma~\ref{lem:X_change}, $|h^{a}_{j, \assgnk}(\mathbf{S}^L(t))|, |h^{b}_{j, \assgnk}(\mathbf{S}^L(t))|\leq (1 + \Kmax)^{\gcnt-1}:=M$. 
	Then the scaled processes $\frac{1}{L}A^{L}_{\assgnk}(t_1,t_2)$ and $\frac{1}{L}D^{L}_{\assgnk}(t_1, t_2)$ in (\ref{eq:nkL}) are asymptotically Lipschitz continuous, 
    which implies that they have a convergent subsequence \cite{ethier2009markov}. 
    This is because for any $t_1 < t_2$, 
	\begin{equation}
	\begin{aligned}
	&\limsup_{L} \frac{1}{L} \left \lvert A^{L}_{\assgnk}(t_1,t_2) \right \rvert \le 
	\limsup_{L} \frac{1}{L} \Pi^a_j \left(
	\int_{t_1}^{t_2} \lambda_j L ds
	\right)M  \\
	&= \limsup_{L} \frac{1}{L} \Pi^a_j \left(
	\lambda_j L (t_2 - t_1) 
	\right)M  = \lambda_j M(t_2-t_1),
	\end{aligned}
	\end{equation}
	where we used \dref{eq:uoc} to get almost sure convergence. 
    We can similarly bound $\frac{1}{L} D^{L}_{\assgnk}(t_1,t_2) $ by noting that $Y^L_j(s) \le L \Kmax$. 
    Hence, with the stated initialization, the scaled process
	${X}^{L(e)}_{\assgnk}(t)/L$ converges to a Lipschitz continuous sample path ${{x}}_{\assgnk}^{(e)}(t)$ along the subsequence \cite{ethier2009markov}. 
	Similarly, it can be shown that the fluid limits of processes ${x}^{L}_{\assgnk}(t)$ and $\hat{x}^{L}_{\assgnk}(t)$  
    exist and they are Lipschitz continuous.
	
	Similarly, ${Y}^L_j(t)$ increases by at most 1 every time a type-$j$ job arrives and decreases by 1 every time a type-$j$ job in the system departs. 
    Hence, the limit of ${y}_j^{(L)}(t)$ also exists by asymptotic Lipschitz continuity.

%% file: convergence_proof3.tex
\section{Proof of Lemma~\ref{lem:S}} \label{proof:lem:S}

	For each job type $j$, the number of type-$j$ jobs in the system is bounded by the number of type-$j$ jobs in an $M/M/\infty$ system where all arrivals are accepted. This implies that $y_j(t)$ is also bounded by the fluid limit of type-$j$ jobs in the $M/M/\infty$ system, i.e.,
	\begin{equation}\label{eq:yj_fluid}
	y_j(t) \le y_j(0) + \lambda_j t - \int_0^t y_j(s) \mu_j ds.
	\end{equation}
	This implies $y_j(t) \le \rho_j + (y_j(0) - \rho_j) e^{-\mu_j t}$.
	Considering that for any initial state $\statez(0)$, 
	$y_j(0) \leq \Kmax$, we can get that
	$y_j(t) < \rho_j + \epsilon_\rho$ if 
	$t > T_{\epsilon_\rho, j}$ where
	$T_{\epsilon_\rho, j} = \frac{ -\log{\epsilon_\rho} + \log{\Kmax - \rho_j} }{\mu_j}$.
	Finally, we can choose $
	T_{\epsilon_\rho} := \max_{j \in \mathcal{J}} T_{\epsilon_\rho, j}.$

%% file: bound_proof.tex
\section{Bounds on the Change of Scaled Processes}\label{proof:bound}
    In this section, we provide a few lemmas which will be used in the proofs later. 
    Their proofs are straightforward and based on concentration inequalities for Poisson distribution.
    \begin{lemma}\label{lem:Chern_Poisson}
    	Consider a time interval $[\tau_a, \tau_b)$,
        and a Poisson process $N$, with $N[\tau_a, \tau_b]$ being the number 
    	of events of the process in $[\tau_a, \tau_b)$, and function $\subL$ as given in Definition~\ref{def:sublen}. 
        Then we have:
    	
    	If rate of $N$ is at least $L\lambda$ and length of $[\tau_a, \tau_b)$
    	is at least $c\subL$,  
    	\begin{equation}
    	\mathds{P} \Big(N[\tau_a, \tau_b] > L\lambda c \subL + o(L\subL) \Big)
    	\ge 1 - o(L^{-2}).
    	\end{equation}
    	If $N$ has rate exactly $L\lambda$ and length of $[\tau_a, \tau_b)$ is
    	at least $c\subL$,
    	\begin{equation}
    	\mathds{P} \Big(N[\tau_a, \tau_b] = L\lambda c \subL + o(L\subL) \Big)
    	\ge 1 - o(L^{-2}).
    	\end{equation} 
    	Lastly if $N$ has rate at most $L \lambda$, and length of $[\tau_a, \tau_b)$ is
    	at most $c\subL$,
    	\begin{equation}
    	\mathds{P} \Big(N[\tau_a, \tau_b] < L \lambda c \subL + o(L\subL) \Big)
    	\ge 1 - o(L^{-2}).
    	\end{equation}
    \end{lemma}
    \begin{proof}
    The proofs of all the cases are based on the tail bounds of Poisson distribution. Specifically, we use the following bounds~\cite{mini}.

	For a Poisson random variable $X$ with mean $\lambda$ we have
	\ben
&	\mathds{P} \left(
	X - \lambda > x 
	\right)  \le \exp\left(-\frac{x^2}{2(\lambda+x)}\right),\\
&	\mathds{P} \left(
	X - \lambda < -x 
	\right)  \le \exp\left(-\frac{x^2}{2(\lambda+x)}\right),\\
&	\mathds{P} \left(
	\lvert X - \lambda\rvert > x 
	\right)  \le 2 \exp\left(-\frac{x^2}{2(\lambda+x)}\right).
	\een

Then, in the case that the rate is at least $L \lambda$ 
and length of interval is at least $c\subL$, we have that
for any $\epsilon^\prime > 0$
\begin{equation}
\mathds{P} \left(
N[{\tau}_{a}, {\tau}_{b}] - c \subL \lambda L >
\epsilon^\prime \subL L
\right) < 
\exp\left(-\frac{(\epsilon^\prime)^2 \subL L}
{2c \lambda + 2\epsilon^\prime}\right).
\end{equation}
Last expression is $o(L^{-2})$ which can be shown by taking its 
logarithm and using the fact that
$\log L = o(\subL L)$ by Definition~\ref{def:sublen}. 

Since $\epsilon^\prime$ was arbitrary, we eventually get 
\begin{equation}
\mathds{P} \left(
N[{\tau}_{a}, {\tau}_{b}] > c \subL \lambda L +
o (\subL L)
\right) < o(L^{-2}).
\end{equation}
Other cases can be shown in a similar way.
\end{proof}

\begin{lemma}\label{lem:rbound}
	Consider a time interval $[\tau_a, \tau_b) \subset [t_n, t_{n+1})$,
	with $t_n$ defined in Section~\ref{sec:subintervals}. 
	Assume that the interval is of length at 
	most $\epsilon\subL$, for function $\subL$ as in 
	Definition~\ref{def:sublen}, and constant $\epsilon > 0$ sufficiently small.
	Then, with probability 
	at least $1 - o(L^{-2})$,
	\begin{equation}\label{eq:r_bound}
	x^{L(e)}_{\greedk^{(i)}} (\tau_{b}) - x^{L(e)}_{\greedk^{(i)}} (\tau_{a})
	> - B_i \epsilon \subL + o(\subL),
	\end{equation}
	and
	\begin{equation}\label{eq:r_bound2}
	{q^{L}_{\greedk^{(i)}, j}(\tau_{b}) - q^{L}_{\greedk^{(i)}, j}(\tau_{a})}
	> - B_{i+1} \epsilon L\subL + o(L\subL)
	\end{equation}
	where $B_i:= (\Kmax+1)^{i-1} 2 (\Kmax \mumax + \sum_{i=1}^J \lambda_j)$.
\end{lemma}

\begin{proof}[Proof of Lemma~\ref{lem:rbound}]
	The state changes only at arrivals and departures. By definition of $\tau_n$, $\tau_0 = t$, and by Lipschitz continuity of $\mathbf{y}(t)$, for any time $\tau \in [\tau_n, \tau_{n+1}]$, and $j \in \mathcal{J}$, $y^L_j(\tau) = y_j(t)+O(\epsilon)$, almost surely, for $L$ large enough along the subsequence.
	
	Let $N_{ad}[{\tau}_{a}, {\tau}_{b}]$ be the number of
	arrival or departure events of any job type in the interval
	$[{\tau}_{a}, {\tau}_{b})$. 
	This process is Poisson with rate at most 
	$L (\mu_j y_j(t) +O(\epsilon)+ \lambda_j) < L \hat{R}$,
	where $\hat{R}:= 2(\Kmax\mumax +\sum_{j=1}^J \lambda_j)$. Also
	$\epsilon \subL$ is an upper bound on length of interval 
	$[{\tau}_{a}, {\tau}_{b})$, so by applying Lemma~\ref{lem:Chern_Poisson},
	we have
	\begin{equation}
	\mathds{P} \Big(
	N_{ad}[{\tau}_{a}, {\tau}_{b}] - \epsilon \hat{R} \subL L  >
	o(\subL L)
	\Big) < o(L^{-2}).
	\end{equation}

    Now suppose the event $N_{ad}[{\tau}_{a}, {\tau}_{b}] <
	\epsilon \hat{R} \subL L + o(\subL L)$ holds.
	The absolute change that occurs to variables 
	$X^{L(e)}_{\greedk^{(i)}}(\tau)$ after each event 
	is at most $(\Kmax+1)^{i-1}$ according to Lemma~\ref{lem:X_change}, hence
	\begin{equation*}
	\begin{aligned}
	&{X^{L(e)}_{\greedk^{(i)}} (\tau_{b}) 
		- X^{L(e)}_{\greedk^{(i)}} (\tau_{a})} \ge 
	- {N_{ad}[{\tau}_{a}, {\tau}_{b}] (\Kmax+1)^{i-1}}
	\ge \\ 
	&- \hat{R} (\Kmax+1)^{i-1} \epsilon L\subL + o({L\subL}) =
	-B_i \epsilon L \subL + o({L\subL}).
	\end{aligned}
	\end{equation*}
	Dividing both sides by $L$ we get \dref{eq:r_bound}. 
	Based on (\ref{eq:q_def}), we can write
	\begin{equation}
	\begin{aligned}
	&{q^{L}_{\greedk^{(i)}, j}(\tau_{b}) 
		- q^{L}_{\greedk^{(i)}, j}(\tau_{a})} = \\
	&\sum_{\ell=1}^i \greedki^{(\ell)}_j
	\left( 
	X^{L(e)}_{\greedk^{(\ell)}}(\tau_{b}) - X^{L(e)}_{\greedk^{(\ell)}}(\tau_{a})
	\right)
	- Y^L_{j}(\tau_b) + Y^L_{j}(\tau_a).
	\end{aligned}
	\end{equation}
	Then (\ref{eq:r_bound2}) follows by considering
	\begin{equation*}
	\begin{aligned}
	& \sum_{\ell=1}^i \greedki^{(\ell)}_j 
	\left(X^{L(e)}_{\greedk^{(i)}} (\tau_{b}) 
	- X^{L(e)}_{\greedk^{(i)}} (\tau_{a}) \right)\ge 
	- \Kmax \sum_{\ell=1}^i B_\ell \epsilon L\subL + o({L\subL}).
	\end{aligned}
	\end{equation*}
	and
	\begin{equation}
	\begin{aligned}
	& - {Y^{L}_{j}(\tau_{b}) + Y^{L}_{j}(\tau_{a})} \ge
    - \hat{R} \epsilon L\subL + o({L\subL}).
	\end{aligned}
	\end{equation}
\end{proof}

\begin{lemma}\label{lem:emptyrate} 
	Consider function $\subL$ as in Definition~\ref{def:sublen}
	and, a time interval $[\tau_a, \tau_b)$ with length at most $\subL$.
	Suppose for some $i \in \{1, \ldots, \gcnt-1\}$, 
    we have that at any time $\tau \in [\tau_a, \tau_b)$
	\begin{equation}\label{eq:ql0}
	\max_{j \in \calJ: \greedki^{(i)}_j>0} 
	q^L_{\greedk^{(i)}, j}(\tau) < 0.
	\end{equation}
	Let $N_e[\tau_a, \tau_b]$ be the number of times that 
    servers in $\overline{\texttt{RG}}(i)$ 
    empty during $[\tau_a, \tau_b)$. 
	Then with probability $1 - o(L^{-2})$,
	\begin{equation}\label{eq:empty}
	\frac{N_e[\tau_a, \tau_b]}{L(\tau_b - \tau_a)} >
	\frac{\mumin}{J \Kmax \cfcnt^2}
	\left(1 - \sum_{\ell=1}^{i} x^{(e)}_{\greedk^{(\ell)}}(t) \right)
	+ \frac{o(\subL)}{\tau_b - \tau_a}.
	\end{equation}
\end{lemma}
\begin{subproof}

	Let $p:=1 - \sum_{\ell=1}^{i} x^{(e)}_{\greedk^{(\ell)}}(t)$,
    which is strictly positive and is the number of servers without effective configuration in $\greedk^{(\ell)}$ for $\ell=1, \ldots, i$.
    Notice that due to~\dref{eq:ql0}, rank $1$ servers in this set 
    will always belong to Reject Group servers $\overline{\texttt{RG}}(i)$ (Definition~\ref{def:rgi}).

	Since there are at most $\cfcnt$ different configurations, 
    one of the configurations, say $\bk$,  is assigned to at least 
    $\frac{p}{\cfcnt}$ servers without effective configuration
    in $\greedk^{(\ell)}$ for $\ell=1, \ldots, i$ at the fluid limit 
    at time $t$, i.e.,
    ${x}_\bk(t) \geq p^\prime := \frac{p}{\cfcnt} > 0$. 
    We also define set $\mathcal{J}^\star := \{j:k_j > 0\}$
    with cardinality $J^\star$. 

    Due to migrations performed on departure instances,	at least one server in $\overline{\texttt{RG}}(i)$ gets empty, when, for every $j \in \mathcal{J}^\star$, 
    the number of type-$j$ jobs that departs is at least the number of type-$j$ jobs that are in the servers of $\overline{\texttt{RG}}(i)$. 
    This is because any type-$j$ departure will create
	a new empty type-$j$ slot in the servers of $\overline{\texttt{RG}}(i)$, 
    if there is a type-$j$ job in any of them.
    Hence one of the servers in this set will empty after at most $\cfcnt \Kmax$ jobs of each job type in $\mathcal{J}^\star$ depart, 
    where $\cfcnt \Kmax$ is an upper bound on the total number of jobs of any 
    type that can be in Reject Group servers.
    Then we need to bound
	\begin{equation}\label{eq:empty_equiv2}
	\begin{aligned}
	&\mathds{P} \left(
	N_e[\tau_a, \tau_b] < \frac{p^\prime \mumin L}{J \Kmax \cfcnt}
	(\tau_b - \tau_a) - \epsilon L \subL
	\right).
	\end{aligned}
	\end{equation}
	
	If 
    $\tau_b - \tau_a < \frac{\epsilon \subL J \Kmax \cfcnt}{p^\prime \mumin},$
    then
	\ben
	&\mathds{P} \left(
	N_e[\tau_a, \tau_b] - p^\prime \mumin L
	\frac{\tau_b - \tau_a}{J} < - \epsilon \Kmax \cfcnt L \subL
	\right) \le \\
    &\mathds{P} \left(N_e[\tau_a, \tau_b] < 0 \right) = 0.
	\een
    If instead $\tau_b - \tau_a \ge \frac{\epsilon \subL J \Kmax \cfcnt}{p^\prime \mumin},$ then we consider the counting 
    process $D_{\mathcal{J}^\star}(\tau)$ defined as follows.
    $D_{\mathcal{J}^\star}(\tau_a) = 0$ and it
    is incremented at times $\tau^{(i)}$ for $i \in \mathds{Z}^+$ if
    $\tau^{(i)}$ is the first time since $\tau^{(i-1)}$ at which
    at least one departure of type $j$ occurred for all 
    $j \in \mathcal{J}^\star$ and as a convention $\tau^{(0)} = \tau_a$.
    Based on arguments so far, the process $N_e[\tau_a, \tau_b]$
    will increment by at least $1$ between times
    $\tau^{(i)}$ and $\tau^{(i+\cfcnt \Kmax)}$, for $i \ge 0$,
    i.e., $N_e[\tau_a, \tau_b] \ge \left\lfloor
    \frac{D_{\mathcal{J}^\star}(\tau_b)}{\cfcnt \Kmax}
    \right\rfloor$. 
    Then applying Lemma~\ref{lem:Chern_Poisson} to 
    process $D_{\mathcal{J}^\star}$, which has rate at least
    $\frac{p^\prime \mumin L}{J}$, we get
    \ben
	&\mathds{P} \left(
	N_e[\tau_a, \tau_b] - \frac{p^\prime \mumin L}{J \cfcnt \Kmax} 
	(\tau_b - \tau_a) < - \epsilon \Kmax \cfcnt L \subL
	\right) \le \\ 
    &\mathds{P} \left(
    \left\lfloor
    \frac{D_{\mathcal{J}^\star}(\tau_b)}{\cfcnt \Kmax}
    \right\rfloor 
    - \frac{p^\prime \mumin L}{J \cfcnt \Kmax} 
    (\tau_b - \tau_a) < - \epsilon \Kmax \cfcnt L \subL
    \right)
    = o(L^{-2}).
	\een

	We have thus proven, since $\epsilon$ can be arbitrarily close to $0$,
	that with probability $1 - o(L^{-2})$
	\begin{equation}
	N_e[\tau_a, \tau_b] > \frac{p^\prime \mumin L}{J \cfcnt \Kmax }
	(\tau_b - \tau_a) + o(L \subL),
	\end{equation}
	which implies \dref{eq:empty}.	

\end{subproof}

%% file: convergence_proof_of_drift5.tex
%%%%%%%%%%%%%%%%%%%%%%%%%%%%%%%%%%%%%%%%%%%%%%%%%
\section{Proof of Proposition~\ref{prop:properties}}\label{proof:prop:properties}

\input{convergence_anymu}

%% file: convergence_anymu.tex
We prove Proposition~\ref{prop:properties} 
for the following values of parameters:
\begin{itemize}[leftmargin=*]
	\item $\alpha_i$, $i \in \{1, \ldots, \gtcnt-1\}$, is given by 
	\begin{equation}\label{eq:alphadef}
	\alpha_i := \delta \frac{\mumin v^{i-\gtcnt} v_i}
	{\mumax 8 \Kmax^2 \cfcnt},
	\end{equation}
	where $\Kmax$ is the maximum number of jobs in a server, 
	$\mumin := \min_{j \in \mathcal{J}} \mu_j$, 
	$\mumax := \max_{j \in \mathcal{J}} \mu_j$, 
	$v := 12 \Kmax \frac{\mumax}{\mumin}$,
	$v_i := 1$ if $i \le \gtcntp{\alpha}(t)$ else
	$v_i := \Kmax$ and $\delta$ is a positive constant 
	sufficiently small such that
	\begin{equation}\label{eq:deltacond}
	\delta < \frac{1}{2J \cfcnt}\globlt{x}^{(g)}_{\greedk^{(\gtcnt)}}  
	\end{equation}
	and for any two indexes $j_a, j_b \in \{1, \ldots, \gcnt\}$,
	\begin{equation}\label{eq:non_trivial}
	\rho_{j_a} > \sum_{i=1}^{j_b} \globlt{x}^{(g)}_{\greedk^{(i)}} 
	\greedki^{(i)}_{j_a} + \delta,
	%\end{equation}
	\text{  or  }
	%\begin{equation}\label{eq:trivial}
	\rho_{j_a} = \sum_{i=1}^{j_b} \globlt{x}^{(g)}_{\greedk^{(i)}} 
	\greedki^{(i)}_{j_a}.
	\end{equation}
	\item $\epsilon_\rho$ is chosen as
	\begin{equation}\label{eq:epsrho_cond}
	\epsilon_\rho = \frac{\alpha_1}{2}.  
	\end{equation} 
\end{itemize}

Before presenting the main proof, we state a few lemmas.
%The following Lemmas bound the rates that servers 
%in different configurations empty and are assigned.

\begin{lemma}\label{lem:rank0frac}
	The fraction of servers without effective configuration
	in the set $\{\greedk^{(i)}: i=1, \ldots, \gtcntp{\alpha}(t)\}$,
	with $\gtcntp{\alpha}(t)$ as in Definition~\ref{def:gprime},
	is at least $\frac{1}{2} \globlt{x}^{(g)}_{\greedk^{(\gtcnt)}}$.
\end{lemma}
\begin{subproof}
	The bound can be inferred as follows
	\begin{equation*}
	\begin{aligned}
	& 1 - \sum_{\ell=1}^{i} x^{(e)}_{\greedk^{(\ell)}}(t) \ge
	x^{(g)}_{\greedk^{(\gtcnt)}} + 
	\sum_{\ell=1}^{i} \left(\globlt{x}^{(g)}_{\greedk^{(\ell)}} 
	- x^{(e)}_{\greedk^{(\ell)}}(t)\right)
	> \\
	& x^{(g)}_{\greedk^{(\gtcnt)}} 
	- \sum_{\ell=1}^{i} \alpha_{\ell} >
	x^{(g)}_{\greedk^{(\gtcnt)}} 
	- \sum_{\ell=1}^{i} \frac{\delta}{2\gtcnt}
	\stackrel{}{>}
	x^{(g)}_{\greedk^{(\gtcnt)}}/2.
	\end{aligned}
	\end{equation*}
	To get this result we used that $i \le \gtcntp{\alpha}(t) < \gtcnt$,
	$\sum_{\ell=1}^{\gtcnt} \globlt{x}^{(g)}_{\greedk^{(\ell)}}=1$,
	$\alpha_i < \frac{\delta}{2\gtcnt}$ from (\ref{eq:alphadef}) 
	and (\ref{eq:deltacond}).
\end{subproof}

\begin{lemma}\label{lem:hrec}
	Consider an interval $[\tau_a, \tau_b) \subseteq [\tau^{(m)}_n, \tau^{(m+1)}_n)$,
	with $\tau^{(m)}_n$ being defined in Section~\ref{sec:subintervals}
	and the corresponding driving set of indexes
	$\subind[m] = \{\ji{i}: i=1, \ldots, G\}$.
	If the following holds for $i=1, \ldots, G$
	\begin{equation}\label{eq:p1}
	\sum_{\ell=1}^i \greedki^{(\ell)}_{\ji{i}} 
	\nabla {x}^{L(e)}_{\greedk^{(\ell)}}[\tau_a, \tau_b] = 
	\lambda_{\ji{i}} - \mu_{\ji{i}} y_{\ji{i}}(t) 
	+ \frac{o(\subL)}{\tau_b - \tau_a},
	\end{equation}
	then
	\begin{equation}\label{eq:xe_h}
		\nabla {x}^{L(e)}_{\greedk^{(i)}}[\tau_a, \tau_b] = 
		h^{\subind,(i)}(t) + \frac{o(\subL)}{\tau_b - \tau_a},
	\end{equation}
\end{lemma}
\begin{subproof}
The proof is by induction.	For $i=1$ we have
	\begin{equation*}
	\begin{aligned}
	&\nabla {x}^{L(e)}_{\greedk^{(1)}}[\tau_a, \tau_b] = \\
	&\frac{\lambda_{\ji{1}} - \mu_{\ji{1}} y_{\ji{1}}(t)}{\greedki^{(1)}_{\ji{1}}} 
	+ \frac{o(\subL)}{\tau_b - \tau_a}
	= h^{\subind, (1)}(t) + \frac{o(\subL)}{\tau_b - \tau_a}.
	\end{aligned}
	\end{equation*}
	
	Now  assume that for every $\ell \in \{1, \ldots, i-1\}$,
	\begin{equation*}
	\nabla {x}^{L(e)}_{\greedk^{(\ell)}}[\tau_a, \tau_b] = 
	h^{\subind, (\ell)}(t) + \frac{o(\subL)}{\tau_b - \tau_a}.
	\end{equation*}
	
	Then (\ref{eq:p1}) implies
	\begin{equation*}
	\begin{aligned}
	&\nabla {x}^{L(e)}_{\greedk^{(i)}}[\tau_a, \tau_b] = \\
	&\frac{\lambda_\ji{i} - \mu_\ji{i} y_\ji{i}(t) - \sum_{\ell=1}^{i-1} \greedki^{(\ell)}_\ji{i}
		\nabla {x}^{L(e)}_{\greedk^{(\ell)}}[\tau_a, \tau_b]}
	{\greedki^{(i)}_{\ji{i}}} 
	+ \frac{o(\subL)}{\tau_b - \tau_a} \stackrel{(a)}{=} \\
	&\frac{\lambda_\ji{i} - \mu_\ji{i} y_\ji{i}(t) - \sum_{\ell=1}^{i-1} \greedki^{(\ell)}_\ji{i}
		h^{\subind, (\ell)}(t)}
	{\greedki^{(i)}_{\ji{i}}} 
	+ \frac{o(\subL)}{\tau_b - \tau_a} = \\
	&h^{\subind, (i)}(t) + \frac{o(\subL)}{\tau_b - \tau_a}.
	\end{aligned}
	\end{equation*}
	In (a), we used the fact that sum of a finite number of $\frac{o(\subL)}{\tau_b - \tau_a}$ terms is still $\frac{o(\subL)}{\tau_b - \tau_a}$
\end{subproof}

\begin{lemma}\label{lem:hbound}
Suppose	$\statez(t) \in \setConv[\epsilon_\rho]$, and for every $i \in \{1, \ldots, G\}$,
	$h^{\subind[m], (i)}(t)$ is defined as in (\ref{eq:h_system}),
	and $G \le \gtcntp{\alpha}(t)$ with $\gtcntp{\alpha}(t)$
	as in Definition~\ref{def:gprime}. Then
	\begin{equation}\label{eq:hbound}
	\vert h^{\subind[m], (i)}(t) \vert < 
	\mumax \alpha_i +
	\mumax \sum_{\ell=1}^{i-1} 2(\Kmax + 1)^{i-\ell} \alpha_\ell.
	\end{equation}
\end{lemma}
\begin{subproof}
	Let  $\subind=\subind[m]$. The proof is by induction.
	
	\textbf{Base Case:}
	It suffices to show that $h^{\subind, (1)}(t) < \mumax \alpha_1$
	and $h^{\subind, (1)}(t) > -\mumax \alpha_1$.
	
	To show $h^{\subind, (1)}(t) < \mumax \alpha_1$, consider $\ji{1}$ which is the first index of $\subind$, then
	\begin{equation*}
		h^{\subind, (1)}(t) = \frac{\lambda_{\ji{1}} - \mu_{\ji{1}} y_{\ji{1}}(t)}{\greedki^{(1)}_{\ji{1}}} \stackrel{(a)}{=}
		\mu_\ji{1}\left(x^{(g)}_{\greedk^{(1)}} 
		- x^{(e)}_{\greedk^{(1)}}(t)\right) <
		\mumax \alpha_1.
	\end{equation*}
	In (a), we used (\ref{eq:defcalJ}) for $i=1$,
	according to which 
	$\rho_{\ji{1}} = \greedki^{(1)}_{\ji{1}} x^{(g)}_{\greedk^{(1)}}$ 
	and $y_{\ji{1}}(t) = \greedki^{(1)}_{\ji{1}} x^{(e)}_{\greedk^{(1)}}(t)$.
	
	To show $h^{\subind, (1)}(t) > - \mumax \alpha_1$, we use the fact that $y_{\ji{1}}(t) - \rho_{\ji{1}} < \epsilon_\rho$ (since $\statez(t) \in \setConv[\epsilon_\rho]$) which can be applied as
	\begin{equation}
	h^{\subind, (1)}(t) = \frac{\lambda_{\ji{1}} - \mu_{\ji{1}} y_{\ji{1}}(t)}{\greedki^{(1)}_{\ji{1}}} {>}
	- \frac{\mu_{\ji{1}} \epsilon_\rho}{\greedki^{(1)}_{\ji{1}}} 
	\stackrel{(a)}{\ge}
	- \mumax \alpha_1.
	\end{equation}
	In (a), we used that $\epsilon_\rho < \alpha_1$,
	which is due to (\ref{eq:epsrho_cond}).

	\textbf{Inductive Case:} 
    We assume (\ref{eq:hbound}) is true for all indexes up to $i-1$. 
    Then we can upper and lower bound $h^{\subind, (i)}(t)$ as follows.
	
	To show $h^{\subind, (i)}(t) < \mumax \alpha_i + \mumax\sum_{\ell=1}^{i-1} 2(\Kmax + 1)^{i-\ell} \alpha_\ell$, 
    let $\ji{i}$ be the $i$th index in $\subind$, then
	\begin{equation*}
	\begin{aligned}
	& h^{\subind, (i)}(t) = 
	\frac{\lambda_{\ji{i}} - \mu_{\ji{i}} y_{\ji{i}}(t) - \sum_{\ell=1}^{i-1} \greedki^{(\ell)}_{\ji{i}} h^{\subind, (\ell)}(t)}
	{\greedki^{(i)}_{\ji{i}}} \stackrel{(a)}{=} \\
	& \frac{\mu_{\ji{i}}\sum_{\ell=1}^{i} \greedki^{(\ell)}_{\ji{i}} (x^{(g)}_{\greedk^{(\ell)}} - x^{(e)}_{\greedk^{(\ell)}}(t)) 
		- \sum_{\ell=1}^{i-1} \greedki^{(\ell)}_{\ji{i}} h^{\subind, (\ell)}(t)}
	{\greedki^{(i)}_{\ji{i}}} < \\
%	&\mumax \alpha_i + \mumax \sum_{\ell=1}^{i-1} \Kmax \alpha_\ell
%	+ \Kmax \sum_{\ell=1}^{i-1}
%	\Kmax \mumax 
%	\sum_{\ell^\prime=1}^\ell \frac{\Kmax^{\ell-\ell^\prime+1}-1}{\Kmax-1} 
%	\alpha_{\ell^\prime} = \\
	&\mumax \alpha_i + \mumax \sum_{\ell=1}^{i-1} \Kmax \alpha_\ell
	+ \sum_{\ell=1}^{i-1} \Kmax \mumax \alpha_\ell + \\
	&\sum_{\ell=1}^{i-1} \Kmax \mumax 
	\sum_{\ell^\prime=1}^{\ell-1} 2(1+\Kmax)^{\ell-\ell^\prime}
	\alpha_{\ell^\prime} \le 
%	& \mumax \alpha_i + \frac{\Kmax \mumax}{\Kmax-1} \sum_{\ell=1}^{i-1} 
%	\left(\Kmax - 1
%	+ \sum_{\ell^\prime=\ell}^{i-1} (\Kmax^{i-\ell^\prime} - 1) \right) \alpha_\ell
%	\le \\
%	& \frac{\Kmax \mumax}{\Kmax-1} 
%	\sum_{\ell=1}^i (\Kmax^{i-\ell+1} - 1) \alpha_\ell,
	\mumax \alpha_i + 2 \mumax \sum_{\ell=1}^{i-1} (1+\Kmax)^{i-\ell} \alpha_\ell
	\end{aligned}
	\end{equation*}
	In (a), we used (\ref{eq:defcalJ}),
	according to which $\rho_{\ji{i}} = \sum_{\ell=1}^i 
	\greedki^{(\ell)}_{\ji{i}} x^{(g)}_{\greedk^{(\ell)}}$ 
	and $y_{\ji{i}}(t) = \sum_{\ell=1}^i
	\greedki^{(\ell)}_{\ji{i}} x^{(e)}_{\greedk^{(\ell)}}(t)$.
	The rest of the inequalities come from recursive application of
	(\ref{eq:hbound}) and algebraic manipulations.

%	\textbf{Proof of $h^{\subind, (i)}(t) > -\frac{\Kmax \mumax}{\Kmax-1} 
%	\sum_{\ell=1}^i (\Kmax^{i-\ell+1} - 1) \alpha_\ell$:}
	To show {$h^{\subind, (i)}(t) > -\mumax\alpha_i 
		- \mumax\sum_{\ell=1}^{i-1} 2(\Kmax + 1)^{i-\ell} \alpha_\ell$}, we use 
	$y_{\ji{i}}(t) - \rho_{\ji{i}} < \epsilon_\rho$ as follows,
	\begin{equation*}
	\begin{aligned}
	& h^{\subind, (i)}(t) = 
	\frac{\lambda_{\ji{i}} - \mu_{\ji{i}} y_{\ji{i}}(t) - \sum_{\ell=1}^{i-1} \greedki^{(\ell)}_{\ji{i}} h^{\subind, (\ell)}(t)}
	{\greedki^{(i)}_\ji{i}} = \\
	& - \frac{\mu_{\ji{i}} \epsilon_\rho}{\greedki^{(i)}_\ji{i}}
	-\sum_{\ell=1}^{i-1} \frac{\greedki^{(\ell)}_\ji{i} h^{\subind, (\ell)}(t)}
	{\greedki^{(i)}_\ji{i}} \stackrel{(a)}{>} \\
%	&- \frac{\mu_{\ji{i}} \epsilon_\rho}{\greedki^{(i)}_\ji{i}}
%	- \Kmax \sum_{\ell=1}^{i-1}
%	\Kmax \mumax 
%	\sum_{\ell^\prime=1}^\ell \frac{\Kmax^{\ell-\ell^\prime+1}-1}{\Kmax-1} 
%	\alpha_{\ell^\prime} \stackrel{(a)}{>} \\
%	& -\mumax \alpha_i - \frac{\Kmax \mumax}{\Kmax-1} \sum_{\ell=1}^{i-1} 
%	\left(
%	\sum_{\ell^\prime=\ell}^{i-1} (\Kmax^{i-\ell^\prime} - 1)
%	\right) \alpha_\ell
%	\ge \\
%	& - \frac{\Kmax \mumax}{\Kmax-1} 
%	\sum_{\ell=1}^i (\Kmax^{i-\ell+1} - 1) \alpha_\ell
	& -\mumax \alpha_i - \sum_{\ell=1}^{i-1} \Kmax \mumax \alpha_\ell - 
	\sum_{\ell=1}^{i-1} \Kmax \mumax 
	\sum_{\ell^\prime=1}^{\ell-1} 2(1+\Kmax)^{\ell-\ell^\prime}
	\alpha_{\ell^\prime} > \\
	&-\mumax \alpha_i - 2 \mumax \sum_{\ell=1}^{i-1} (1+\Kmax)^{i-\ell} \alpha_\ell.
	\end{aligned}
	\end{equation*}
	In particular for (a) we used that 
	$\epsilon_\rho < \alpha_1 < \alpha_i$,
	which is due to (\ref{eq:epsrho_cond}) and (\ref{eq:alphadef}).
	We also made recursive use of (\ref{eq:hbound}).
\end{subproof}

\begin{lemma}\label{lem:xg_xebound}
	
	If $\statez(t) \in \setConv[\epsilon_\rho]$, then
	for any $i \in \{1, \ldots, \gtcntp{\alpha}(t)+1\}$:
	\begin{equation}\label{eq:xg_xe}
		\globlt{x}^{(g)}_{\greedk^{(i)}} - x^{(e)}_{\greedk^{(i)}}(t) >
		- \epsilon_\rho - \Kmax \sum_{\ell=1}^{i-1} \alpha_\ell.
	\end{equation}
\end{lemma}

\begin{proof}
	Since $i-1 \le \gtcntp{\alpha}(t) < \gtcnt$ then 
	for $j = \perm{i-1}$,
	\begin{equation}\label{eq:rhoj}
		\rho_j = \sum_{\ell=1}^{i-1} \greedk^{(\ell)}_j x^{(g)}_{\greedk^{(\ell)}}.
	\end{equation}
	
	Considering that $j$ next, we can prove (\ref{eq:xg_xe}) as follows.
	\begin{equation*}
	\begin{aligned}
		& \globlt{x}^{(g)}_{\greedk^{(i)}} - x^{(e)}_{\greedk^{(i)}}(t) 
		\stackrel{(a)}{\ge}
		\frac{\rho_j - y_j(t)
			- \sum_{\ell=1}^{i-1} \greedki^{(\ell)}_j 
			\left(\globlt{x}^{(g)}_{\greedk^{(\ell)}} - 
				x^{(e)}_{\greedk^{(\ell)}}(t)\right)}
		{\greedki^{(i)}_j} \stackrel{(b)}{\ge} \\
		&\frac{- \epsilon_\rho - \sum_{\ell=1}^{i-1} \greedki^{(\ell)}_j \alpha_\ell}{\greedki^{(i)}_j} \ge
		- \epsilon_\rho - \Kmax \sum_{\ell=1}^{i-1} \alpha_\ell.
	\end{aligned}
	\end{equation*}
	In (a) we used (\ref{eq:rhoj}) and the fact that
	$\sum_{\ell=1}^{i-1} \greedki^{(\ell)}_j 
	x^{(e)}_{\greedk^{(\ell)}}(t) \le y_j(t)$.
	In (b) we used  $\rho_j - y_j(t) \ge -\epsilon_\rho$
	as implied by (\ref{eq:epsrho}), and 
	$- x^{(g)}_{\greedk^{(\ell)}} + x^{(e)}_{\greedk^{(\ell)}}(t) >
	-\alpha_\ell$ as implied by (\ref{eq:alpha_prop}) for 
	$\ell \le i-1 \le \gtcntp{\alpha}(t)$.
\end{proof}

\textsc{\textbf{Main Proof of Proposition~\ref{prop:properties}:}}

	As a reminder in what follows we will use the notations
	$\subind[m] = \{\ji{i}: i=1, \ldots, G_m\}$ and $\ell_m := G_m + 1$, 
	given in description of Proposition. If not clear from context, 
	we will make the association of $\ji{i}$ with $m$ explicit using also the 
	notation $\jiB{i}{m}$.
	Notice that for any $m \in \{0, \ldots, \maxm{n}-1\}$, 
	$\ell_m \le \gtcntp{\alpha}(t)$.
	
	\textbf{Proof of Property~\ref{pr:exact}:}
	This property follows from two claims.
	\begin{claim}\label{cl:P1A}
		Consider $m \in \{0, \ldots, \maxm{n}-1\}$.
		If for every $i \in \{1, \ldots, \ell_m-1\}$, 
		$q^L_{\greedk^{(i)}, \jiB{i}{m}}(\tau^{(m)}_n) = o(Lf(L))$
		with probability $1 - o(L^{-2})$, 
		then (\ref{eq:exact}) holds with probability $1 - o(L^{-2})$.
	\end{claim}

	\begin{claim}\label{cl:P1B}
		For every $m \in \{0, \ldots, \maxm{n}-1\}$ and for every 
		$i \in \{1, \ldots, \ell_m-1\}$, we have that
		$q^L_{\greedk^{(i)}, \jiB{i}{m}}(\tau^{(m)}_n) = o(Lf(L))$
		with probability $1 - o(L^{-2})$.
	\end{claim}

	\textbf{Proof of Claim~\ref{cl:P1A}:}

	\textbf{Base Case $i=1$:}
		\input{anymu_P1base}

	%%%%%%%%%%%%%%%%%%%%%%%%%%%%%%%%%%%%%%%%%
	\textbf{Inductive Case $i > 1$:}
		\input{anymu_P1ind}

	%%%%%%%%%%%%%%%%%%%%%%%%%%%%%%%%%%%%%%%%%	
	\textbf{Proof of Claim~\ref{cl:P1B}:}
		\input{anymu_P1Claim2}

	%%%%%%%%%%%%%%%%%%%%%%%%%%%%%%%%%%%%%%%%%
	\textbf{Proof of Property~\ref{pr:bound1}:}
		\input{anymu_P2}

	%%%%%%%%%%%%%%%%%%%%%%%%%%%%%%%%%%%%%%%%%
	\textbf{Proof of Property~\ref{pr:bound2}:}
		\input{anymu_P3}

%% file: anymu_P1base.tex
	Let $\ji{1} = \jiB{1}{m}$.
	For (\ref{eq:exact}) to be true,
	it suffices to prove, that for any time 
	$\tau \in (\tau^{(m)}_{n}, \tau^{(m+1)}_{n})$,
	\begin{equation}\label{eq:res1}
	\greedki^{(1)}_\ji{1} \nabla {x}^{L(e)}_{\greedk^{(1)}}
	[\tau^{(m)}_{n}, \tau] = 
	\lambda_\ji{1} - \mu_\ji{1} y_\ji{1}(t) 
	+ \frac{o(\subL)}{\tau - \tau^{(m)}_{n}},
	\end{equation}
	with probability $1 - o(L^{-2})$
	or equivalently, for any $\epsilon > 0$,
	\begin{equation}\label{eq:res1eq}
	\begin{aligned}
	& \mathds{P} \left(
	\left\vert
	\greedki^{(1)}_\ji{1} 
	\left(x^{L(e)}_{\greedk^{(1)}}(\tau) 
	- x^{L(e)}_{\greedk^{(1)}}(\tau^{(m)}_{n})\right) 
	\right. \right. \\
	& \left. \left. - (\lambda_\ji{1} - \mu_\ji{1} y_\ji{1}(t))(\tau - \tau^{(m)}_{n})
	\right\vert \le \epsilon \subL \right) > 1 - o(L^{-2}).
	\end{aligned}
	\end{equation}
	
	For a given $\tau$, consider $\tau^\prime$ to be the latest time
	in $[\tau^{(m)}_{n}, \tau]$ such that
	\begin{equation}\label{eq:taupdef}
		\max_{j \in \mathcal{J}: \greedki^{(1)}_j > 0}
			q^L_{\greedk^{(1)}, j}(\tau^\prime) \ge 0.
	\end{equation}
	This time always exists since (\ref{eq:taupdef}) holds
	for $\tau^\prime = \tau^{(m)}_{n}$.
	To prove (\ref{eq:res1eq}) then, it is sufficient to prove
	\begin{equation}\label{eq:res1eqA}
	\begin{aligned}
	&\mathds{P} \left(
	\vert
	\greedki^{(1)}_\ji{1} (x^{L(e)}_{\greedk^{(1)}}(\tau^\prime) 
	- x^{L(e)}_{\greedk^{(1)}}(\tau^{(m)}_{n})) \right. \\
	& \left. - (\lambda_\ji{1} - \mu_\ji{1} y_\ji{1}(t))(\tau^\prime - \tau^{(m)}_{n})
	\vert \le \frac{\epsilon \subL}{2}\right) = 1 - o(L^{-2}),
	\end{aligned}
	\end{equation}
	and
	\begin{equation}\label{eq:res1eqB}
	\begin{aligned}
	&\mathds{P} \left(
	\vert
	\greedki^{(1)}_\ji{1} (x^{L(e)}_{\greedk^{(1)}}(\tau) 
	- x^{L(e)}_{\greedk^{(1)}}(\tau^\prime)) \right.\\
	& \left. - (\lambda_\ji{1} - \mu_\ji{1} y_\ji{1}(t))(\tau - \tau^\prime)
	\vert \le \frac{\epsilon \subL}{2}\right) = 1 - o(L^{-2}).
	\end{aligned}
	\end{equation}
%	\begin{equation}\label{eq:res1eqA}
%	\begin{aligned}
%	&\lim_{L \to \infty}
%	\mathds{P} \left(
%	\vert
%	\greedki^{(1)}_\ji{1} (x^{L(e)}_{\greedk^{(1)}}(\tau^\prime) 
%	- x^{L(e)}_{\greedk^{(1)}}(\tau^{(m)}_{n})) \right. \\
%	& \left. - (\lambda_\ji{1} - \mu_\ji{1} y_\ji{1}(t))(\tau^\prime - \tau^{(m)}_{n})
%	\vert \le \frac{\epsilon \subL}{2}\right) = 1,
%	\end{aligned}
%	\end{equation}
%	and
%	\begin{equation}\label{eq:res1eqB}
%	\begin{aligned}
%	&\lim_{L \to \infty}
%	\mathds{P} \left(
%	\vert
%	\greedki^{(1)}_\ji{1} (x^{L(e)}_{\greedk^{(1)}}(\tau) 
%	- x^{L(e)}_{\greedk^{(1)}}(\tau^\prime)) \right.\\
%	& \left. - (\lambda_\ji{1} - \mu_\ji{1} y_\ji{1}(t))(\tau - \tau^\prime)
%	\vert \le \frac{\epsilon \subL}{2}\right) = 1.
%	\end{aligned}
%	\end{equation}
	
%	We will prove (\ref{eq:res1}) by contradiction,
%	so we start by assuming 
%	\begin{equation}\label{eq:nabla_contr_base}
%	\greedki^{(1)}_j \nabla {x}^{(e)}_{\greedk^{(1)}}
%	[\tau^{(m)}_{n}, \tau] = 
%	\lambda_j - \mu_j y_j(t) + \frac{f^\prime(L)/L}{\tau - \tau^{(m)}_{n}},
%	\end{equation}
%	where
%	\begin{equation}\label{eq:Pfpf_base}
%	\lim_{L \to \infty} \mathds{P}\left(
%	\left \vert \frac{f^\prime(L)}{f(L)} \right \vert
%	\ge \epsilon
%	\right) > 0,
%	\end{equation}
%	for some $\epsilon > 0$.
	\textbf{Proof of (\ref{eq:res1eqA}):}
	We will now prove (\ref{eq:res1eqA}) by considering two cases
	depending on length of $\tau^\prime - \tau^{(m)}_{n}$.

	We consider
	\begin{equation}\label{eq:taucond1}
	\tau^\prime - \tau^{(m)}_{n} \le 
	\frac{\epsilon \subL}{4\Kmax R}
	\end{equation}
	or
	\begin{equation}\label{eq:taucond2}
	\tau^\prime - \tau^{(m)}_{n} > 
	\frac{\epsilon \subL}{4\Kmax R},
	\end{equation}
	where $R := \sum_{j \in \mathcal{J}} \lambda_j + \Kmax\mumax$
	
	\textbf{Case (\ref{eq:taucond1}):}
		We notice ${x}^{L(e)}_{\greedk^{(1)}}(\tau)$ 
		will change by at most $1/L$ at each arrival
		or departure according to Lemma~\ref{lem:X_change}
		and thus $\greedki^{(1)}_\ji{1} {x}^{L(e)}_{\greedk^{(1)}}(\tau)$
		will change by at most $\Kmax/L$.
		
%		By the law of large numbers the number of arrivals and departures 
%		in $[\tau^{(m)}_{n}, \tau^\prime]$ is almost surely at most
%		\begin{equation}
%		\begin{aligned}
%		&\Big(\sum_{j \in \mathcal{J}} \lambda_j 
%			+ \Kmax\mumax \Big) L
%		\frac{\epsilon \subL}{5\Kmax R} + o(L\subL) = \\
%		&\frac{\epsilon L \subL}{5\Kmax} + o(L\subL),
%		\end{aligned}
%		\end{equation} 
	The number of arrivals and departures 
		in $[\tau^{(m)}_{n}, \tau^\prime]$ is stochastically bounded by a Poisson Process of rate $\Big(\sum_{j \in \mathcal{J}} \lambda_j 
		+ \Kmax\mumax \Big) L$ 	on an interval of length at most $\frac{\epsilon \subL}{4\Kmax R}$, which, according to Lemma~\ref{lem:Chern_Poisson}, 
		with probability $1 - o(L^{-2})$ is at most
		\begin{equation*}
		\begin{aligned}
		&\Big(\sum_{j \in \mathcal{J}} \lambda_j 
		+ \Kmax\mumax \Big) L
		\frac{\epsilon \subL}{4\Kmax R} + o(L\subL) = \\
		&\frac{\epsilon L \subL}{4\Kmax} + o(L\subL),
		\end{aligned}
		\end{equation*} 
%		\begin{equation}\label{eq:contr1a}
%		\lim_{L \to \infty}
%		\mathds{P}\left(
%		\greedki^{(1)}_\ji{1}
%		\left( {x}^{L(e)}_{\greedk^{(1)}}(\tau^\prime) 
%		- {x}^{L(e)}_{\greedk^{(1)}}(\tau^{(m)}_{n}) \right) \le
%		\frac{\epsilon \subL}{4}\right) = 1.
%		\end{equation}
		therefore,
		\begin{equation}\label{eq:contr1a}
		\mathds{P}\left(
		\greedki^{(1)}_\ji{1}
		\left( {x}^{L(e)}_{\greedk^{(1)}}(\tau^\prime) 
		- {x}^{L(e)}_{\greedk^{(1)}}(\tau^{(m)}_{n}) \right) \le
		\frac{\epsilon \subL}{4}\right) \ge 1 - o(L^{-2}).
		\end{equation}
		
		Considering (\ref{eq:taucond1}) holds, we also have
%		\begin{equation}\label{eq:contr1b}
%		\begin{aligned}
%		&\lim_{L \to \infty}
%		\mathds{P}\left(
%		\left(\lambda_\ji{1} - \mu_\ji{1} y_\ji{1}(t)\right)
%		\left(\tau^\prime - \tau^{(m)}_{n}\right) \le
%		\frac{\epsilon \subL}{4}\right) \ge \\
%		&\mathds{P}\left(
%		\frac{\epsilon \subL}{5} \le
%		\frac{\epsilon \subL}{4}\right) = 1.
%		\end{aligned}
%		\end{equation}
		\begin{equation}\label{eq:contr1b}
		%\begin{aligned}
		\mathds{P}\left(
		\left(\lambda_\ji{1} - \mu_\ji{1} y_\ji{1}(t)\right)
		\left(\tau^\prime - \tau^{(m)}_{n}\right) \le
		\frac{\epsilon \subL}{4}\right) =1
		%\ge \\
%		&\mathds{P}\left(
%		\frac{\epsilon \subL}{4} \le
%		\frac{\epsilon \subL}{4}\right) = 1.
%		\end{aligned}
		\end{equation}
		It is now easy to verify that equations (\ref{eq:contr1a}) 
		and (\ref{eq:contr1b}) imply (\ref{eq:res1eqA}).
%		holds with probability $1 - o(L^{-2})$.
	
	\textbf{Case (\ref{eq:taucond2}):}
		In this case we notice using Lemma~\ref{lem:Chern_Poisson} for
		the process of jobs of type $\ji{1}$ in the system 
		which is Poisson with rate $L(\lambda_\ji{1} - \mu_\ji{1} y_\ji{1}(t))$,
		that with probability $1 - o(L^{-2})$
		\ben
		o(L\subL) &=& q^L_{\greedk^{(1)},\ji{1}}(\tau^{(m)}_{n}) \\
		&= & q^L_{\greedk^{(1)},\ji{1}}(\tau^\prime) - \greedki^{(1)}_\ji{1} L 
		(x^{L(e)}_{\greedk^{(1)}}(\tau^\prime) 
		- x^{L(e)}_{\greedk^{(1)}}(\tau^{(m)}_{n}))\\
		&& + L y^L_\ji{1}(\tau^\prime) - L y^L_\ji{1}(\tau^{(m)}_{n}) \\ 
		&\le& \Kmax - \greedki^{(1)}_\ji{1} L 
		(x^{L(e)}_{\greedk^{(1)}}(\tau^\prime) 
		- x^{L(e)}_{\greedk^{(1)}}(\tau^{(m)}_{n})) \\
		&&+ L(\lambda_\ji{1} - \mu_\ji{1} y_\ji{1}(t))(\tau^\prime - \tau^{(m)}_{n})
		+ o(L\subL),
		\een
		or equivalently, since trivially $\Kmax = o(L\subL)$,
		\begin{equation}\label{eq:final1a}
		\begin{aligned}
		& \frac{x^{L(e)}_{\greedk^{(1)}}(\tau^\prime) 
		- x^{L(e)}_{\greedk^{(1)}}(\tau^{(m)}_{n})}
		{\tau^\prime - \tau^{(m)}_{n}} \le \\
		&\frac{1}{\greedki^{(1)}_\ji{1}}(\lambda_\ji{1} - \mu_\ji{1} y_\ji{1}(t))
		+ \frac{o(\subL)}{\tau^\prime - \tau^{(m)}_{n}} \stackrel{(a)}{=} \\
		&\frac{1}{\greedki^{(1)}_\ji{1}}(\lambda_\ji{1} - \mu_\ji{1} y_\ji{1}(t))
		+ o(1),
		\end{aligned}
		\end{equation}
		where in (a) we just used (\ref{eq:taucond2}). 
		Let
		\begin{equation}
		j^\prime := \argmax_{j \in \mathcal{J}: \greedki^{(1)}_j > 0}
		q^L_{\greedk^{(1)}, j}(\tau^\prime).
		\end{equation}
		Then we also have, using Lemma~\ref{lem:Chern_Poisson} for
		the process of jobs of type $j^\prime$ in the system 
		which is Poisson with rate 
		$L(\lambda_{j^\prime} - \mu_{j^\prime} y_{j^\prime}(t))$,
		that with probability $1 - o(L^{-2})$
		\ben
		0 &  \le & q^L_{\greedk^{(1)}, j^\prime}(\tau^\prime) = 
		q^L_{\greedk^{(1)}, j^\prime}(\tau^{(m)}_{n}) \\
		&&+ \greedki^{(1)}_{j^\prime} L 
		(x^{L(e)}_{\greedk^{(1)}}(\tau^\prime) 
		- x^{L(e)}_{\greedk^{(1)}}(\tau^{(m)}_{n}))
		- L(y^L_{j^\prime}(\tau^\prime) - y^L_{j^\prime}(\tau^{(m)}_{n}))  \\
		& \le & \Kmax +
		\greedki^{(1)}_{j^\prime} L 
		(x^{L(e)}_{\greedk^{(1)}}(\tau^\prime) 
		- x^{L(e)}_{\greedk^{(1)}}(\tau^{(m)}_{n})) \\
		&&- L(\lambda_{j^\prime} - \mu_{j^\prime} y_{j^\prime}(t))
		(\tau^\prime - \tau^{(m)}_{n})
		+ o(L\subL),
		\een
		from which, 
		%equivalently, if we bound with respect to 
		%$\frac{x^{L(e)}_{\greedk^{(1)}}(\tau^\prime) 
		%	- x^{L(e)}_{\greedk^{(1)}}(\tau^{(m)}_{n})}
		%	{\tau^\prime - \tau^{(m)}_{n}}$
		since trivially $\Kmax = o(L\subL)$, it follows
		\begin{equation}\label{eq:final1b}
		\begin{aligned}
		& \frac{x^{L(e)}_{\greedk^{(1)}}(\tau^\prime) 
		- x^{L(e)}_{\greedk^{(1)}}(\tau^{(m)}_{n})}
		{\tau^\prime - \tau^{(m)}_{n}} \ge
		\frac{\lambda_{j^\prime} - \mu_{j^\prime} y_{j^\prime}(t)}
		{\greedki^{(1)}_{j^\prime}}
		 + o(1).
		\end{aligned}
		\end{equation}
		
%		Let
%		\begin{equation}
%		\mathcal{\bar J}^\prime = \{j^\prime\}.
%		\end{equation}
		Considering that (\ref{eq:h_cond}) does not hold
		for index $j=j^\prime$ and $G_m=1$,
		we get
		\begin{equation}\label{eq:final1ab}
		\begin{aligned}
		&\frac{\lambda_{j^\prime} - \mu_{j^\prime} 
		y_{j^\prime}(t)}{\greedki^{(1)}_{j^\prime}} \ge 
%		h^{\mathcal{\bar J}^\prime, (1)}(t) \ge \\
		h^{\subind[m], (1)}(t) =
		\frac{\lambda_\ji{1} - \mu_\ji{1} y_\ji{1}(t)}{\greedki^{(1)}_\ji{1}}.
		\end{aligned}
		\end{equation}
		From (\ref{eq:final1a}), (\ref{eq:final1b}) and (\ref{eq:final1ab})
		we get
		\begin{equation}
		\frac{x^{L(e)}_{\greedk^{(1)}}(\tau^\prime) 
		- x^{L(e)}_{\greedk^{(1)}}(\tau^{(m)}_{n})}{\tau^\prime - \tau^{(m)}_{n}}
		= \frac{\lambda_\ji{1} - \mu_\ji{1} y_\ji{1}(t)}{\greedki^{(1)}_\ji{1}}
		+ o(1),
		\end{equation}
		which holds with probability $1 - o(L^{-2})$
		and therefore it implies (\ref{eq:res1eqA}).

%		Comparing this result with (\ref{eq:final1a}) and (\ref{eq:final1b}), 
%		both must be satisfied as equalities and in particular 
%		equality in (\ref{eq:final1a}) implies (\ref{eq:res1eqA}).

	\textbf{Proof of (\ref{eq:res1eqB}):}
	We will now prove (\ref{eq:res1eqB}) by considering two cases
	depending on length of $\tau - \tau^\prime$.
	
	We consider
	\begin{equation}\label{eq:taucond3}
	\tau - \tau^\prime \le 
	\frac{\epsilon \subL}{4\Kmax R}
	\end{equation}
	or
	\begin{equation}\label{eq:taucond4}
	\tau - \tau^\prime > 
	\frac{\epsilon \subL}{4\Kmax R},
	\end{equation}
	where $R := \sum_{j \in \mathcal{J}} \lambda_j + \Kmax\mumax$.

	We further assume that
	\begin{equation}\label{eq:eps_cond}
	\epsilon < \frac{\mu_\ji{1} \alpha_1}
		{4 \Kmax R}.
%	\epsilon < \frac{\frac{\mumin}{J \Kmax \cfcnt^2}
%		\left(1 - x^{(e)}_{\greedk^{(1)}}(t) \right) 
%		- \mu_\ji{1} \alpha_1}
%	{\Kmax^2 R} 
%	\left(\frac{\mumin}{J \Kmax \cfcnt^2}
%	\left(1 - x^{(e)}_{\greedk^{(1)}}(t) \right) - \mu_\ji{1} \alpha_1\right),
	\end{equation}
%	and such an $\epsilon$ exists given right hand side is positive 
%	from (\ref{eq:alphadef}).

	\textbf{Case (\ref{eq:taucond3}):}
		Following the same arguments as in the case of (\ref{eq:taucond1})
		we can infer the equivalent of (\ref{eq:contr1a}) and 
		(\ref{eq:contr1b}) for interval $(\tau^\prime, \tau)$, i.e.
%		\begin{equation}\label{eq:contr1c}
%		\lim_{L \to \infty}
%		\mathds{P}\left(
%		\greedki^{(1)}_\ji{1}
%		\left({x}^{L(e)}_{\greedk^{(1)}}(\tau) 
%		- {x}^{L(e)}_{\greedk^{(1)}}(\tau^\prime)\right)  \le
%		\frac{\epsilon \subL}{4}\right) = 1
%		\end{equation}
%		and
%		\begin{equation}\label{eq:contr1d}
%		\lim_{L \to \infty}
%		\mathds{P}\left(
%		\left(\lambda_\ji{1} - \mu_\ji{1} y_\ji{1}(t)\right)
%		(\tau - \tau^\prime) \le
%		\frac{\epsilon \subL}{4}\right) = 1,
%		\end{equation}
		\begin{equation}\label{eq:contr1c}
		\mathds{P}\left(
		\greedki^{(1)}_\ji{1}
		\left({x}^{L(e)}_{\greedk^{(1)}}(\tau) 
		- {x}^{L(e)}_{\greedk^{(1)}}(\tau^\prime)\right)  \le
		\frac{\epsilon \subL}{4}\right) \ge 1 - o(L^{-2})
		\end{equation}
		and
		\begin{equation}\label{eq:contr1d}
		\mathds{P}\left(
		\left(\lambda_\ji{1} - \mu_\ji{1} y_\ji{1}(t)\right)
		(\tau - \tau^\prime) \le
		\frac{\epsilon \subL}{4}\right) = 1 ,
		\end{equation}
		which imply (\ref{eq:res1eqB}).
		
	\textbf{Case (\ref{eq:taucond4}):}
		First we will prove that
		\begin{equation}\label{eq:q_taup1}
		q^L_{\greedk^{(1)}, \ji{1}}(\tau^\prime) \ge
		- \epsilon^2 L\subL + o(L\subL),
		\end{equation}
		or equivalently
		\begin{equation}
		\begin{aligned}
		& 1/L q^L_{\greedk^{(1)}, \ji{1}}(\tau^\prime) 
			- 1/L q^L_{\greedk^{(1)}, \ji{1}}(\tau^{(m)}_n)
		\ge - \epsilon^2 \subL +  o(\subL),	
		\end{aligned}
		\end{equation}
		which both will hold with probability $1 - o(L^{-2})$.
		The analysis for this is same as with the proof of (\ref{eq:res1eqA}) 
		so we highlight only the parts that are different. 
		\begin{itemize}
			\item Instead of considering cases $\tau^\prime - \tau^{(m)}_{n} \le 
			\frac{\epsilon \subL}{4\Kmax R}$ and $\tau^\prime - \tau^{(m)}_{n} > 
			\frac{\epsilon \subL}{4\Kmax R}$, we should consider
			$\tau^\prime - \tau^{(m)}_{n} \le 
			\frac{\epsilon^2 \subL}{2(\Kmax+1) R}$ and
			$\tau^\prime - \tau^{(m)}_{n} > 
			\frac{\epsilon^2 \subL}{2(\Kmax+1) R}$.
            \item If interval is short, we can bound the 
            absolute change of variable 
            $1/L q^L_{\greedk^{(1)}, \ji{1}}(\tau)$ which changes by at most $(\Kmax+1)/L$ after each arrival or departure.
			\item If interval is long, we can still prove that the equivalent
			of (\ref{eq:final1a}) is satisfied as equality and considering
			$y^L_\ji{1}(\tau^\prime) - y^L_\ji{1}(\tau^{(m)}_{n}) =
			\lambda_\ji{1} - \mu_\ji{1} y_\ji{1}(t) + o(\subL)$, 
			we can get through (\ref{eq:final1a}) that 
			$1/L q^L_{\greedk^{(1)}, \ji{1}}(\tau^\prime) = o(\subL)$
			with probability $1 - o(L^{-2})$.
		\end{itemize}

		In this case, because of (\ref{eq:q_taup1}) we get
		\begin{equation}
		\begin{aligned}
		& o(L\subL) - \epsilon^2 L\subL
		\le q^L_{\greedk^{(1)},\ji{1}}(\tau^\prime)  = 
		q^L_{\greedk^{(1)},\ji{1}}(\tau) \\
		& - \greedki^{(1)}_\ji{1} L 
		\left(x^{L(e)}_{\greedk^{(1)}}(\tau) 
		- x^{L(e)}_{\greedk^{(1)}}(\tau^\prime)
		\right)
		+ L y^L_\ji{1}(\tau) - L y^L_\ji{1}(\tau^\prime) \le \\
		& \Kmax
		- \greedki^{(1)}_\ji{1} L 
		(x^{L(e)}_{\greedk^{(1)}}(\tau) 
		- x^{L(e)}_{\greedk^{(1)}}(\tau^\prime)) \\
		& + L(\lambda_\ji{1} - \mu_\ji{1} y_\ji{1}(t))(\tau - \tau^\prime)
		+ o(L\subL),
		\end{aligned}
		\end{equation}
		from which, after considering $\Kmax = o(L\subL)$, it follows that
		\begin{equation}\label{eq:final1c}
		\begin{aligned}
		& \frac{x^{L(e)}_{\greedk^{(1)}}(\tau)  - x^{L(e)}_{\greedk^{(1)}}(\tau^\prime)}{\tau - \tau^\prime} \le
		\frac{\lambda_\ji{1} - \mu_\ji{1} y_\ji{1}(t)}{\greedki^{(1)}_\ji{1}} \\
		& + \frac{\epsilon^2 \subL}{\greedki^{(1)}_\ji{1}(\tau - \tau^\prime)} 
		+ \frac{o(\subL)}{\tau - \tau^\prime} \stackrel{(a)}{\le} \\
		&\mu_\ji{1} (\globlt{x}^{(g)}_{\greedk^{(1)}} - x^{(e)}_{\greedk^{(1)}}(t)) +
		\frac{4 \Kmax \epsilon R}{\greedki^{(1)}_\ji{1}} + o(1) \stackrel{(b)}{<}\\
		&\mu_\ji{1} \alpha_1 + \mu_\ji{1} \alpha_1 + o(1) <
		2 \mumax \alpha_1 + o(1).
		\end{aligned}
		\end{equation}
		In (a) we applied (\ref{eq:taucond4}) and properties of $\ji{1}$
		which come from (\ref{eq:yjA}) and (\ref{eq:defcalJ}) for $i=1$,
		while in (b) we used (\ref{eq:alpha_prop}) and (\ref{eq:eps_cond}).

		Also from Lemma~\ref{lem:emptyrate} and given that all the servers
		that empty during $[\tau^\prime, \tau)$ will be assigned to
		configuration $\greedk^{(1)}$, we will have
		that with probability $1 - o(L^{-2})$,
		\begin{equation}\label{eq:final1d_}
		\begin{aligned}
		& x^{L(e)}_{\greedk^{(1)}}(\tau) 
		- x^{L(e)}_{\greedk^{(1)}}(\tau^\prime)
		\ge 
		\frac{\mumin}{J \Kmax \cfcnt^2}
		\left(1 - x^{(e)}_{\greedk^{(1)}}(t) \right)
		(\tau - \tau^\prime) + o(\subL)
		\end{aligned}
		\end{equation}
		or equivalently,
		\begin{equation}\label{eq:final1d}
		\begin{aligned}
		& \frac{x^{L(e)}_{\greedk^{(1)}}(\tau) 
		- x^{L(e)}_{\greedk^{(1)}}(\tau^\prime)}{\tau - \tau^\prime}
		\ge 
		\frac{\mumin}{J \Kmax \cfcnt^2}
		\left(1 - x^{(e)}_{\greedk^{(1)}}(t) \right)
		 + o(1) \stackrel{(a)}{>} \\
		&\frac{\mumin \globlt{x}_{\greedk^{(\gtcnt)}}}{2J \Kmax \cfcnt^2}
		+ o(1),
		\end{aligned}
		\end{equation}
		where (a) is due to Lemma~\ref{lem:rank0frac}.
		So far we proved that if (\ref{eq:taucond4}) holds, 
		(\ref{eq:final1c}) and (\ref{eq:final1d}) also hold.
		
		However, considering (\ref{eq:alphadef}) we have
		\begin{equation}\label{eq:contr1}
		\begin{aligned}
		& 2\mumax \alpha_1 < 
		\frac{\mumin \globlt{x}^{(g)}_{\greedk^{(\gtcnt)}}}{2J \Kmax \cfcnt^2}
		\end{aligned}
		\end{equation}
		and because of that, the probability 
		that (\ref{eq:final1c}) and (\ref{eq:final1d}) are both true 
		is $o(L^{-2})$.
		This means that (\ref{eq:taucond3}) holds
		with probability at least $1 - o(L^{-2})$,
		and thus the analysis of (\ref{eq:taucond3}) is
		sufficient for (\ref{eq:res1eqB}) to hold.

%% file: anymu_P1ind.tex
	For $i \in \{1, \ldots, \ell_m\}$ we have according to
	the assumptions of Claim~\ref{cl:P1B}  that
	$q^L_{\greedk^{(i)}, \jiB{i}{m}}(\tau^{(m)}_{n}) = o(L\subL)$
	with probability $1 - o(L^{-2})$.

	What we need to prove is that under this assumption,
	for any time $\tau \in (\tau^{(m)}_{n}, \tau^{(m+1)}_{n})$,
	$i \in \{1, \ldots, \ell_m\}$ and
	$\ji{i} := \jiB{i}{m}$,
	\begin{equation}\label{eq:res2}
	\sum_{\ell=1}^i \greedki^{(\ell)}_\ji{i} \nabla {x}^{L(e)}_{\greedk^{(\ell)}}
	[\tau^{(m)}_{n}, \tau] =
	\lambda_\ji{i} - \mu_\ji{i} y_\ji{i}(t) + \frac{o(\subL)}{\tau - \tau^{(m)}_{n}}.
	\end{equation}
	with probability $1 - o(L^{-2})$, so it is sufficient to only consider
	the case
	$q^L_{\greedk^{(i)}, \jiB{i}{m}}(\tau^{(m)}_{n}) = o(L\subL)$.
	Equivalently, it suffices to show that for any $\epsilon > 0$,
	\begin{equation}\label{eq:res2eq}
	\begin{aligned}
	&\mathds{P} (
	\vert
	\sum_{\ell=1}^i \greedki^{(\ell)}_\ji{i}
	(x^{L(e)}_{\greedk^{(\ell)}}(\tau)
	- x^{L(e)}_{\greedk^{(\ell)}}(\tau^{(m)}_{n})) \\
	& - (\lambda_\ji{i} - \mu_\ji{i} y_\ji{i}(t))(\tau - \tau^{(m)}_{n})
	\vert \le \epsilon \subL) \ge 1 - o(L^{-2}).
	\end{aligned}
	\end{equation}

	For a given $\tau$, let $\tau^\prime$ to be the latest time
	in $(\tau^{(m)}_{n}, \tau)$ such that
	\begin{equation}
	\max_{j \in \mathcal{J}: \greedki^{(i)}_j > 0}
	q^L_{\greedk^{(i)}, j} \ge 0.
	\end{equation}
	Using the same argument as in the base case,
	this time always exists.
	To prove (\ref{eq:res2eq}) then,
	it is sufficient to prove
	\begin{equation}\label{eq:res1eqC}
	\begin{aligned}
	&\mathds{P} \left(
	\left\vert
	\sum_{\ell=1}^i \greedki^{(\ell)}_\ji{i}
	(x^{L(e)}_{\greedk^{(\ell)}}(\tau^\prime)
	- x^{L(e)}_{\greedk^{(\ell)}}(\tau^{(m)}_{n}))
	\right. \right.\\
	&\left.\left. - (\lambda_\ji{i} - \mu_\ji{i} y_\ji{i}(t))
	(\tau^\prime - \tau^{(m)}_{n})
	\right\vert \le \frac{\epsilon \subL}{2}\right) \ge 1 - o(L^{-2}),
	\end{aligned}
	\end{equation}
	and
	\begin{equation}\label{eq:res1eqD}
	\begin{aligned}
	&\mathds{P} (
	\vert
	\sum_{\ell=1}^i \greedki^{(\ell)}_\ji{i}
	(x^{L(e)}_{\greedk^{(\ell)}}(\tau)
	- x^{L(e)}_{\greedk^{(\ell)}}(\tau^\prime)) \\
	& - (\lambda_\ji{i} - \mu_\ji{i} y_\ji{i}(t))(\tau - \tau^\prime)
	\vert \le \frac{\epsilon \subL}{2}) \ge 1 - o(L^{-2}).
	\end{aligned}
	\end{equation}

	\textbf{Proof of (\ref{eq:res1eqC}):}
	We will now prove (\ref{eq:res1eqC}) by considering two cases
	depending on length of	$\tau^\prime - \tau^{(m)}_{n}$, i.e., we have either
	\begin{equation}\label{eq:taucond5}
	\tau^\prime - \tau^{(m)}_{n} \le
	\frac{\epsilon \subL}{4(1+\Kmax)^i R}
	\end{equation}
	or
	\begin{equation}\label{eq:taucond6}
	\tau^\prime - \tau^{(m)}_{n} >
	\frac{\epsilon \subL}{4(1+\Kmax)^i R},
	\end{equation}
	where $R:= \sum_{j \in \mathcal{J}} \lambda_j + \Kmax\mumax$

	\textbf{Case (\ref{eq:taucond5}):}
		We notice ${x}^{L(e)}_{\greedk^{(\ell)}}(\tau)$
		will change by at most $(1+\Kmax)^{\ell-1}/L$ at each arrival
		or departure according to Lemma~\ref{lem:X_change}
		and thus $\sum_{\ell=1}^i \greedki^{(\ell)}_\ji{i}
		{x}^{L(e)}_{\greedk^{(\ell)}}(\tau)$
		will change by at most $\frac{(1+\Kmax)^i-1}{L}$.

		For the number of arrivals and departures
		in $[\tau^{(m)}_{n}, \tau^\prime]$
		which are Poisson processes with means at most
		$\Big(\sum_{j \in \mathcal{J}} \lambda_j
		+ \Kmax\mumax \Big) L$
		on an interval of length at most $\frac{\epsilon \subL}{4(1+\Kmax)^i R}$
		we have, according to Lemma~\ref{lem:Chern_Poisson}, that
		with probability $1 - o(L^{-2})$ they are at most
		\begin{equation*}
		\left(\sum_{j \in \mathcal{J}} \lambda_j
		+ \Kmax\mumax \right) L
		\frac{\epsilon \subL}{4(1+\Kmax)^i R} + o(L\subL) =
		\frac{\epsilon L\subL}{4(1+\Kmax)^i} + o(L\subL),
		\end{equation*}
		therefore,
		\begin{equation}\label{eq:contr1e}
		\mathds{P}\left(
		\sum_{\ell=1}^i \greedki^{(\ell)}_\ji{i}
		({x}^{L(e)}_{\greedk^{(\ell)}}(\tau^\prime)
		- {x}^{L(e)}_{\greedk^{(\ell)}}(\tau^{(m)}_{n}) ) \le
		\frac{\epsilon \subL}{4}\right) = 1 - o(L^{-2}).
		\end{equation}

		Considering (\ref{eq:taucond5}) holds, in this case we clearly have
		\begin{equation}\label{eq:contr1f}
		\begin{aligned}
		&\mathds{P}\left(
		\left(\lambda_\ji{i} - \mu_\ji{i} y_\ji{i}(t)\right)
		\left(\tau^\prime - \tau^{(m)}_{n}\right) \le
		\frac{\epsilon \subL}{4}\right)=1
		\end{aligned}
		\end{equation}
		It is now easy to verify that equations (\ref{eq:contr1e})
		and (\ref{eq:contr1f}) imply (\ref{eq:res1eqC}).

	\textbf{Case (\ref{eq:taucond6}):}
		In this case we notice,
		\begin{equation*}
		\begin{aligned}
		o(L\subL) &= q^L_{\greedk^{(i)},\ji{i}}(\tau^{(m)}_{n}) =
		q^L_{\greedk^{(i)},\ji{i}}(\tau^\prime) \\
		&- \sum_{\ell=1}^i \greedki^{(\ell)}_\ji{i} L
		(x^{L(e)}_{\greedk^{(\ell)}}(\tau^\prime)
		- x^{L(e)}_{\greedk^{(\ell)}}(\tau^{(m)}_{n}))
		+ L y^L_\ji{i}(\tau^\prime) - L y^L_\ji{i}(\tau^{(m)}_{n})  \\
		&\le \Kmax
		- \sum_{\ell=1}^i \greedki^{(\ell)}_\ji{i} L
		(x^{L(e)}_{\greedk^{(\ell)}}(\tau^\prime)
		- x^{L(e)}_{\greedk^{(\ell)}}(\tau^{(m)}_{n})) \\
		& + L (\lambda_\ji{i} - \mu_\ji{i} y_\ji{i}(t))
		(\tau^\prime - \tau^{(m)}_{n})
		+ o(L\subL),
		\end{aligned}
		\end{equation*}
		which holds with probability $1 - o(L^{-2})$ by applying Lemma~\ref{lem:Chern_Poisson} in the last step. 
        It therefore follows that, with the same probability,
		\begin{equation}\label{eq:final1e}
		\begin{aligned}
		& \frac{x^{L(e)}_{\greedk^{(i)}}(\tau^\prime)  -
		x^{L(e)}_{\greedk^{(i)}}(\tau^{(m)}_{n})}{\tau^\prime - \tau^{(m)}_{n}}
		\le \\
		&\frac{\lambda_\ji{i} - \mu_\ji{i} y_\ji{i}(t)}
		{\greedki^{(i)}_\ji{i}}
		-\sum_{\ell=1}^{i-1}\frac{\greedki^{(\ell)}_\ji{i}
			(x^{L(e)}_{\greedk^{(\ell)}}(\tau^\prime)  -
			x^{L(e)}_{\greedk^{(\ell)}}(\tau^{(m)}_{n}))}
		{\greedki^{(i)}_\ji{i}(\tau^\prime - \tau^{(m)}_{n})}
		+ \frac{o(\subL)}{\tau^\prime - \tau^{(m)}_{n}} \stackrel{(a)}{=} \\
		&\frac{
			\lambda_\ji{i} - \mu_\ji{i} y_\ji{i}(t)
			- \sum_{\ell=1}^{i-1} \greedki^{(\ell)}_\ji{i}
			h^{\subind[m](\ell)}(t)}
		{\greedki^{(i)}_\ji{i}} + o(1),
		\end{aligned}
		\end{equation}
		where in (a) we applied (\ref{eq:xe_h}) of Lemma~\ref{lem:hrec} for
		indexes $1, \ldots, i-1$, and used (\ref{eq:taucond6}).

		Next, let
		\begin{equation*}
		j^\prime := \argmax_{j \in \mathcal{J}: \greedki^{(i)}_j > 0}
		q^L_{\greedk^{(i)}, j}(\tau^\prime).
		\end{equation*}
		Then again we have that,
		with probability $1 - o(L^{-2})$,
		\begin{equation*}
		\begin{aligned}
		0 &\le q^L_{\greedk^{(i)}, j^\prime}(\tau^\prime) =
		q^L_{\greedk^{(i)}, j^\prime}(\tau^{(m)}_{n}) + \\
		&\sum_{\ell=1}^i \greedki^{(\ell)}_{j^\prime} L
		(x^{L(e)}(\tau^\prime)_{\greedk^{(\ell)}}
		- x^{L(e)}_{\greedk^{(\ell)}}(\tau^{(m)}_{n}))
		- L(y^L_{j^\prime}(\tau^\prime) - y^L_{j^\prime}(\tau^{(m)}_{n}))  \\
		&\le \Kmax +
		\sum_{\ell=1}^i \greedki^{(\ell)}_{j^\prime} L
		(x^{L(e)}_{\greedk^{(\ell)}}(\tau^\prime)
		- x^{L(e)}_{\greedk^{(\ell)}}(\tau^{(m)}_{n})) \\
		&- (\lambda_{j^\prime} - \mu_{j^\prime} y_{j^\prime}(t))
		(\tau^\prime - \tau^{(m)}_{n})
		+ o(L\subL),
		\end{aligned}
		\end{equation*}
		from which it follows that
		\begin{equation}\label{eq:final1f}
		\begin{aligned}
		& \frac{x^{L(e)}_{\greedk^{(i)}}(\tau^\prime)
		- x^{L(e)}_{\greedk^{(i)}}(\tau^{(m)}_{n})}
		{\tau^\prime - \tau^{(m)}_{n}}
		\ge \\
		&
		\frac{\lambda_{j^\prime} - \mu_{j^\prime} y_{j^\prime}(t)}
		{\greedki^{(i)}_{j^\prime}}
		-\sum_{\ell=1}^{i-1}\frac{\greedki^{(\ell)}_{j^\prime}
			(x^{L(e)}_{\greedk^{(\ell)}}(\tau^\prime) -
			x^{L(e)}_{\greedk^{(\ell)}}(\tau^{(m)}_{n}))}
		{\greedki^{(i)}_{j^\prime}(\tau^\prime - \tau^{(m)}_{n})}
		+ \frac{o(\subL)}{\tau^\prime - \tau^{(m)}_{n}} \stackrel{(a)}{=} \\
		&\frac{\lambda_{j^\prime} - \mu_{j^\prime} y_{j^\prime}(t) -
		\sum_{\ell=1}^{i-1} \greedki^{(\ell)}_{j^\prime} h^{\subind[m](\ell)}(t) }
		{\greedki^{(i)}_{j^\prime}} + o(1).
		\end{aligned}
		\end{equation}
		where in (a) we applied (\ref{eq:xe_h}) of Lemma~\ref{lem:hrec} for
		indexes $1, \ldots, i-1$, and used (\ref{eq:taucond6}).

		Considering (\ref{eq:h_system}) holds, that (\ref{eq:h_cond}) 
        does not hold for $j=j^\prime$ and $G_m=i$, we get
		\begin{equation}\label{eq:final1ef}
		\begin{aligned}
		&\frac{\lambda_{j^\prime} - \mu_{j^\prime} y_{j^\prime}(t)
		- \sum_{\ell=1}^{i-1} \greedki^{(\ell)}_{j^\prime}
		h^{\mathcal{\bar J}^\prime, (\ell)}(t)}
		{\greedki^{(i)}_{j^\prime}} \ge \\
		& h^{\subind[m], (i)}(t) =
		\frac{\lambda_\ji{i} - \mu_\ji{i} y_\ji{i}(t) -
			\sum_{\ell=1}^{i-1} \greedki^{(\ell)}_\ji{i}
			h^{\subind[m], (\ell)}(t)}
		{\greedki^{(i)}_\ji{i}}.
		\end{aligned}
		\end{equation}
		From (\ref{eq:final1e}), (\ref{eq:final1f}) and (\ref{eq:final1ef})
		we get
		\ben
		&&\frac{x^{L(e)}_{\greedk^{(i)}}(\tau^\prime)
		- x^{L(e)}_{\greedk^{(i)}}(\tau^{(m)}_{n})}
		{\tau^\prime - \tau^{(m)}_{n}}
		=\\
		&& \frac{\lambda_\ji{i} - \mu_\ji{i} y_\ji{i}(t) -
		\sum_{\ell=1}^{i-1} \greedki^{(\ell)}_\ji{i}
			h^{\subind[m], (\ell)}(t)}
		{\greedki^{(i)}_\ji{i}}
		+ o(1),
		\een
		which holds with probability $1 - o(L^{-2})$
		and therefore it implies (\ref{eq:res1eqC}).

	\textbf{Proof of (\ref{eq:res1eqD}):}
	We will now prove (\ref{eq:res1eqD}) by considering two cases
	depending on length of $\tau - \tau^\prime$.

	We consider
	\begin{equation}\label{eq:taucond7}
	\tau - \tau^\prime \le
	\frac{\epsilon \subL}{4 (1+\Kmax)^i R}
	\end{equation}
	or
	\begin{equation}\label{eq:taucond8}
	\tau - \tau^\prime >
	\frac{\epsilon \subL}{4 (1+\Kmax)^i R},
	\end{equation}
	where $R := \sum_{j \in \mathcal{J}} \lambda_j + \Kmax\mumax$, and
	\begin{equation}\label{eq:eps_cond2}
	\epsilon < \frac{\mu_\ji{i} \alpha_i}{4 (1+\Kmax)^i R}.
	\end{equation}

	\textbf{Case (\ref{eq:taucond7}):}
		Following the same arguments as in the Case of (\ref{eq:taucond5})
		we can infer the equivalent of (\ref{eq:contr1e}) and
		(\ref{eq:contr1f}) for interval $(\tau^\prime, \tau)$, i.e.
		\begin{equation}\label{eq:contr1g}
		\lim_{L \to \infty}
		\mathds{P}\left(
		\sum_{\ell=1}^{i} \greedki^{(\ell)}_\ji{i}
		({x}^{L(e)}_{\greedk^{(\ell)}}(\tau)
		- {x}^{L(e)}_{\greedk^{(\ell)}}(\tau^\prime)) \le
		\frac{\epsilon \subL}{4}\right) = 1 - o(L^{-2})
		\end{equation}
		and
		\begin{equation}\label{eq:contr1h}
		\mathds{P}\left(
		(\lambda_\ji{i} - \mu_\ji{i} y_\ji{i}(t))(\tau - \tau^\prime) \le
		\frac{\epsilon \subL}{4}\right) = 1 - o(L^{-2}).
		\end{equation}
		which imply (\ref{eq:res1eqD}).

	\textbf{Case (\ref{eq:taucond8}):}
		First we will prove that
		\begin{equation}\label{eq:q_taup2}
		q^L_{\greedk^{(i)}, \ji{i}}(\tau^\prime) \ge
		o(L\subL) - \epsilon^2 L\subL,
		\end{equation}
		or equivalently
		\begin{equation}
		\begin{aligned}
		& 1/L q^L_{\greedk^{(i)}, \ji{i}}(\tau^\prime)
		- 1/L q^L_{\greedk^{(i)}, \ji{i}}(\tau^{(m)}_n)
		\ge o(\subL) - \epsilon^2 \subL.
		\end{aligned}
		\end{equation}
		Both will hold with probability $1 - o(L^{-2})$.
		The analysis is same as with the proof of (\ref{eq:res1eqC})
		with the following changes.
		\begin{itemize}[leftmargin=*]
			\item Instead of considering cases $\tau^\prime - \tau^{(m)}_{n} \le
			\frac{\epsilon \subL}{4 (1+\Kmax)^i R}$ and
			$\tau^\prime - \tau^{(m)}_{n} >
			\frac{\epsilon \subL}{4 (1+\Kmax)^i R}$, we should consider
			$\tau^\prime - \tau^{(m)}_{n} \le
			\frac{\epsilon^2 \subL}{2 (1+\Kmax)^i R}$ and
			$\tau^\prime - \tau^{(m)}_{n} >
			\frac{\epsilon^2 \subL}{2 (1+\Kmax)^i R}$.
			\item If interval is short, we can bound the 
            absolute change of variable 
			$1/L q^L_{\greedk^{(i)}, \ji{i}}(\tau)$ which changes by at most $(1+\Kmax)^i/L$.
			\item If interval is long, we can still prove that the equivalent
			of (\ref{eq:final1e}) is satisfied as equality and considering
			$y^L_\ji{i}(\tau^\prime) - y^L_\ji{i}(\tau^{(m)}_{n}) =
			\lambda_\ji{i} - \mu_\ji{i} y_\ji{i}(t) + o(\subL)$, we can get
			through (\ref{eq:final1e}) that
			$1/L q^L_{\greedk^{(i)}, \ji{i}}(\tau^\prime) = o(\subL)$
			with probability $1 - o(L^{-2})$.
		\end{itemize}
		In this case, because of (\ref{eq:q_taup2}) we get
		\begin{equation*}
		\begin{aligned}
		& o(L\subL) - \epsilon^2 L\subL
		\le q^L_{\greedk^{(i)},\ji{i}}(\tau^\prime) =
		q^L_{\greedk^{(i)},\ji{i}}(\tau) \\
		&- \sum_{\ell=1}^i \greedki^{(\ell)}_\ji{i} L
		(x^{L(e)}_{\greedk^{(\ell)}}(\tau)
		- x^{L(e)}_{\greedk^{(\ell)}}(\tau^\prime))
		+ L y^L_\ji{i}(\tau) - L y^L_\ji{i}(\tau^\prime) \le \\
		& \Kmax
		- \sum_{\ell=1}^i \greedki^{(\ell)}_\ji{i} L
		(x^{L(e)}_{\greedk^{(\ell)}}(\tau)
		- x^{L(e)}_{\greedk^{(\ell)}}(\tau^\prime))
		+ L (\lambda_\ji{i} - \mu_\ji{i} y_\ji{i}(t))(\tau - \tau^\prime) \\
		& + o(L\subL),
		\end{aligned}
		\end{equation*}
		from which it follows
		\begin{equation}\label{eq:final1g}
		\begin{aligned}
		& \frac{x^{L(e)}_{\greedk^{(i)}}(\tau)
		- x^{L(e)}_{\greedk^{(i)}}(\tau^\prime)}{\tau - \tau^\prime} \le
		\frac{\lambda_\ji{i} - \mu_\ji{i} y_\ji{i}(t)}
		{\greedki^{(i)}_\ji{i}} \\
		& - \sum_{\ell=1}^{i-1} \frac{\greedki^{(\ell)}_\ji{i}
		\left(x^{L(e)}_{\greedk^{(\ell)}}(\tau)
		- x^{L(e)}_{\greedk^{(\ell)}}(\tau^\prime)\right) - \epsilon^2 \subL}
		{\greedki^{(i)}_\ji{i}(\tau - \tau^\prime)}
		+ \frac{o(\subL)}{\tau - \tau^\prime} 	\stackrel{(a)}{<} \\
		& \mu_\ji{i} \sum_{\ell=1}^i
		\frac{\greedki^{(\ell)}_\ji{i}}{\greedki^{(i)}_\ji{i}}
		(\globlt{x}^{(g)}_{\greedk^{(\ell)}} - x^{(e)}_{\greedk^{(\ell)}}(t))
		- \sum_{\ell=1}^{i-1}
		\frac{\greedki^{(\ell)}_\ji{i}
			h^{\subind[m], (\ell)}(t)}{\greedki^{(i)}_\ji{i}} \\
		& + \frac{\epsilon 4 (1+\Kmax)^i R}{\greedki^{(i)}_\ji{i}}
		+ o(1)
		\stackrel{(b)}{<}
		\mumax (\alpha_i + \sum_{\ell=1}^{i-1}\Kmax \alpha_\ell) \\
		&+ \sum_{\ell=1}^{i-1}
		\Kmax\mumax
		\left(\alpha_\ell + \sum_{\ell^\prime=1}^{\ell-1}
		2(1+\Kmax)^{\ell-\ell^\prime} \alpha_{\ell^\prime}\right) + \\
		& 4 \epsilon (1 + \Kmax)^i R + o(1) \stackrel{(c)}{<}
		2 \mumax \alpha_i +
		2 \mumax \sum_{\ell=1}^{i-1} (1 + \Kmax)^{i-\ell} \alpha_\ell + o(1).
		\end{aligned}
		\end{equation}
		In (a) we applied the properties of $\ji{i}$, from (\ref{eq:yjA}), (\ref{eq:defcalJ}), and
		(\ref{eq:xe_h}) of Lemma~\ref{lem:hrec}, for indexes $1, \ldots, i-1$, and then replaced $\tau - \tau^\prime$ with its bound
		from \dref{eq:taucond8}. In (b), we used property (\ref{eq:alpha_prop}), and the fact that  $\frac{\greedki^{(\ell)}_\ji{i}}{\greedki^{(i)}_\ji{i}} \le \Kmax$
		for $\ell=1, \ldots, i-1$.
		In (c), we used (\ref{eq:eps_cond2}) and simplified.

		Also from Lemma~\ref{lem:emptyrate} and given that all the servers
		that empty during $[\tau^\prime, \tau)$ will be assigned to
		configuration $\greedk^{(\ell)}$ for $\ell=1, \ldots, i$,
		we will have that with probability $1 - o(L^{-2})$,
		\begin{equation}\label{eq:final1h}
		\begin{aligned}
		& \frac{x^{L(e)}_{\greedk^{(i)}}(\tau)
		- x^{L(e)}_{\greedk^{(i)}}(\tau^\prime)}{\tau - \tau^\prime}
		\ge	\frac{\mumin}{J \Kmax \cfcnt^2}
		\left(1 - \sum_{\ell=1}^{i} x^{(e)}_{\greedk^{(\ell)}}(t) \right) \\
		&- \sum_{\ell=1}^{i-1}
		\frac{\left(x^{L(e)}_{\greedk^{(\ell)}}(\tau)
		- x^{L(e)}_{\greedk^{(\ell)}}(\tau^\prime)\right)^+}
		{\tau - \tau^\prime}
		+ \frac{o(\subL)}{\tau - \tau^\prime} \stackrel{(a)}{=} \\
		& \frac{\mumin}{J \Kmax \cfcnt^2}
		\left(1 - \sum_{\ell=1}^{i} x^{(e)}_{\greedk^{(\ell)}}(t) \right)
		- \sum_{\ell=1}^{i-1} h^{\subind[m], (\ell)}(t)^+ \\
		& + o(1) \stackrel{(b)}{>}
		\frac{\mumin}{J \Kmax \cfcnt^2}
		\frac{\globlt{x}_{\greedk^{(\gtcnt)}}}{2}
		- 2 \mumax
		\sum_{\ell=1}^{i-1} \frac{(1 + \Kmax)^{i-\ell}}{\Kmax}
		\alpha_\ell + o(1).
		\end{aligned}
		\end{equation}
		In (a) we applied (\ref{eq:xe_h}) of Lemma~\ref{lem:hrec}
		for indexes $1, \ldots, i-1$ and used (\ref{eq:taucond8}).
		In (b) we applied Lemma~\ref{lem:rank0frac}
		and equation (\ref{eq:hbound}) for
		indexes $1, \ldots, i-1$ and simplified.

		So far we proved that if (\ref{eq:taucond8}) holds,
		(\ref{eq:final1g}) and (\ref{eq:final1h}) also hold.

		However, considering (\ref{eq:alphadef}) we have
		\begin{equation}\label{eq:contr2}
		\begin{aligned}
		&2 \mumax \alpha_i +
		2 \mumax \sum_{\ell=1}^{i-1} (1 + \Kmax)^{i-\ell} \alpha_\ell
		\stackrel{(a)}{<} \\
		& \frac{\mumin}{J \Kmax \cfcnt^2}\frac{\globlt{x}_{\greedk^{(\gtcnt)}}}{2}
		- 2 \mumax \sum_{\ell=1}^{i-1} \frac{(1 + \Kmax)^{i-\ell}}{\Kmax}
		\alpha_\ell
		\end{aligned}
		\end{equation}
		and because of that, the probability
		that (\ref{eq:final1g}) and (\ref{eq:final1h})
		are both true is $o(L^{-2})$.
		This means that (\ref{eq:taucond7}) holds
		with probability at least $1 - o(L^{-2})$,
		and thus the analysis of (\ref{eq:taucond7}) is
		sufficient for (\ref{eq:res1eqB}) to hold.

%% file: anymu_P1Claim2.tex
%\begin{claim}\label{cl:P1B}
%		For every $i \in \{1, \ldots, \ell_m\}$, 
%		we have that
%		$q^L_{\greedk^{(i)}, j^{(i)}(\tau^{(m-1)}_n)} = o_p(f(L))$.
%\end{claim}

We will prove the result inductively on $m$.
For $m=0$, $\ell_m = 0$ so there is nothing to prove,
so we will start with the base case $m=1$
and then move on to the inductive step.

\textbf{Base case $m=1$:}

In this case $\ell_m=1$ and
$0 \le q^L_{\greedk^{(1)}, \jiB{1}{1}}(\tau^{(1)}_n) < \Kmax$
from which it trivially follows that 
$q^L_{\greedk^{(1)}, \jiB{1}{1}}(\tau^{(1)}_n) = o(L\subL)$.

%\textbf{Inductive case $m>1$:}
%For this case we define $\ji{i} := \jiB{i}[m]$
%for $i=1, \ldots, \ell_m$ to simplify notation.
% q + y = sum k xe
% if at previous time step is o_p(f(L)) then
% q(now) = q(prev) + (sum k Dxe - Dy) >=? o_p(f(L)) + o_p(f(L)) 

\textbf{Inductive case $m>1$, Base case $i=\ell_m$:}
In this case
$$
0 \le q^L_{\greedk^{(i)}, \jiB{i}{m}}(\tau^{(m)}_n) < \Kmax
$$
from which it trivially follows that 
$$q^L_{\greedk^{(i)}, \jiB{i}{m}}(\tau^{(m)}_n) = o(L\subL).$$

\textbf{Inductive case $m>1$, Base case $i<\ell_m$:}
In this case, we know by the induction hypothesis that
$q^L_{\greedk^{(i)}, \jiB{i}{m-1}}(\tau^{(m-1)}_n) = o(L\subL)$. 
Further, notice that $\jiB{i}{m-1} = \jiB{i}{m}$ for $i=1, \ldots, \ell_m-1$
so for simplicity we will refer to both as $\ji{i}$.
Then we have that with probability $1 - o(L^{-2})$,
\begin{equation}
\begin{aligned}
&q^L_{\greedk^{(i)}, \ji{i}}(\tau^{(m)}_n) \ge 
%q^L_{\greedk^{(i)}, \ji{i}(\tau^{(m-1)}_n)}(\tau^{(m)}_n) \ge \\
q^L_{\greedk^{(i)}, \ji{i}}(\tau^{(m-1)}_n) + \\
&L \left(
\sum_{\ell=1}^{i} \greedki^{(\ell)}_{\ji{i}} 
\left(x^{(e)}_{\greedk^{(\ell)}}(\tau^{(m)}_n) -
x^{(e)}_{\greedk^{(\ell)}}(\tau^{(m-1)}_n) \right)
- (\lambda_{\ji{i}} - \mu_{\ji{i}} y_{\ji{i}}(t)) \right) \times \\
& (\tau^{(m)}_n - \tau^{(m-1)}_n) + o(L\subL) \stackrel{(a)}{=}
o(L\subL).
\end{aligned}
\end{equation}
In (a) we have used that (\ref{eq:exact}) holds
for index $m-1$ in place of $m$.

%% file: anymu_P2.tex
	\begin{definition}\label{def:lastj}
		For an index $i$ for which $1 \le i \le \ell_m$, the time ${\tau}^{(m,i)}_{n}$ 
		is defined as the latest time in $[\tau^{(m)}_{n}, \tau^{(m+1)}_{n})$, such that
		\begin{equation}
		\max_{j \in \mathcal{J}: \greedki^{(i)}_j > 0}
		{q^L_{\greedk^{(i)}, j}({\tau}^{(m,i)}_{n})} \ge 0.
		\end{equation}
		Also let $\jix{i}:=\arg\max_{j \in \mathcal{J}: \greedki^{(i)}_j > 0}
		{q^L_{\greedk^{(i)}, j}({\tau}^{(m,i)}_{n})}$.
	\end{definition}

	For (\ref{pr:bound1}) to be true,
	it suffices that for any $\epsilon > 0$
	\begin{equation}\label{eq:res3eq}
	\begin{aligned}
	&\mathds{P} \Big(
	x^{(e)}_{\greedk^{(\ell_m)}}(\tau^{(m+1)}_{n}) 
	- x^{(e)}_{\greedk^{(\ell_m)}}(\tau^{(m)}_{n}) \\
	& - \mumin \frac{\alpha_{\ell_m}}{2} (\tau - \tau^{(m)}_{n})
	> - \epsilon \subL\Big) = 1 - o(L^{-2}).
	\end{aligned}
	\end{equation}

	To prove (\ref{eq:res3eq}), it is sufficient to prove
	\begin{equation}\label{eq:ralpha1}
	\begin{aligned}
	&\mathds{P} \left(
	x^{(e)}_{\greedk^{(\ell_m)}}({\tau}^{(m,\ell_m)}_{n}) 
	- x^{(e)}_{\greedk^{(\ell_m)}}(\tau^{(m)}_{n}) \right. \\
	& \left. - \mumin \frac{\alpha_{\ell_m}}{2} 
	\left({\tau}^{(m,\ell_m)}_{n} - \tau^{(m)}_{n}\right)
	> - \frac{\epsilon \subL}{2}
	\right) = 1 - o(L^{-2}),
	\end{aligned}
	\end{equation}
	and
	\begin{equation}\label{eq:ralpha2}
	\begin{aligned}
	&\mathds{P} \left(
	x^{(e)}_{\greedk^{(\ell_m)}}(\tau^{(m+1)}_{n}) 
	- x^{(e)}_{\greedk^{(\ell_m)}}(\tau^{(m,\ell_m)}_{n}) \right. \\
	&\left. - \mumin \frac{\alpha_{\ell_m}}{2} (\tau^{(m+1)}_{n} - \tau^{(m,\ell_m)}_{n})
	> - \frac{\epsilon \subL}{2} \right) = 1 - o(L^{-2}).
	\end{aligned}
	\end{equation}

	%%%%%%%%%%%%%%%%%%%%%%%%%%%%%%%%%%%%
	\textbf{Proof of (\ref{eq:ralpha1}):}
	We will now prove (\ref{eq:ralpha1}) by considering two cases
	depending on length of $\tau^{(m,\ell_m)}_{n} - \tau^{(m)}_{n}$.
	
	We consider
	\begin{equation}\label{eq:taucond2_1}
	\tau^{(m,\ell_m)}_{n} - \tau^{(m)}_{n} \le 
	\frac{\epsilon \subL}{4(1+\Kmax)^{\ell_m-1} R}
	\end{equation}
	or
	\begin{equation}\label{eq:taucond2_2}
	\tau^{(m,\ell_m)}_{n} - \tau^{(m)}_{n} > 
	\frac{\epsilon \subL}{4(1+\Kmax)^{\ell_m-1} R}
	\end{equation}
	where $R:= \sum_{j \in \mathcal{J}} \lambda_j 
	+ \Kmax\mumax$
	
	\textbf{Case (\ref{eq:taucond2_1}):}
		We notice ${x}^{(e)}_{\greedk^{(\ell_m)}}(\tau)$ 
		will change by at most $(1+\Kmax)^{\ell_m-1}/L$ at each arrival
		or departure according to Lemma~\ref{lem:X_change}.

		Using Lemma~\ref{lem:Chern_Poisson}, with probability $1 - o(L^{-2})$, The number of arrivals and departures 
		in $[\tau^{(m)}_{n}, \tau^{(m,\ell_m)}_{n}]$
		can be bounded by
		\begin{equation}
		\begin{aligned}
		&\left(\sum_{j \in \mathcal{J}} \lambda_j 
		+ \Kmax\mumax \right) L
		\frac{\epsilon \subL}{4(1+\Kmax)^{\ell_m-1} R} + o(L\subL) = \\
		&\frac{\epsilon L \subL}{4(1+\Kmax)^{\ell_m-1}} + o(L\subL),
		\end{aligned}
		\end{equation} 
		therefore
		\begin{equation}\label{eq:contr2a}
		\mathds{P}\left(
		{x}^{(e)}_{\greedk^{(\ell_m)}}(\tau^{(m,\ell_m)}_{n}) 
		- {x}^{(e)}_{\greedk^{(\ell_m)}}(\tau^{(m)}_{n})  \le
		\frac{\epsilon \subL}{4}\right) = 1 - o(L^{-2}).
		\end{equation}
		
		Considering (\ref{eq:taucond2_1}) holds, we trivially have
		\begin{equation}\label{eq:contr2b}
		\begin{aligned}
		&\mathds{P}\left(
		\mumin \frac{\alpha_{\ell_m}}{2}(\tau^{(m,\ell_m)}_{n} - \tau^{(m)}_{n}) \le
		\frac{\epsilon \subL}{4}\right) = 1 %\\
		\end{aligned}
		\end{equation}
		It is now easy to verify that (\ref{eq:contr2a}) 
		and (\ref{eq:contr2b}) imply (\ref{eq:ralpha1}).
		
	\textbf{Case (\ref{eq:taucond2_2}):}

		Let $j := \jix{\ell_m}$.
		In this case, we have
		\begin{equation}
		\begin{aligned}
		0 & \le q^L_{\greedk^{(\ell_m)},j}(\tau^{(m, \ell_m)}_{n}) = 
		q^L_{\greedk^{(\ell_m)},j}(\tau^{(m)}_{n}) \\
		&+ \sum_{\ell=1}^{\ell_m} \greedki^{(\ell)}_j L 
		\left(x^{L(e)}_{\greedk^{(\ell)}}(\tau^{(m, \ell_m)}_{n})  
		- x^{L(e)}_{\greedk^{(\ell)}}(\tau^{(m)}_{n})\right) \\
		& - L \left( y^L_j(\tau^{(m, \ell_m)}_{n}) 
		- y^L_j(\tau^{(m)}_{n}) \right)  \\
		&\le \Kmax	+ \sum_{\ell=1}^{\ell_m} \greedki^{(\ell)}_j L
		(x^{L(e)}_{\greedk^{(\ell)}}(\tau^{(m, \ell_m)}_{n}) 
		- x^{L(e)}_{\greedk^{(\ell)}}(\tau^{(m)}_{n})) \\
		&- L(\lambda_j - \mu_j y_j(t))
		\left(\tau^{(m, \ell_m)}_{n} - \tau^{(m)}_{n}\right)
		+ o(L\subL),
		\end{aligned}
		\end{equation}
		which holds with probability $1 - o(L^{-2})$ by application of 
		Lemma~\ref{lem:Chern_Poisson} in the last step.
		If $\ell_m > \gtcntp{\alpha}(t)$ then,
		since (\ref{eq:alpha_prop}) does not hold for
		$i=\ell_m$, we have
		\begin{equation} \label{eq:P2b1}
		\frac{\lambda_j - \mu_j y_j(t)}{\greedki^{(\ell_m)}_j} >
		\frac{\mu_j \alpha_{\ell_m} + \mu_j
			\sum_{\ell=1}^{\ell_m-1} \greedki^{(\ell)}_j 
			(\globlt{x}^{(g)}_{\greedk^{(\ell)}} - x^{(e)}_{\greedk^{(\ell)}}(t))
		}{\greedki^{(\ell_m)}_j}, 
		\end{equation}
		otherwise, by property (\ref{eq:non_trivial}) for $j_a = j$
		and $j_b = \ell_m$ we have
		\begin{equation} \label{eq:P2b2}
		\frac{\lambda_j - \mu_j y_j(t)}{\greedki^{(\ell_m)}_j} >
		\frac{\mu_j \delta + \mu_j
			\sum_{\ell=1}^{\ell_m} \greedki^{(\ell)}_j 
			(\globlt{x}^{(g)}_{\greedk^{(\ell)}} - x^{(e)}_{\greedk^{(\ell)}}(t))
		}{\greedki^{(\ell_m)}_j}. 
		\end{equation}

		Based on the above inequality, we can write
		\begin{equation}\label{eq:final2a}
		\begin{aligned}
		& \frac{x^{L(e)}_{\greedk^{(\ell_m)}}(\tau^{(m, \ell_m)}_{n})  - 
		x^{L(e)}_{\greedk^{(\ell_m)}}(\tau^{(m)}_{n})}
		{\tau^{(m, \ell_m)} - \tau^{(m)}_{n}} \ge \\
		& \frac{\lambda_j - \mu_j y_j(t)}{\greedki^{(\ell_m)}_j} 
		- \sum_{\ell=1}^{\ell_m-1} 
		\frac{\greedki^{(\ell)}_j
		\left(x^{(e)}_{\greedk^{(\ell)}}(\tau^{(m, \ell_m)}_{n}) 
		- x^{(e)}_{\greedk^{(\ell)}}(\tau^{(m)}_{n})\right)		
		}
		{\greedki^{(\ell_m)}_j (\tau^{(m, \ell_m)}_{n} - \tau^{(m)}_{n})} \\
		& + \frac{o(\subL)}{\tau^{(m, \ell_m)}_n - \tau^{(m)}_{n}}
		\stackrel{(a)}{\ge} 
		\mu_j \min\left(
		\frac{\delta}{\greedki^{(\ell_m)}_j} -
		(\globlt{x}^{(g)}_{\greedk^{(\ell_m)}} - x^{(e)}_{\greedk^{(\ell_m)}}(t)), 
		\alpha_{\ell_m}
		\right) \\
		& + \frac{\mu_j
			\sum_{\ell=1}^{\ell_m-1} \greedki^{(\ell)}_j 
			(\globlt{x}^{(g)}_{\greedk^{(\ell)}} - x^{(e)}_{\greedk^{(\ell)}}(t))
			- \sum_{\ell=1}^{\ell_m-1} 
			\greedki^{(\ell)}_j h^{\subind[m], (\ell)}(t)
		}{\greedki^{(\ell_m)}_j} \\
		& + o(1) \stackrel{(b)}{\ge} \left(
		\mumin \min \left(\delta 
		-(\epsilon_\rho + \Kmax \sum_{\ell=1}^{\ell_m-1} \alpha_{\ell}),
		\alpha_{\ell_m} \right)
		\right. \\
		&\left. - \mumax
			\sum_{\ell=1}^{\ell_m-1} \Kmax
			(\epsilon_\rho + \Kmax \sum_{\ell^\prime=1}^{\ell-1} \alpha_{\ell^\prime}) \right. \\
		&\left.
			- \sum_{\ell=1}^{\ell_m-1} \Kmax \mumax 
			\left(\alpha_\ell +
			\sum_{\ell^\prime=1}^{\ell-1} 2(1 + \Kmax)^{\ell-\ell^\prime} 
			\alpha_{\ell^\prime} \right)
		\right) + o(1) \stackrel{(c)}{>} \\
		& \frac{\mumin \alpha_{\ell_m}}{2} + o(1).
		\end{aligned}
		\end{equation}
		In (a) we used: 1) (\ref{eq:P2b1}) and (\ref{eq:P2b2}), 
		2) equation (\ref{eq:hbound}) for
		indexes $1, \ldots, \ell_m-1$, and 3) replaced
		$\tau^{(m, \ell_m)} - \tau^{(m)}_{n}$ with its bound from 
		(\ref{eq:taucond2_2}). In (b) we used Lemma~\ref{lem:xg_xebound}
		for indexes $1, \ldots, \ell_m$. Finally, in (c) we used 
		(\ref{eq:alphadef}) and (\ref{eq:epsrho_cond}). 		

		We have thus shown by (\ref{eq:final2a}) that
		$$\frac{x^{L(e)}_{\greedk^{(\ell_m)}}(\tau^{(m, \ell_m)}_{n})  - 
		x^{L(e)}_{\greedk^{(\ell_m)}}(\tau^{(m)}_{n})}
		{\tau^{(m, \ell_m)} - \tau^{(m)}_{n}} > 
		\frac{\mumin \alpha_{\ell_m}}{2} + o(1)
		$$
		with probability $1 - o(L^{-2})$
		which is equivalent to (\ref{eq:ralpha1}).
	
	\textbf{Proof of (\ref{eq:ralpha2}):}
	We will now prove (\ref{eq:ralpha2}) by considering two cases
	depending on length of $\tau^{(m+1)}_{n} - \tau^{(m, \ell_m)}_{n}$ 
	and reach a contradiction for each of them.

	We consider
	\begin{equation}\label{eq:taucond2_3}
	\tau^{(m+1)}_{n} - \tau^{(m,\ell_m)}_{n} \le 
	\frac{\epsilon \subL}{4(1+\Kmax)^{\ell_m-1}(\sum_{j \in \mathcal{J}} \lambda_j 
		+ \Kmax\mumax)}
	\end{equation}
	or
	\begin{equation}\label{eq:taucond2_4}
	\tau^{(m+1)}_{n} - \tau^{(m, \ell_m)}_{n} > 
	\frac{\epsilon \subL}{4(1+\Kmax)^{\ell_m-1}(\sum_{j \in \mathcal{J}} \lambda_j 
		+ \Kmax\mumax)}
	\end{equation}
	
	\textbf{Case (\ref{eq:taucond2_3}):}
		Following the same arguments as in the Case of (\ref{eq:taucond2_1})
		we can infer the equivalent of (\ref{eq:contr2a}) and 
		(\ref{eq:contr2b}) for interval 
		$(\tau^{(m, \ell_m)}_{n}, \tau^{(m+1)}_{n})$, i.e.
		\begin{equation}\label{eq:contr2c}
		\mathds{P}\left(
		{x}^{(e)}_{\greedk^{(\ell_m)}}(\tau^{(m+1)}_{n}) 
		- {x}^{(e)}_{\greedk^{(\ell_m)}}(\tau^{(m,\ell_m)}_{n}) \le
		\frac{\epsilon \subL}{4}\right) = 1 - o(L^{-2}),
		\end{equation}
		and
		\begin{equation}\label{eq:contr2d}
		\mathds{P}\left(
		\frac{\mumin \alpha_{\ell_m}}{2}(\tau^{(m+1)}_{n} 
		- \tau^{(m,\ell_m)}_{n}) \le
		\frac{\epsilon \subL}{4}\right) = 1 - o(L^{-2}).
		\end{equation}
		which imply (\ref{eq:ralpha2}).

	\textbf{Case (\ref{eq:taucond2_4}):}
		In this case we notice, using Lemma~\ref{lem:emptyrate}
		that with probability $1 - o(L^{-2})$
		\begin{equation}\label{eq:final2b_}
		\begin{aligned}
		& (x^{L(e)}_{\greedk^{(\ell_m)}}(\tau^{(m+1)}_{n} ) 
		- x^{L(e)}_{\greedk^{(\ell_m)}}(\tau^{(m,\ell_m)}_{n})) \\
		& + \sum_{\ell=1}^{\ell_m-1}
		\left(x^{L(e)}_{\greedk^{(\ell)}}(\tau^{(m+1)}_{n} ) 
		- x^{L(e)}_{\greedk^{(\ell)}}(\tau^{(m,\ell_m)}_{n})\right)^+ > \\
		& \frac{\mumin}{J \Kmax \cfcnt^2}
		\left(1 - \sum_{\ell=1}^{\ell_m} x^{(e)}_{\greedk^{(\ell)}}(t) \right)
		\left(\tau^{(m+1)}_{n} - \tau^{(m,\ell_m)}_{n} \right) + o(\subL),
		\end{aligned}
		\end{equation}
		or equivalently
		\begin{equation}\label{eq:final2b}
		\begin{aligned}
		&\frac{x^{L(e)}_{\greedk^{(\ell_m)}}(\tau^{(m+1)}_{n}) 
		- x^{L(e)}_{\greedk^{(\ell_m)}}(\tau^{(m,\ell_m)}_{n})}
		{\tau^{(m+1)}_{n} - \tau^{(m,\ell_m)}_{n}} 
		\stackrel{(a)}{>} \\
		&\frac{\mumin}{J \Kmax \cfcnt^2}
		\left(1 - \sum_{\ell=1}^{\ell_m} x^{(e)}_{\greedk^{(\ell)}}(t) \right)
		- \sum_{\ell=1}^{\ell_m-1} h^{(\ell)}(t)^+
		+ \frac{o(\subL)}{\tau^{(m+1)}_{n} - \tau^{(m,\ell_m)}_{n}}
		\stackrel{(b)}{>} \\
		& \frac{\mumin}{J \Kmax \cfcnt^2}
		\frac{\globlt{x}^{(g)}_{\greedk^{(\gtcnt)}}}{2}
		- \sum_{\ell=1}^{\ell_m-1} \mumax \left(
		\alpha_\ell + 2\sum_{\ell^\prime=1}^{\ell-1} 
		(1 + \Kmax)^{\ell-\ell^\prime} \alpha_{\ell^\prime}\right)
		+ o(1)
		\stackrel{(c)}{>} \\
		& \frac{\mumin \alpha_{\ell_m}}{2} + o(1).
		\end{aligned}
		\end{equation}
		Inequality (a) comes from applying to (\ref{eq:final2b_}),
		the equation (\ref{eq:xe_h}) of Lemma~\ref{lem:hrec}
		for indexes $1, \ldots, \ell_m-1$. 
		In (b) we used Lemma~\ref{lem:rank0frac} 
		and equation (\ref{eq:hbound}) for
		indexes $1, \ldots, \ell_m-1$.
		Finally, in (c) we used (\ref{eq:alphadef}).
		
		We have thus shown through (\ref{eq:final2b}) that
		$$\frac{x^{L(e)}_{\greedk^{(\ell_m)}}(\tau^{(m+1)}_{n})  - 
		x^{L(e)}_{\greedk^{(\ell_m)}}(\tau^{(m, \ell_m)}_{n})}
		{ \tau^{(m+1)}_{n} - \tau^{(m, \ell_m)}} > 
		\frac{\mumin \alpha_{\ell_m}}{2} + o(1)
		$$
		with probability $1 - o(L^{-2})$,
		which is equivalent to (\ref{eq:ralpha2}).

%% file: anymu_P3.tex
	Consider $j = \jix{\gtcnt}$ given from Definition~\ref{def:lastj}.
	
	For (\ref{pr:bound2}) to be true,
	it suffices to prove that for any $\epsilon > 0$,
	\begin{equation}\label{eq:res3eqA}
	\begin{aligned}
	&P_1 := \mathds{P} \left(
	\sum_{\ell=1}^{\gtcnt} \greedki^{(\ell)}_j
	\left( x^{(e)}_{\greedk^{(\ell)}}(\tau^{(m,\gtcnt)}_{n}) 
	- x^{(e)}_{\greedk^{(\ell)}}(\tau^{(m)}_{n}) \right) 
	\right. \\
	& \left. 
	- (\lambda_j - \mu_j y_j(t)) ({\tau}^{(m,\gtcnt)}_{n} - \tau^{(m)}_{n})
	> - \frac{\epsilon \subL}{2}
	\right) = 1 - o(L^{-2}),
	\end{aligned}
	\end{equation}
	and
	\begin{equation}\label{eq:res3eqB}
	\begin{aligned}
	&P_2 := \mathds{P} \Bigg(
	x^{(e)}_{\greedk^{(\gtcnt)}}(\tau^{(m+1)}_{n}) 
	- x^{(e)}_{\greedk^{(\gtcnt)}}(\tau^{(m,\gtcnt)}_{n}) \\
	& + \sum_{\ell=1}^{\gtcnt-1}
	(x^{(e)}_{\greedk^{(\ell)}}(\tau^{(m+1)}_{n}) 
	- x^{(e)}_{\greedk^{(\ell)}}(\tau^{(m,\gtcnt)}_{n}))^+ \\
	& - \mumin \frac{1 - \sum_{i=1}^{\gtcnt} {x}^{(e)}_{\greedk^{(i)}}(t)}
	{J \Kmax \cfcnt^2} (\tau^{(m+1)}_{n} - \tau^{(m,\gtcnt)}_{n})
	> - \frac{\epsilon \subL}{2}\Bigg) \\
	&= 1 - o(L^{-2}).
	\end{aligned}
	\end{equation}

	% lim P >= lim P (fst> sth) * lim P (sec > sth) >=
	% lim P (fst >= sthA) * lim P (sec > sthB) = 
	% lim P (fst >= sthA specific) * lim P (sec > sthB specific) = 1
	
	% lim P (x(t) > o(L)) = P (x(t) > 0) 
	To show why this is sufficient we first introduce the following notations
	\begin{equation}
	\begin{aligned}
	&f_1[\tau_a, \tau_b] := \mumin
	\frac{1 - \sum_{i=1}^{\gtcnt} 
		{x}^{(e)}_{\greedk^{(i)}}(t)}{J \Kmax \cfcnt^2}
	-\sum_{i=1}^{\gtcnt-1}  \nabla {x}^{(e)}_{\greedk^{(i)}}
	[\tau_a, \tau_b]^+ \\
	&f_2[\tau_a, \tau_b] := \min_{j \in \mathcal{J}}
	\left(
	\frac{\lambda_j - \mu_j y_j(t) 
		- \sum_{i=1}^{\gtcnt-1} \greedki^{(i)}_j 
		\nabla {x}^{(e)}_{\greedk^{(i)}}
		[\tau_a, \tau_b]}{\greedki^{(\gtcnt)}_j} \right)
	\end{aligned}
	\end{equation}
	
	Then if (\ref{eq:res3eqA}) and (\ref{eq:res3eqB})
	indeed hold, we can get that
	\begin{equation}
	\begin{aligned}
	&\mathds{P} \left( x^{(e)}_{\greedk^{(\gtcnt)}}(\tau^{(m+1)}_{n}) 
	- x^{(e)}_{\greedk^{(\gtcnt)}}(\tau^{(m)}_{n}) > \right. \\
	&\left. \min ( f_1[\tau^{(m)}_{n}, \tau^{(m+1)}_{n}], 
	f_2[\tau^{(m)}_{n}, \tau^{(m+1)}_{n}]) 
	- \epsilon \subL \right) \ge \\
	&\mathds{P} \left(x^{(e)}_{\greedk^{(\gtcnt)}}(\tau^{(m, \gtcnt)}_{n}) 
	- x^{(e)}_{\greedk^{(\gtcnt)}}(\tau^{(m)}_{n}) > \right. \\
	& \left. \min(f_1[\tau^{(m)}_{n}, \tau^{(m, \gtcnt)}_{n}], 
	f_2[\tau^{(m)}_{n}, \tau^{(m, \gtcnt)}_{n}]) 
	- \frac{\epsilon \subL}{2} \right)  \\
	& \mathds{P} \left(x^{(e)}_{\greedk^{(\gtcnt)}}(\tau^{(m+1)}_{n}) 
	- x^{(e)}_{\greedk^{(\gtcnt)}}(\tau^{(m, \gtcnt)}_{n}) > \right. \\
	&\left. \min(f_1[\tau^{(m, \gtcnt)}_{n}, \tau^{(m+1)}_{n}], 
	f_2[\tau^{(m, \gtcnt)}_{n}, \tau^{(m+1)}_{n}]) 
	- \frac{\epsilon \subL}{2}\right) \\
	& \ge P_1 P_2 = 1 - o(L^{-2}).
	\end{aligned}
	\end{equation}

	%%%%%%%%%%%%%%%%%%%%%%%%%%%%%%%%%%%%%
	\textbf{Proof of (\ref{eq:res3eqA}):}
	We will now prove (\ref{eq:res3eqA}) by considering two cases
	depending on length of
	$\tau^{(m,\gtcnt)}_{n} - \tau^{(m)}_{n}$.
	
	We consider
	\begin{equation}\label{eq:taucond3_1}
	\tau^{(m,\gtcnt)}_{n} - \tau^{(m)}_{n} \le 
	\frac{\epsilon \subL}{4(1+\Kmax)^{\gtcnt}R}
	\end{equation}
	or
	\begin{equation}\label{eq:taucond3_2}
	\tau^{(m,\gtcnt)}_{n} - \tau^{(m)}_{n} > 
	\frac{\epsilon \subL}{4(1+\Kmax)^{\gtcnt}R},
	\end{equation}
	where $R := \sum_{j \in \mathcal{J}} \lambda_j + \Kmax\mumax$.
	
	\textbf{Case (\ref{eq:taucond3_1}):}
		We notice ${x}^{L(e)}_{\greedk^{(\ell)}}(\tau)$ 
		will change by at most $(1+\Kmax)^{\ell-1}/L$ at each arrival
		or departure according to Lemma~\ref{lem:X_change}
		and thus 
		$\sum_{\ell=1}^{\gtcnt} \greedki^{(\ell)}_j 
		{x}^{L(e)}_{\greedk^{(\ell)}}(\tau)$
		will change by at most $(1+\Kmax)^\gtcnt/L - 1/L$.

		Using Lemma~\ref{lem:Chern_Poisson}, 
		with probability $1 - o(L^{-2})$, the number of arrivals and departures 
		in interval $[\tau^{(m)}_{n}, \tau^{(m,\gtcnt)}_{n}]$
		of length at most 
		$\frac{\epsilon \subL}{4(1+\Kmax)^{\gtcnt} R}$
        is at most
		\begin{equation*}
		\begin{aligned}
		&\left(\sum_{j \in \mathcal{J}} \lambda_j 
		+ \Kmax\mumax \right) L
		\frac{\epsilon \subL}{4(1+\Kmax)^\gtcnt R
		} + o(L\subL) = \\
		&\frac{\epsilon L\subL}{4(1+\Kmax)^\gtcnt} 
		+ o(L\subL),
		\end{aligned}
		\end{equation*} 
		therefore
		\begin{equation}\label{eq:contr3a}
		\mathds{P}\left(
		\sum_{\ell=1}^{\gcnt} \greedki^{(\ell)}_j
		({x}^{(e)}_{\greedk^{(\ell)}}(\tau^{(m,\gtcnt)}_{n}) 
		- {x}^{(e)}_{\greedk^{(\ell)}}(\tau^{(m)}_{n}) ) \le
		\frac{\epsilon \subL}{4}\right) = 1 - o(L^{-2}).
		\end{equation}

		Considering (\ref{eq:taucond3_1}) holds, we will also have
		\begin{equation}\label{eq:contr3b}
		\mathds{P}\left(
		(\lambda_j - \mu_j y_j(t))(\tau^{(m,\gtcnt)}_{n} - \tau^{(m)}_{n}) \le
		\frac{\epsilon \subL}{4}\right) = 1.
		\end{equation}
		It is now easy to verify that equations (\ref{eq:contr3a}) 
		and (\ref{eq:contr3b}) imply (\ref{eq:res3eqA}).

	\textbf{Case (\ref{eq:taucond3_2}):}
		In this case we notice, using Lemma~\ref{lem:Chern_Poisson} for
		the process of jobs of type $j$ in the system 
		which is Poisson with rate 
		$L(\lambda_{j} - \mu_{j} y_{j}(t))$,
		that with probability $1 - o(L^{-2})$
		\begin{equation}
		\begin{aligned}
		0 & \le q^L_{\greedk^{(\gtcnt)},j}(\tau^{(m, \ell_m)}_{n}) = 
		q^L_{\greedk^{(\gtcnt)},j}(\tau^{(m)}_{n}) + \\
		&\sum_{\ell=1}^{\gtcnt} \greedki^{(\ell)}_j L 
		(x^{L(e)}_{\greedk^{(\ell)}}(\tau^{(m, \ell_m)}_{n})  
		- x^{L(e)}_{\greedk^{(\ell)}}(\tau^{(m)}_{n})) \\
		& - L y^L_j(\tau^{(m, \gtcnt)}_{n}) - L y^L_j(\tau^{(m)}_{n})  \\
		&\le \Kmax
		+ \sum_{\ell=1}^{\gtcnt} \greedki^{(\ell)}_j L
		\left(x^{L(e)}_{\greedk^{(\ell)}}(\tau^{(m, \gtcnt)}_{n}) 
		- x^{L(e)}_{\greedk^{(\ell)}}(\tau^{(m)}_{n})\right) \\
		&- L(\lambda_j - \mu_j y_j(t))(\tau^{(m, \gtcnt)}_{n} - \tau^{(m)}_{n})
		+ o(L\subL),
		\end{aligned}
		\end{equation}
		from which it follows 
		\begin{equation}\label{eq:final2c}
		\begin{aligned}
		& \sum_{\ell=1}^{\gtcnt} \greedki^{(\ell)}_j
		\left(x^{L(e)}_{\greedk^{(\gtcnt)}}(\tau^{(m, \gtcnt)}_{n})  - 
		x^{L(e)}_{\greedk^{(\gtcnt)}}(\tau^{(m)}_{n})\right) \ge \\
		& (\lambda_j - \mu_j y_j(t))
		(\tau^{(m, \gtcnt)}_{n} - \tau^{(m)}_{n})
		+ o(\subL),
		\end{aligned}
		\end{equation}
		which implies (\ref{eq:res3eqA}).

	%%%%%%%%%%%%%%%%%%%%%%%%%%%%%%%%%%%%%
	\textbf{Proof of (\ref{eq:res3eqB}):}
	We will now prove (\ref{eq:res3eqB}) by considering two cases
	depending on length of
	$\tau^{(m+1)}_{n} - \tau^{(m,\gtcnt)}_{n}$ and 
	reach a contradiction for each of them.
	
	We consider
	\begin{equation}\label{eq:taucond3_3}
	\tau^{(m+1)}_{n} - \tau^{(m,\gtcnt)}_{n} \le 
	\frac{\epsilon \subL}{4(1+\Kmax)^{\gtcnt-1}R}
	\end{equation}
	or
	\begin{equation}\label{eq:taucond3_4}
	\tau^{(m+1)}_{n} - \tau^{(m,\gtcnt)}_{n} > 
	\frac{\epsilon \subL}{4(1+\Kmax)^{\gtcnt-1}R}
	\end{equation}
	where $R := \sum_{j \in \mathcal{J}} \lambda_j + \Kmax\mumax$.
	
	\textbf{Case (\ref{eq:taucond3_3}):}
		We notice ${x}^{(e)}_{\greedk^{(\ell)}}(\tau)$ 
		will change by at most $(1+\Kmax)^{\ell-1}/L$ at each arrival
		or departure according to Lemma~\ref{lem:X_change} and thus 
		(${x}^{(e)}_{\greedk^{(\gtcnt)}}(\tau) +
		\sum_{\ell=1}^{\gtcnt-1} {x}^{(e)}_{\greedk^{(\ell)}}(\tau)^+$)
		will change by at most $\frac{(1+\Kmax)^\gtcnt}{\Kmax L}$.
		
		Using Lemma~\ref{lem:Chern_Poisson}, 
		with probability $1 - o(L^{-2})$, the number of arrivals and departures 
		in interval $[\tau^{(m,\gtcnt)}_{n}, \tau^{(m+1)}_{n}]$
		of length at most 
		$\frac{\epsilon \subL}{4(1+\Kmax)^{\gtcnt-1} R}$, is at most
		\begin{equation*}
		\begin{aligned}
		&\left(\sum_{j \in \mathcal{J}} \lambda_j 
		+ \Kmax\mumax \right) L
		\frac{\epsilon \subL}{4(1+\Kmax)^{\gtcnt-1} R
		} + o(L\subL) = \\
		&\frac{\epsilon L\subL}{4(1+\Kmax)^{\gtcnt-1}} + o(L\subL),
		\end{aligned}
		\end{equation*} 
		therefore
		\begin{equation}\label{eq:contr3c}
		\begin{aligned}
		&\mathds{P}\left(
		\sum_{\ell=1}^\gtcnt \greedki^{(\ell)}_j
		\left({x}^{(e)}_{\greedk^{(\ell)}}(\tau^{(m+1)}_{n}) 
		- {x}^{(e)}_{\greedk^{(\ell)}}(\tau^{(m, \gtcnt)}_{n}) \right) \le
		\frac{\epsilon \subL}{4}\right) = \\
		&1 - o(L^{-2}).
		\end{aligned}
		\end{equation}
		
		Considering (\ref{eq:taucond3_3}) holds, we will also have
		\begin{equation}\label{eq:contr3d}
		\mathds{P}\left(
		\mumin\frac{1 - \sum_{i=1}^{\gtcnt} {x}^{(e)}_{\greedk^{(i)}}(t)}
		{J \Kmax \cfcnt^2}
		\left(\tau^{(m+1)}_{n} - \tau^{(m, \gtcnt)}_{n}\right) \le
		\frac{\epsilon \subL}{4}\right) = 1.
		\end{equation}
		It is now easy to verify that equations (\ref{eq:contr3c}) 
		and (\ref{eq:contr3d}) imply (\ref{eq:res3eqB}).

	\textbf{Case (\ref{eq:taucond3_4}):}
		In this case we notice that whenever a server
		without effective configuration in set
		$\mathcal{\bar K} := \{\greedk^{(\ell)}:\ell=1, \ldots, \gtcnt \}$
		empties, it will be assigned to one of the configurations
		of $\mathcal{\bar K}$.
				This statement is equivalent to the following, considering the bound 
		of Lemma~\ref{lem:emptyrate},
		\begin{equation*}
		\begin{aligned}
		& \left(x^{L(e)}_{\greedk^{(\gtcnt)}}(\tau^{(m+1)}_{n} ) 
		- x^{L(e)}_{\greedk^{(\gtcnt)}}(\tau^{(m,\gtcnt)}_{n})\right) \\
		&+ \sum_{\ell=1}^{\gtcnt-1}
		\left(x^{L(e)}_{\greedk^{(\ell)}}(\tau^{(m+1)}_{n} ) 
		- x^{L(e)}_{\greedk^{(\ell)}}(\tau^{(m,\gtcnt)}_{n})\right)^+ \\
		& - \frac{\mumin}{J \Kmax \cfcnt^2}
		\left(1 - \sum_{\ell=1}^{\gtcnt} x^{(e)}_{\greedk^{(\ell)}}(t) \right)
		(\tau^{(m+1)}_{n} - \tau^{(m,\gtcnt)}_{n}) > o(\subL).
		\end{aligned}
		\end{equation*}
		Since this holds based on Lemma~\ref{lem:emptyrate} with 
		probability $1-o(L^{-2})$, it implies (\ref{eq:res3eqB}).

%% file: lyapunov4.tex
\section{Details of Proof of Proposition~\ref{prop:VS_}}\label{prf:convergence}

	We notice that the last equation of the system (\ref{eq:sys}) 
	is the same as the previous ones, if $\Delta \eta_{i^\star}$
	has a coefficient $\greedki^{(i^\star)}_\perm{\ggi}:=1$, 
	$\Delta \eta_\maptg{j}$ has a coefficient	$\greedki^{(\maptg{\ell})}_\perm{\ggi}:=1$  
	for $\ell=1, \ldots, \ggi$ and $\Delta\theta_\perm{\ggi} = 0$.
	Thus, in what follows we analyze the system in its most general form
	where
	$$
	\Delta \eta_{i^\star} + \sum_{j=1}^{\ggi} \Delta \eta_\maptg{j} \le 0
	$$
	is replaced with
	$$
	\greedki^{(i^\star)}_\perm{\ggi} \Delta \eta_{i^\star} + 
	\sum_{\ell=1}^\ggi \greedki^{(\maptg{\ell})}_\perm{\ggi}
	\Delta \eta_\maptg{\ell} - \Delta\theta_\perm{\ggi} \le 0.
	$$
	As we showed in the main proof of Proposition~\ref{prop:VS_}, 
	the values of $\beta$ and $\beta_j$, $j=1, \ldots, \ggi$, 
	$\gamma_j$, $j=1, \ldots, J$, which we want to prove they are positive, 
	are given by the following system of equations
	\begin{equation}\label{eq:Zij}
	\begin{aligned}
	& \textstyle \LC_\maptg{\ell} = \sum_{j=\ell}^\ggi \beta_j \greedki^{(\maptg{\ell})}_\perm{j} \\
	& \textstyle \LC_{i^\star} = - \beta + \sum_{j=1}^\ggi \beta_j \greedki^{(i^\star)}_\perm{j} \\
	& \LC = \gamma_{j} + \beta_{j} \quad j \in \{\sigma_\ell: 
		\ell = 1, \ldots, \ggi \} \\
	& \LC = \gamma_j \quad j \in \{\sigma_\ell: 
		\ell = \ggi+1, \ldots, J \}.
	\end{aligned} 
	\end{equation}
	It is straightforward from (\ref{eq:Zij}) that 
	$\beta_\ggi = \frac{\LC_\perm{\ggi}}{\greedki^{(\maptg{\ggi})}_\perm{\ggi}} > 0$.
	We will now show $\greedki^{(\maptg{\ell})}_\perm{\ell} \beta_{\ell} 
	> \LC_\maptg{\ell}/2 > 0$ 
	for $\ell = 1, \ldots, \ggi-1$, when $\LC_i > (2\Kmax+1) \LC_{i+1}$,
	$i=1, \ldots \gtcnt-1$ based on assumptions.
	For shorthand purposes we also define $C := 2\Kmax+1$. 
	The proof is as follows
	\begin{equation}
	\begin{aligned}
		&\greedki^{(\maptg{\ell})}_\perm{\ell} \beta_{\ell} = \LC_\maptg{\ell}
		- \sum_{j=\ell+1}^\ggi \greedki^{(\maptg{\ell})}_\perm{j} \beta_{j}
		\stackrel{(a)}{\ge} \LC_\maptg{\ell} - \Kmax 
		\sum_{j=\ell+1}^\ggi \LC_\maptg{j} \stackrel{(b)}{>} \\
		&\LC_\maptg{\ell} -
		\Kmax \sum_{j=\ell+1}^\ggi \LC_\maptg{\ell} C^{\ell-j} >
		\LC_\maptg{\ell}\left(1 - \frac{\Kmax}{C-1}\right) = \LC_\maptg{\ell}/2.
	\end{aligned}
	\end{equation}
	In (a) we used (\ref{eq:Zij}), according to which, considering 
	$\beta_{\ell^\prime} > 0$ for $\ell^\prime = \ell+1, \ldots, \ggi$,
	we have $\greedki^{(\maptg{j})}_\perm{j} \beta_{j} <
	\LC_\maptg{j}$ or $\greedki^{(\maptg{\ell})}_\perm{j} \beta_{j} < 
	\frac{\greedki^{(\maptg{\ell})}_\perm{j}}
	{\greedki^{(\maptg{j})}_\perm{j}} 
	\LC_{\maptg{j}} \leq K \LC_{\maptg{j}}$ and
	in (b) we used that for $\ell < j$,
	$\LC_\maptg{j} < C^{\ell-j} \LC_{\maptg{j} + \ell-j} \le 
	C^{\ell-j} \LC_{\maptg{\ell}}$. 

	To prove $\beta > 0$, suppose $m$ is the lowest index for which
	$\greedki^{(i^\star)}_{\perm{m}} > 0$.
	Then we will have
	\begin{equation*}
		\beta = - \LC_{i^\star} 
		+ \sum_{j=1}^\ggi \beta_j \greedki^{(i^\star)}_\perm{j} >
		- \LC_{i^\star} + \beta_m \greedki^{(i^\star)}_\perm{m} 
		\stackrel{(a)}{>} - \LC_{i^\star} +
		\frac{\greedki^{(i^\star)}_{i_m}}{\greedki^{(\maptg{m})}_\perm{m}}
		\LC_{\perm{m}}/2 \stackrel{(b)}{>} 0.
	\end{equation*}
	Inequality (a) uses just that $\greedki^{(\maptg{\ell})}_\perm{\ell} \beta_{\ell} 
	> \LC_\perm{\ell}/2$ for $\ell = m$, which we have already proved.
	Inequality (b) follows considering that $\maptg{m} < i^\star$ and that
	$\LC_{i^\star} < {\LC_{i^\star+1}}/C < 
	\frac{\greedki^{(i^\star)}_\perm{m}}{\greedki^{(\maptg{m})}_\perm{m}}
	\LC_{\maptg{m}}/2$. 
	
	To show that $\maptg{m} < i^\star$ we notice that 
	$\greedk^{(i^\star)}$ and $\greedk^{(\maptg{m})}$ are two 
	different configurations whose job types belong to
	$\{\perm{m}, \ldots, \perm{J}\}$ and $\greedk^{(\maptg{m})}$ is the 
	configuration of maximum reward that has this property,
	so its index as given from Definition~\ref{def:greedy_conf}
	should be lower then $i^\star$.
	
	Lastly we need to show that $\gamma_j > 0$ for $j=1, \ldots, J$.
	If $j \in \{\ggi+1, \ldots, J\}$ then $\gamma_j = \LC > 0$.
	If $j \in \{1, \ldots, \ggi\}$ then by using that
	for $\ell=1, \ldots, \gtcnt$ we have $\LC > \LC_\ell$ because of 
	\dref{eq:Vcond} and 
	$\LC_\maptg{\ell} \ge \beta_{\ell}$ because of (\ref{eq:Zij}), we get
	\begin{equation}
	\gamma_{j} = \LC - \beta_{j} > \LC_\maptg{j} - \beta_{j} \ge 0.
	\end{equation}

%% file: average_subsubintervals2.tex
\section{Proof of Lemma~\ref{lem:Dx}}\label{proof:lem:Dx}
	Consider the function $\subL$ as in Definition~\ref{def:sublen}. 
	Recall that $\LC_i > \xi \LC_{i+1}$ for $i=1, \ldots, \gtcnt-1$ and 
	$\LC_{\gtcnt} > 0$. We choose $\xi$ such that:
	\begin{equation}\label{eq:xidef}
	\xi > \frac{\mumax}{\mumin} \left(12 \Kmax^2
	+ 16 \cfcnt \Kmax^2\frac{2(\mumax \Kmax + \sum_{j=1}^J \lambda_j)}{\delta}
	\left(12\Kmax\frac{\mumax}{\mumin}\right)^\gcnt
	\right),
	\end{equation}
	and $\LC$ is chosen such that $\LC > 4\LC_1.$
	We first show the following lemma.
	\begin{lemma}\label{lem:overm}
		For any $m \in \{1, \ldots, \maxm{n}\}$, with probability greater than $1 - o(L^{-2})$,
		\begin{equation}\label{eq:r_weightedC}
		\begin{aligned}
		& \sum_{i=1}^{\gtcnt} \LC_i
		\left(
		x^{L(e)}_{\greedk^{(i)}} ({\tau}^{(m+1)}_{n}) 
		- x^{L(e)}_{\greedk^{(i)}} ({\tau}^{(m)}_{n})
		\right) + \\
		&\sum_{j=1}^J \LC ({\tau}^{(m+1)}_{n} - {\tau}^{(m)}_{n}) 
		(\mu_j y_j(t) - \lambda_j)^+ \ge \\
		&\delta(\epsilon_V)({\tau}^{(m+1)}_{n} - {\tau}^{(m)}_{n}) 
		+ o({\subL}).
		\end{aligned}
		\end{equation}
	\end{lemma}
	\begin{proof}
		We will use \ref{pr:exact}, \ref{pr:bound1} and \ref{pr:bound2} to refer
		to the properties in Proposition~\ref{prop:properties} and whenever
		we apply such a property any resulting relation holds with probability
		$1 - o(L^{-2})$.
		
		For compactness we also define
		\begin{equation*}
		W^{(m)}(t) := \sum_{i=1}^{\gtcnt} \LC_i
		\nabla x^{L(e)}_{\greedk^{(i)}} [{\tau}^{(m)}_{n}, {\tau}^{(m+1)}_{n}]
		+ \sum_{j=1}^J \LC (\mu_j y_j(t) - \lambda_j)^+.
		\end{equation*}
		Our objective is thus to find $\delta(\epsilon_V)$ such that
		\begin{equation}
		W^{(m)}(t) \ge \delta(\epsilon_V) + 
		\frac{o({\subL})}{{\tau}^{(m+1)}_{n} - {\tau}^{(m)}_{n}},
		\end{equation}
		with probability $1 - o(L^{-2})$.
		
		We will analyze two separate cases depending on whether 
		$\ell_m < \gtcnt$ or $\ell_m = \gtcnt$.
		
		\textbf{Case $\ell_m < \gtcnt$:} 
		Considering Lemma~\ref{lem:rbound} and Property~\ref{pr:bound1},
		we first can get the following bound
		\begin{equation}\label{eq:bound3}
		\begin{aligned}
		& \sum_{i=\ell_m}^{\gtcnt} \LC_i \left(		
		x^{L(e)}_{\greedk^{(i)}} ({\tau}^{(m+1)}_{n}) 
		- x^{L(e)}_{\greedk^{(i)}} ({\tau}^{(m)}_{n}) \right)
		> \\
		& \left(\LC_{\ell_m} \alpha_{\ell_m} - \sum_{i=\ell_m+1}^{\gtcnt} \LC_i B_i\right)
		\left({\tau}^{(m+1)}_{n} - {\tau}^{(m)}_{n}\right)
		+ o(\subL) \stackrel{(a)}{>} \\
		& \LC_{\ell_m} \alpha_{\ell_m}/2({\tau}^{(m+1)}_{n} - {\tau}^{(m)}_{n}) + {o(\subL)},
		\end{aligned}
		\end{equation}
		where (a) follows from definitions of $\xi$, for which 
		$\LC_i > \xi \LC_{i+1}$, and $\alpha_{\ell_m}$ given in
		(\ref{eq:xidef}) and (\ref{eq:alphadef}) respectively.
		If we further show that
		\begin{equation}\label{eq:inter1}
			\sum_{i=1}^{\ell_m-1} 
			\LC_i \nabla x^{(e)}_{\greedk^{(i)}}
			[\tau^{(m)}_{n}, \tau^{(m+1)}_{n}]
			+ \sum_{j=1}^J \LC (\mu_j y_j(t) - \lambda_j)^+ \ge 
			\frac{o(\subL)}{\tau^{(m+1)}_n - \tau^{(m)}_{n}},
		\end{equation}
		then it follows from \dref{eq:inter1} and \dref{eq:bound3} that 
		\begin{equation}\label{eq:delta1}
		\begin{aligned}
		W^{(m)}(t) \ge \LC_{\ell_m}\alpha_{\ell_m}/2 + \frac{o(\subL)}
		{\tau^{(m+1)}_{n} - \tau^{(m)}_{n}}.
		\end{aligned}
		\end{equation}

		One way to show (\ref{eq:inter1}) is to find constants $\LC_i^\star$ for 
		$i=1, \ldots,\ell_m-1$ such that
		\begin{equation}\label{eq:Zp1}
		\begin{aligned}
		&\sum_{i=1}^{\ell_m-1} 
		\LC_i \nabla x^{L(e)}_{\greedk^{(i)}}
		[\tau^{(m)}_{n}, \tau^{(m+1)}_{n}]
		+ \sum_{j=1}^J \LC \mu_j(y_j(t) - \rho_j)^+ \ge \\
		&\sum_{i=1}^{\ell_m-1} 
		\LC_i^\star
		(\lambda_{\ji{i}} - \mu_{\ji{i}} y_{\ji{i}}(t))
		+ \frac{o(\subL)}{\tau^{(m+1)}_n - \tau^{(m)}_{n}}
		\end{aligned}
		\end{equation}
		and 
		\begin{equation}\label{eq:Zp2}
		\LC_i^\star
		(\lambda_{\ji{i}} - \mu_{\ji{i}} y_{\ji{i}}(t)) \ge 0
		\quad i=1, \ldots, \ell_m-1.
		\end{equation}
		
		A choice of constants $\LC_i^\star$ that satisfies (\ref{eq:Zp1}) 
		and (\ref{eq:Zp2}) for $i=1,\ldots,\ell_m-1$ is
		\begin{equation*}
			\LC_i^\star = \LC_i^\prime - \dsone ( \lambda_{\ji{i}} -\mu_{\ji{i}} y_{\ji{i}}(t) < 0) \LC 2^{-i}
		\end{equation*}	
		where constants $\LC^\prime_{i}$ are given by the 
		following system of equations 
		\begin{equation}\label{eq:Z_Zp}
			\sum_{\ell=i}^{\ell_m-1} 
			\greedk^{(\ell)}_{\ji{i}} \LC_{\ell}^\prime = \LC_i
			\quad i = 1,\ldots, \ell_m-1.
		\end{equation}

		We will now justify why this choice of $\LC_i^\star$ satisfies
		(\ref{eq:Zp1}) and (\ref{eq:Zp2}).
		The requirement (\ref{eq:Zp1}) can be inferred by adding 
		the next two relationships.

		The first relationship is 
		\begin{equation}\label{eq:Zp1a}
		\begin{aligned}
		&\sum_{i=1}^{\ell_m-1}
		\LC_i^\prime (\lambda_{\ji{i}} - \mu_{\ji{i}} y_{\ji{i}}(t)) \stackrel{(a)}{=} \\
		&\sum_{i=1}^{\ell_m-1}
		\LC_i^\prime \sum_{\ell=1}^{i} \greedki^{(\ell)}_\ji{i}
		\nabla {x}^{L(e)}_{\greedk^{(\ell)}}[\tau^{(m)}_n, \tau^{(m+1)}_n]
		+ \frac{o(\subL)}{\tau^{(m+1)}_n - \tau^{(m)}_{n}} \stackrel{(b)}{=} \\
		&\sum_{i=1}^{\ell_m-1} \LC_i 
		\nabla {x}^{L(e)}_{\greedk^{(i)}}[\tau^{(m)}_n, \tau^{(m+1)}_n]
		+ \frac{o(\subL)}{\tau^{(m+1)}_n - \tau^{(m)}_{n}},
		\end{aligned}
		\end{equation}
		where (a) is due to \ref{pr:exact}  % (\ref{eq:exact}) 
		for $i=1, \ldots, \ell_m-1$ and (b) is due to (\ref{eq:Z_Zp}).
		The second relationship is
		\begin{equation}\label{eq:Zp1b}
		\begin{aligned}
		&- \sum_{i=1}^{\ell_m-1}
		\LC 2^{-i} ( \lambda_{\ji{i}}-\mu_{\ji{i}} y_{\ji{i}}(t) )
		\dsone(\mu_{\ji{i}} y_{\ji{i}}(t) > \lambda_{\ji{i}}) = \\
		&\sum_{j=1}^J \LC (\mu_{j} y_{j}(t) - \lambda_{j})^+
		\sum_{i=1}^{\ell_m-1}  2^{-i} \dsone(\ji{i}=j) \le 
		\sum_{j=1}^J \LC (\mu_j y_{j}(t) - \lambda_{j})^+.
		\end{aligned}
		\end{equation}
		
		Finally, we should prove that (\ref{eq:Zp2}) also holds. 
		If $\lambda_{\ji{i}} - \mu_{\ji{i}} y_{\ji{i}}(t) \ge 0$ it 
		suffices that $\LC^\star_{i} > 0$ or $\LC^\prime_{i} > 0$.
		For this, we will show recursively that 
		\begin{equation}\label{eq:Zrec}
		2 \LC_i > \LC^\prime_{i} > 0, \quad i=1,\ldots,\ell_m-1,
		\end{equation}
		For $i = \ell_m-1$, using \dref{eq:Z_Zp}, we have
		\begin{equation}
			2 \LC_{\ell_m-1} > \LC_{\ell_m-1} \frac{1}{\greedk^{(\ell_m-1)}_{\ji{\ell_m-1}}} 
			= \LC^\prime_{\ell_m-1} > 0,
		\end{equation}
		while for $i < \ell_m-1$,
		\begin{equation}
		\begin{aligned}
		&2 \LC_{i} > \LC_{i} \frac{1}{\greedk^{(i)}_{\ji{i}}} 
		- \sum_{\ell=i+1}^{\ell_m-1}  \frac{\greedk^{(\ell)}_{\ji{i}}}{\greedk^{(i)}_{\ji{i}}} \LC^\prime_\ell \\
		& \stackrel{(a)}{=} \LC^\prime_{i} > \frac{1}{\greedk^{(i)}_{\ji{i}}}
		\left( \LC_{i} - \sum_{\ell=i+1}^{\ell_m-1} \Kmax \LC^\prime_\ell \right) \\
		&> \frac{1}{\greedk^{(i)}_{\ji{i}}}
		\left(\LC_{i} - \sum_{\ell=i+1}^{\ell_m-1} \Kmax 2 \LC_\ell\right) 
		\stackrel{(b)}{>} \frac{1}{\greedk^{(i)}_{\ji{i}}}\LC_{\ell_m-1} > 0,
		\end{aligned}
		\end{equation}
		where in (a) we used \dref{eq:Z_Zp}, and in (b) we used $\LC_\ell > (2\Kmax + 1) \LC_{\ell+1}$ for 
		any $\ell < \ell_m-1$. Notice that this claim is consistent with 
		assumption $\LC_\ell > \xi \LC_{\ell+1}$, since $\xi > 2\Kmax + 1$.
		
		If $\lambda_{\ji{i}} - \mu_{\ji{i}} y_{\ji{i}}(t) < 0$, it 
		suffices that $\LC^\star_{i} < 0$ or $\LC^\prime_{i} < 2^{-i} \LC$. We can get this result from (\ref{eq:Zrec}), which we proved
		earlier, as follows
		\begin{equation}\label{eq:inter2}
		\LC^\prime_{i} < 2 \LC_i \le 2 (2\Kmax + 1)^{-i+1} \LC_1 <
		2^{-i} \LC.
		\end{equation}
		where the last inequality is because $4 \LC_1 < \LC$.

		\textbf{Case $\ell_m = \gtcnt$:} For notation compactness, for 
		$j \in \mathcal J$ such that $\greedki^{(\gtcnt)}_j > 0$, we define:
		\begin{equation*}
		\begin{aligned}
		f_j(t)  := 
		\frac{1}{\greedki^{(\gtcnt)}_j} \Big(\lambda_j - \mu_j y_j(t) 
			- \sum_{i=1}^{\gtcnt-1} \greedki^{(i)}_j 
			h^{\subind[m],(i)}(t)\Big),
		\end{aligned}
		\end{equation*}
		\begin{equation*}
		f^\star(t)  := \frac{\mumin}{J \Kmax \cfcnt^2}
		\Big(1 - \sum_{i=1}^{\gtcnt}{x}^{(e)}_{\greedk^{(i)}}(t)\Big)
		- \sum_{i=1}^{\gtcnt-1} \left( h^{\subind[m],(i)}(t) \right)^+,
		\end{equation*} 
		\begin{equation*}
		j^\prime  := \argmin_{j \in \calJ: \greedki^{(\gtcnt)}_j > 0} f_j(t).
		\end{equation*}
		If $\gtcnt=\gcnt$ the set $\{j \in \calJ: \greedki^{(\gtcnt)}_j > 0 \}$ is empty in which case we consider $j^\prime=\emptyset$.
		We distinguish two sub-cases.
		
		%% Case 1 
%		\textbf{SubCase 1.}
		\textbf{Subcase $j^\prime=\emptyset$ or $f_{j^\prime}(t) \ge f^\star(t)$:}
				In this case it suffices to find constants $\LC_{i}^\star$ for
		$i=1, \ldots, \ell_m$, such that
		\begin{equation}\label{eq:DxI1}
		\begin{aligned}
		&\sum_{i=1}^{\gtcnt} 
		\LC_i \nabla x^{(e)}_{\greedk^{(i)}}
		[\tau^{(m)}_{n}, \tau^{(m+1)}_{n}] 
		+ \sum_{j=1}^J \LC \mu_j (y_j(t) - \rho_j)^+
		\ge \\
		&\sum_{i=1}^{\gtcnt-1} \LC_{i}^\star
		(\lambda_{\ji{i}} - \mu_{\ji{i}} y_{\ji{i}}(t))
		+ \LC_\gtcnt^\star \frac{\mumin}{J \Kmax \cfcnt^2}
		\left(1-\sum_{i=1}^{\gtcnt} {x}^{(e)}_{\greedk^{(i)}}(t)\right) \\
		&+ \frac{o(\subL)}{\tau^{(m+1)}_n - \tau^{(m)}_{n}},
		\end{aligned}
		\end{equation}
		and if we further define
		\begin{equation}\label{eq:DxI1__}
		\begin{aligned}
		& W^a_\gtcnt(\statez(t)) := &\sum_{i=1}^{\gtcnt-1} \LC_{i}^\star
		(\lambda_{\ji{i}} - \mu_{\ji{i}} y_{\ji{i}}(t)) + \\
		&&\LC_{\gtcnt}^\star \frac{\mumin}{J \Kmax \cfcnt^2}
		(1-\sum_{i=1}^{\gtcnt} {x}^{(e)}_{\greedk^{(i)}}(t)).
		\end{aligned}
		\end{equation}
		then
		\begin{equation}\label{eq:Zp3}
		\LC_i^\star
		(\lambda_{\ji{i}} - \mu_{\ji{i}} y_{\ji{i}}(t)) \ge 0
		\quad i=1, \ldots, \gtcnt-1,
		\end{equation}
		\begin{equation}\label{eq:Zp3_}
		\LC_{\gtcnt}^\star \frac{\mumin}{J \Kmax \cfcnt^2}
		\left(1-\sum_{i=1}^{\gtcnt} {x}^{(e)}_{\greedk^{(i)}}(t)\right) \ge 0,
		\end{equation}
		and
		\begin{equation}\label{eq:Wa_cond}
		W^a_\gtcnt(\statez(t)) > 0 \Leftrightarrow \statez(t) \not \in \setOpt.
		\end{equation}
		A choice of $\LC_i^\star$ that satisfies those 
		requirements is
		\begin{equation}\label{eq:Zstar2gt}
		\LC_\gtcnt^\star = \LC_\gtcnt,
		\end{equation}
		and
		\begin{equation}\label{eq:Zstar2}
			\LC_i^\star = \LC_i^\prime - 
			\dsone ( \lambda_{\ji{i}} -\mu_{\ji{i}} y_{\ji{i}}(t) < 0) \LC 2^{-i}
		\end{equation}
		where constants $\LC^\prime_{i}$ are given by the 
		following system of equations 
		\begin{equation}\label{eq:Z_Zp2}
		\begin{aligned}
		&\dsone\left(h^{\subind[m], (i)}(t) > 0\right) 
		\LC_\gtcnt^\star 
		+ \sum_{\ell=i}^{\gtcnt-1} 
		\greedk^{(\ell)}_{\ji{i}} \LC_{\ell}^\prime = \LC_i
		\quad i = 1,\ldots, \gtcnt-1. \\
		\end{aligned}
		\end{equation} 

		We will now justify why this choice of $\LC_i^\star$ satisfies
		(\ref{eq:DxI1}) and (\ref{eq:Zp3}) and (\ref{eq:Zp3_}).
		To prove (\ref{eq:DxI1}) we can add three relationships (\ref{eq:Zp3a}), (\ref{eq:Zp3b}) and
		(\ref{eq:Zp3c}) that we prove below. The first relationship is 
		\begin{equation}\label{eq:Zp3a}
		\begin{aligned}
		&\sum_{i=1}^{\gtcnt-1}
		\LC_i^\prime (\lambda_{\ji{i}} - \mu_{\ji{i}} y_{\ji{i}}(t)) \stackrel{(a)}{=} \\
		&\sum_{i=1}^{\gtcnt-1}
		\LC_i^\prime \sum_{\ell=1}^{i} \greedki^{(\ell)}_\ji{i}
		\nabla {x}^{L(e)}_{\greedk^{(\ell)}}[\tau^{(m)}_n, \tau^{(m+1)}_n]
		+ \frac{o(\subL)}{\tau^{(m+1)}_n - \tau^{(m)}_{n}} = \\
		&\sum_{i=1}^{\gtcnt-1} \LC_i 
		\nabla {x}^{L(e)}_{\greedk^{(i)}}[\tau^{(m)}_n, \tau^{(m+1)}_n]
		+ \frac{o(\subL)}{\tau^{(m+1)}_n - \tau^{(m)}_{n}},
		\end{aligned}
		\end{equation}
		where (a) is due to \ref{pr:exact} %(\ref{eq:exact}) 
		for $i=1, \ldots, \gtcnt-1$. The second relationship is
		\begin{equation}\label{eq:Zp3b}
		\begin{aligned}
		&\LC_\gtcnt^\star \frac{\mumin}{J \Kmax \cfcnt^2}
		\left( 1 - \sum_{i=1}^{\gtcnt} {x}^{(e)}_{\greedk^{(i)}}(t) \right) 
		\stackrel{(a)}{\le} \\
		&\LC_\gtcnt^\star 
		\nabla {x}^{L(e)}_{\greedk^{(\gtcnt)}}[\tau^{(m)}_n, \tau^{(m+1)}_n] \\
		&+ \LC_\gtcnt^\star \sum_{i=1}^{\gtcnt-1}
		\left(\nabla {x}^{L(e)}_{\greedk^{(i)}}[\tau^{(m)}_n, \tau^{(m+1)}_n]\right)^+
		+ \frac{o(\subL)}{\tau^{(m+1)}_n - \tau^{(m)}_{n}} \stackrel{(b)}{=} \\
		&\LC_\gtcnt^\star 
		\nabla {x}^{L(e)}_{\greedk^{(\gtcnt)}}[\tau^{(m)}_n, \tau^{(m+1)}_n] \\
		&+ \LC_\gtcnt^\star \sum_{i=1}^{\gtcnt-1}
		\dsone(h^{\subind[m], (i)}(t) > 0) 
		\nabla {x}^{L(e)}_{\greedk^{(i)}}[\tau^{(m)}_n, \tau^{(m+1)}_n] \\
		&+ \frac{o(\subL)}{\tau^{(m+1)}_n - \tau^{(m)}_{n}},
		\end{aligned}
		\end{equation}
		where (a) is due to \ref{pr:bound2}  %(\ref{eq:bound2}) 
		and assumption $j^\prime = \emptyset$ or $f_{j^\prime}(t) \ge f(t)$,
		while (b) is due to Lemma~\ref{lem:hrec}.
		The third relationship is
		\begin{equation}\label{eq:Zp3c}
		\begin{aligned}
		&\sum_{i=1}^{\gtcnt-1}
		\LC 2^{-i} (\mu_{\ji{i}} y_{\ji{i}}(t) - \lambda_{\ji{i}})
		\dsone(y_{\ji{i}}(t) > \rho_{\ji{i}}) = \\
		&\sum_{j=1}^J \LC (\mu_{j} y_{j}(t) - \lambda_{j})^+
		\sum_{i=1}^{\gtcnt-1}  2^{-i} \dsone(\ji{i}=j) < 
		\sum_{j=1}^J \LC \mu_j (y_{j}(t) - \rho_{j})^+.
		\end{aligned}
		\end{equation}
		
		Then to prove (\ref{eq:Zp3}) it suffices to show, 
		just like in the case $\ell_m < \gtcnt$, that 
		for $i=1, \ldots, \gtcnt-1$,
		$\LC_i^\star > 0$ if $\lambda_{\ji{i}} - \mu_{\ji{i}} y_{\ji{i}}(t) \ge 0$,
		and $\LC_i^\star < 0$ otherwise, while for (\ref{eq:Zp3_}) to be true
		we need $\LC_\gtcnt^\star > 0$, since $
		(1-\sum_{i=1}^{\gtcnt} {x}^{(e)}_{\greedk^{(i)}}(t)) \ge 0$.
		
		If we define $\LC_\gtcnt^\prime := \LC_\gtcnt^\star$ then it suffices 
		to show recursively
		\begin{equation}\label{eq:Zrec2}
		2 \LC_i > \LC^\prime_{i} > 0 \quad i=1,\ldots,\gtcnt.
		\end{equation}
		The process is very similar to the case $\ell_m < \gtcnt$.
		We can now prove \dref{eq:Wa_cond} as follows.

		Considering 
		$$
		1-\sum_{i=1}^{\gtcnt} {x}^{(e)}_{\greedk^{(i)}}(t) =
		\sum_{i=1}^{\gtcnt} \left(
		\globlt{x}^{(g)}_{\greedk^{(i)}} -
		{x}^{(e)}_{\greedk^{(i)}}(t)\right)
		$$ 
		and for $i=1, \ldots, \gtcnt-1$, 
		$$
		\mu_{\ji{i}} y_{\ji{i}}(t) - \lambda_{\ji{i}} = \sum_{\ell=1}^i
		\greedki^{(\ell)}_\ji{i}
		(x^{(e)}_{\greedk^{(\ell)}}(t) - \globlt{x}^{(g)}_{\greedk^{(\ell)}}).
		$$
		then $W^a_\gtcnt(\statez(t)) = 0$, if and only if
		$\globlt{x}^{(g)}_{\greedk^{(i)}} = x^{(e)}_{\greedk^{(i)}}(t)$ for
		$i=1, \ldots, \gtcnt$ or equivalently if and only if 
		$\statez(t) \in \setOpt$.
		If we define the vectors
		$\mathbf{x}(t) := (x^{(e)}_{\greedk^{(i)}}(t))_{i=1, \ldots,\gtcnt}$,
		$\globlt{x} := (\globlt{x}^{(g)}_{\greedk^{(i)}})_{i=1, \ldots,\gtcnt}$, 
		$\Delta\mathbf{x}(t) := \mathbf{x}(t) - \globlt{x}$ and consider
		$\globlt{x}$ the zero vector in space $\setConv[\epsilon_\rho]$,
		then we can also verify $W^a_\gtcnt(\statez(t))$ satisfies the subadditive and
		absolutely scalable properties, i.e
		\begin{equation}
		\begin{aligned}
		&W^a_\gtcnt(\Delta\mathbf{x}_1(t)) +
		W^a_\gtcnt(\Delta\mathbf{x}_2(t))
		\le W^a_\gtcnt(\Delta\mathbf{x}_1(t) + \Delta\mathbf{x}_2(t)) \\
		&W^a_\gtcnt(\alpha\Delta\mathbf{x}(t)) = |\alpha|
		W^a_\gtcnt(\Delta\mathbf{x}(t)).
		\end{aligned}
		\end{equation}

		Thus $W^a_{\gtcnt}(\statez(t))$ has the properties of
		norm in space $\setConv[\epsilon_\rho]$ and since this space has finite dimensions 
		all of its norms are equivalent which means there is $c^a$ such that		
		\begin{equation}\label{eq:delta2a}
		\begin{aligned}
		&W_{\gtcnt}^a(\statez(t)) \ge c^a V(\statez(t)) > c^a \epsilon_V
		\end{aligned}
		\end{equation}

		% Case 2 %%%%%%%%%%%
		\textbf{Subcase $f_{j^\prime}(t) < f(t)$:}
		In this case it suffices to prove
		\begin{equation}\label{eq:DxI2}
		\begin{aligned}
		&\sum_{i=1}^{\gtcnt} 
		\LC_i \nabla x^{(e)}_{\greedk^{(i)}}
		[\tau^{(m)}_{n}, \tau^{(m+1)}_{n}] 
		+ \sum_{j=1}^J \LC \mu_j (y_j(t) - \rho_j)^+
		\ge \\
		&\sum_{i=1}^{\gtcnt} \LC_{i}^\star
		(\lambda_{\ji{i}} - \mu_{\ji{i}} y_{\ji{i}}(t))
		+ \frac{o(\subL)}{\tau^{(m+1)}_n - \tau^{(m)}_{n}},
		\end{aligned}
		\end{equation}
		and 
		\begin{equation}\label{eq:Zp4}
		\LC_i^\star
		(\lambda_{\ji{i}} - \mu_{\ji{i}} y_{\ji{i}}(t)) \ge 0
		\quad i=1, \ldots, \gtcnt.
		\end{equation}

		A choice of values of $\LC_i^\star$ that satisfies those 
		requirements for $i=1,\ldots,\gtcnt-1$ is
		\begin{equation}\label{eq:Zstar4}
		\begin{aligned}
		\LC_i^\star = \LC_i^\prime 
		- \dsone (\lambda_{\ji{i}} -\mu_{\ji{i}} y_{\ji{i}}(t) < 0) \LC 2^{-i}
		\end{aligned}
		\end{equation}
		and if we further define
		\begin{equation}\label{eq:DxI2__}
		\begin{aligned}
		&W^b_\gtcnt(\statez(t)) := \sum_{i=1}^{\gtcnt} \LC_{i}^\star
		(\lambda_{\ji{i}} - \mu_{\ji{i}} y_{\ji{i}}(t)) = 0.
		\end{aligned}
		\end{equation}
		then
		\begin{equation}\label{eq:Zstar4st}
		\LC_\gtcnt^\star = \LC_\gtcnt^\prime 
		- \dsone(\lambda_{j^\prime} - \mu_{j^\prime} y_{j^\prime}(t) < 0)\LC 2^{-\gtcnt}.
		\end{equation}
		and
		\begin{equation}\label{eq:Wb_cond}
		W^b_\gtcnt(\statez(t)) > 0 \Leftrightarrow \statez(t) \not \in \setOpt.
		\end{equation}
		
		Constants $\LC^\prime_{i}$ are given by the 
		following system of equations 
		\begin{equation}\label{eq:Z_Zp4}
		\begin{aligned}
		&\greedk^{(i)}_{j^\prime} \LC_{\gtcnt}^\prime + \sum_{\ell=i}^{\gtcnt-1} 
		\greedk^{(\ell)}_{\ji{i}} \LC_{\ell}^\prime = \LC_i
		\quad i = 1,\ldots, \gtcnt-1, \\
		&\greedk^{(\gtcnt)}_{j^\prime} \LC_{\gtcnt}^\prime = \LC_\gtcnt.
		\end{aligned}
		\end{equation}
		
		We will now justify why this choice of $\LC_i^\star$ satisfies
		(\ref{eq:DxI2}) and (\ref{eq:Zp4}).
		To prove (\ref{eq:DxI2}) we can add relationships 
		(\ref{eq:Zp3a}) and (\ref{eq:Zp3c})
		from sub-case $f_{j^\prime}(t) \ge f(t)$ 
		and (\ref{eq:Zp4b}) proven next.
		\begin{equation}\label{eq:Zp4b}
		\begin{aligned}
		&\LC_\gtcnt^\prime (\lambda_{j^\prime} - \mu_{j^\prime} y_{j^\prime}(t)) \stackrel{(a)}{\le} \\
		&\LC_\gtcnt^\prime \sum_{i=1}^{\gtcnt-1}
		\greedki^{(i)}_{j^\prime}
		\nabla {x}^{L(e)}_{\greedk^{(i)}}[\tau^{(m)}_n, \tau^{(m+1)}_n]
		+ \frac{o(\subL)}{\tau^{(m+1)}_n - \tau^{(m)}_{n}},
		\end{aligned}
		\end{equation}
		where (a) is due to \ref{pr:bound2}  % (\ref{eq:bound2}) 
		and assumption $f_{j^\prime}(t) < f(t)$.
		
		By using systems of equations (\ref{eq:Zstar4}) and (\ref{eq:Zstar4st}),
		we can show, similarly to the case $\ell_m < \gtcnt$, that 
		for $i=1, \ldots, \gtcnt$,
		$\LC_i^\star > 0$ if $\lambda_{\ji{i}} - \mu_{\ji{i}} y_{\ji{i}}(t) \ge 0$
		and $\LC_i^\star < 0$ otherwise, by proving recursively
		\begin{equation}\label{eq:Zrec3}
		2 \LC_i > \LC^\prime_{i} > 0 \quad i=1,\ldots,\gtcnt.
		\end{equation}

		Similarly with $W^a_\gtcnt(\statez(t))$, we can show 
		\dref{eq:Wb_cond} as $W^b_\gtcnt(\statez(t))$
		satisfies all the properties of a norm in space $\setConv[\epsilon_\rho]$.

		Also, since this space has finite dimensions all of its norms are equivalent 
		which means there is $c^b$ such that
		\begin{equation}\label{eq:delta2b}
		\begin{aligned}
		&W_{\gtcnt}^b(\statez(t)) \ge c^b V(\statez(t)) > c^b \epsilon_V.
		\end{aligned}
		\end{equation}

		%% Case 2 
%		\textbf{SubCase 2.}
		
		\textbf{Determining $\delta(\epsilon_V)$:} 
        Expressions (\ref{eq:delta1}), (\ref{eq:delta2a}) and (\ref{eq:delta2b})
		give a lower bound on $W^{(m)}(t)$
		in three different cases, thus for any $m\in \{1, \ldots, \maxm{n}\}$ 
		a value of $\delta(\epsilon_V)$ is the minimum of these expressions
		i.e.
		\begin{equation}
			\delta(\epsilon_V) := 
			\min\Big(\min_{i=1, \ldots, \gtcnt-1} \LC_i \alpha_i/2,
			c^a \epsilon_V, c^b \epsilon_V \Big).
		\end{equation}
		This completes the proof.
	\end{proof}

	Summing (\ref{eq:r_weightedC}) in Lemma~\ref{lem:overm}, 
    over all $m \in \{1, \ldots, \maxm{n}\}$, we get
	\begin{equation}\label{eq:r_weightedFinal}
	\begin{aligned}
	& \sum_{i=1}^{\gtcnt} \LC_i
	\left(
	X^{L(e)}_{\greedk^{(i)}} ({\tau}_{n+1}) 
	- X^{L(e)}_{\greedk^{(i)}} ({\tau}_{n})
	\right) \ge \\
	&\delta(\epsilon_V)({\tau}_{n+1} - {\tau}_{n}) - 
	\sum_{j=1}^J \LC ({\tau}_{n+1} - {\tau}_{n}) \mu_j(y_j(t) - \rho_j)^+
	+ o({\subL}).
	\end{aligned}
	\end{equation}
	Since each of (\ref{eq:r_weightedC}) holds with probability
	$1 - o(L^{-2})$ and \dref{eq:r_weightedFinal} is a finite sum
	of them, it will also be satisfied with probability	$1 - o(L^{-2})$.
	
	To show that (\ref{eq:r_weightedFinal}) implies (\ref{eq:rX}), it
	remains to prove that
	\begin{equation}
		\frac{\rm d}{{\rm d}t} (y_j(t) - \rho_j)^+ \le 
		-\mu_j(y_j(t) - \rho_j)^+.
	\end{equation}
	If $y_j(t) < \rho_j$, then $(y_j(t) - \rho_j)^+ = 0$, and obviously
	\begin{equation}
	\frac{\rm d}{{\rm d}t} (y_j(t) - \rho_j)^+ = 
	\frac{\rm d}{{\rm d}t} 0 = 0.
	\end{equation}
	If $y_j(t) > \rho_j$, then $(y_j(t) - \rho_j)^+ = y_j(t) - \rho_j$
	and hence
	\be\label{eq:dyGTrho}
	\frac{\rm d}{{\rm d}t} (y_j(t) - \rho_j)^+ &=& 
	\frac{\rm d}{{\rm d}t} y_j(t) \nonumber \\
	&\le& \lambda_j - y_j(t)\mu_j= -\mu_j(y_j(t) - \rho_j),
	\ee
	where in the last inequality we used the fact that rate that type-$j$ jobs at fluid limit are admitted cannot be more than $\lambda_j$ (not all type-$j$ jobs that arrive are admitted), and  existing type-$j$ jobs in the system depart at rate $\mu_j y_j(t)$ in fluid limit.

%% file: maintheorem_proof.tex
\section{Proof of Theorem~\ref{thm:main}}\label{sec:maintheorem_proof}
	Using Theorem~\ref{thm:convergence}, we first show that, as $L \to \infty$, the sequence of stationary random variables $\mathbf{x}^{L(e)}(\infty)$ converges in distribution to $\mathbf{\globlt{x}}^{(g)}$ (the unique global greedy assignment), i.e.,  
	$$\mathbf{x}^{L(e)}(\infty) \implies \mathbf{\globlt{x}}^{(g)}.$$
	
	By Theorem~\ref{thm:convergence}, given an $\epsilon_1 > 0$, we can choose $t_{\epsilon_1}$ large enough such that $\| \mathbf{x}^{(e)}(t) -
	\mathbf{\globlt{x}}^{(g)}\|\leq \epsilon_1$, for $t\geq t_{\epsilon_1}$. Also by Theorem~\ref{prop:fl_limits}, we can choose a subsequence of $L_n$ of $L$ such that $x^{L_n(e)}_{\bk}(t) \to x^{(e)}_{\bk}(t)$ (u.o.c). Now for an $\epsilon_2>0$, and $L_n$ large enough, we can choose an $\epsilon_3$ such that uniformly over all initial states $\statez^{L_n}(0)$ we have $\|\statez^{L_n}(0)-\statez(0)\| \leq \epsilon_3$ and
	$$
	\mathbb{P}\Big(\|\mathbf{x}^{L_n(e)}(t_{\epsilon_1})- \mathbf{x}^{(e)}(t_{\epsilon_1})\| < \epsilon_1 \Big) > 
	1-\epsilon_2.
	$$
	This is true because otherwise for a sequence of initial states $\statez^{L_n}(0) \to \statez(0)$,
	$
	\mathbb{P}\Big(\|\mathbf{x}^{L_n(e)}(t_{\epsilon_1})- \mathbf{x}^{(e)}(t_{\epsilon_1}) \| < \epsilon_1 \Big) \leq 1-\epsilon_2,
	$
	which is impossible because, almost surely, $\mathbf{x}^{L_n(e)}(t_{\epsilon_1})\to \mathbf{x}^{(e)}(t_{\epsilon_1})$. Hence,
	\begin{flalign*}
	& \mathbb{P}\Big(\|\mathbf{x}^{L_n(e)}(t_{\epsilon_1})-\mathbf{\globlt{x}}^{(g)}\| <2 \epsilon_1\Big) > \\
	&\mathbb{P}\Big(\|\mathbf{x}^{L_n(e)}(t_{\epsilon_1})-\mathbf{x}^{(e)}(t_{\epsilon_1})\|+\|\mathbf{x}^{(e)}(t_{\epsilon_1})- \mathbf{\globlt{x}}^{(g)}\| < 2 \epsilon_1 \Big) \\
	& > \mathbb{P}\Big(\|\mathbf{x}^{L_n(e)}(t_{\epsilon_1})-\mathbf{x}^{(e)}(t_{\epsilon_1})\|< \epsilon_1\Big) \\
	& > 1-\epsilon_2
	\end{flalign*}
	which implies $x^{L_n(e)}_{\bk}(\infty) \implies \globlt{x}^{(g)}_{\bk} $, since  $\epsilon_1, \epsilon_2$ were chosen arbitrarily. Since this holds for every subsequence $L_n$ of $L$, and all converge to the same limit $\mathbf{\globlt{x}}^{(g)}$, we can conclude  $\mathbf{x}^{L(e)}(\infty) \implies \mathbf{\globlt{x}}^{(g)}$ (e.g., see Theorem 2.6 of \cite{billingsley2013convergence}).

	Next, denote the normalized reward of the system at time $t$ under {\DRA} as
	$U^L(t) = \frac{1}{L}F^{{DRA}}(t)$, and the normalized reward at time $t$ under the optimal policy as	${U^\star}^L(t)=\frac{1}{L}F^{{\star}}(t)$.
	Then,
	\begin{equation}
	U^L(\infty) \ge \sum_{i=1}^{\gcnt} \sum_{j=1}^J x^{L(e)}_{\greedk^{(i)}}(\infty) 
	\greedki^{(i)}_j - \sum_{j=1}^J \frac{g(L) + \Kmax}{L} u_j,
	\end{equation}
	since there are at most $g(L) + \Kmax$ empty slots for each job type
	in all the servers with effective configuration.
	Further note that random variables $\{U^L(\infty)\}$, and $\{x^{L(e)}_{\greedk^{(i)}}(\infty)\}$ are uniformly integrable, so they also converge in expectation (e.g., see Theorem 3.5 of \cite{billingsley2013convergence}), hence taking the expectations from both sides and letting $L \to \infty$, 
	\be\label{eq:help2}
	\lim_{L\to \infty}\mathbb{E}[U^L(\infty)] \ge 
	\sum_{i=1}^{\gcnt} \sum_{j=1}^J \globlt{x}^{(g)}_{\greedk^{(i)}}
	\greedki^{(i)}_j 
	= U^{(g)}[\bm{\rho}] \stackrel{(a)}{\ge} \frac{U^\star[\bm{\rho}]}{2}, 
	\ee
	where (a) is a direct consequence of Theorem~\ref{thm:optimality} (and $\frac{1}{2}$ can be replaced with $(1-1/e)$ if Theorem~\ref{thm:optimality2} is used).
    Also note that for the optimal algorithm
	$
	\mathbb{E}[{U^\star}^L(\infty)] \leq U^\star[\bm{\rho}],
	$
	since $U^\star[\bm{\rho}]$ is the maximum possible normalized reward. 
    Using this, and dividing both sides of \dref{eq:help2} by 	$\mathbb{E}[{U^\star}^L(\infty)]$, completes the proof.